\PassOptionsToPackage{unicode}{hyperref}
\PassOptionsToPackage{hyphens}{url}
\PassOptionsToPackage{dvipsnames,svgnames,x11names}{xcolor}
\documentclass[
  12pt]{article}  
\usepackage{adjustbox}  % for max width/height resize with aspect ratio

\usepackage{amsmath,amssymb}
\usepackage{iftex}
\ifPDFTeX
  \usepackage[T1]{fontenc}
  \usepackage[utf8]{inputenc}
  \usepackage{textcomp} % provide euro and other symbols
\else % if luatex or xetex
  \usepackage{unicode-math}
  \defaultfontfeatures{Scale=MatchLowercase}
  \defaultfontfeatures[\rmfamily]{Ligatures=TeX,Scale=1}
\fi
\usepackage{lmodern}
\ifPDFTeX\else  
\fi
\IfFileExists{upquote.sty}{\usepackage{upquote}}{}
\IfFileExists{microtype.sty}{% use microtype if available
  \usepackage[]{microtype}
  \UseMicrotypeSet[protrusion]{basicmath} % disable protrusion for tt fonts
}{}
\makeatletter
\@ifundefined{KOMAClassName}{% if non-KOMA class
  \IfFileExists{parskip.sty}{%
    \usepackage{parskip}
  }{% else
    \setlength{\parindent}{0pt}
    \setlength{\parskip}{6pt plus 2pt minus 1pt}}
}{% if KOMA class
  \KOMAoptions{parskip=half}}
\makeatother
\usepackage{xcolor}
\makeatletter
\ifx\paragraph\undefined\else
  \let\oldparagraph\paragraph
  \renewcommand{\paragraph}{
    \@ifstar
      \xxxParagraphStar
      \xxxParagraphNoStar
  }
  \newcommand{\xxxParagraphStar}[1]{\oldparagraph*{#1}\mbox{}}
  \newcommand{\xxxParagraphNoStar}[1]{\oldparagraph{#1}\mbox{}}
\fi
\ifx\subparagraph\undefined\else
  \let\oldsubparagraph\subparagraph
  \renewcommand{\subparagraph}{
    \@ifstar
      \xxxSubParagraphStar
      \xxxSubParagraphNoStar
  }
  \newcommand{\xxxSubParagraphStar}[1]{\oldsubparagraph*{#1}\mbox{}}
  \newcommand{\xxxSubParagraphNoStar}[1]{\oldsubparagraph{#1}\mbox{}}
\fi
\makeatother

\usepackage{longtable,booktabs,array}
\usepackage{calc} % for calculating minipage widths
\usepackage{etoolbox}

\makeatletter
\patchcmd\longtable{\par}{\if@noskipsec\mbox{}\fi\par}{}{}
\makeatother
\IfFileExists{footnotehyper.sty}{\usepackage{footnotehyper}}{\usepackage{footnote}}
\makesavenoteenv{longtable}
\usepackage{graphicx}
\makeatletter
\def\maxwidth{\ifdim\Gin@nat@width>\linewidth\linewidth\else\Gin@nat@width\fi}
\def\maxheight{\ifdim\Gin@nat@height>\textheight\textheight\else\Gin@nat@height\fi}
\makeatother
\setkeys{Gin}{width=\maxwidth,height=\maxheight,keepaspectratio}
\makeatletter
\def\fps@figure{htbp}
\makeatother

\makeatletter
\@ifpackageloaded{caption}{}{\usepackage{caption}}
\AtBeginDocument{%
\ifdefined\contentsname
  \renewcommand*\contentsname{Table of contents}
\else
  \newcommand\contentsname{Table of contents}
\fi
\ifdefined\listfigurename
  \renewcommand*\listfigurename{List of Figures}
\else
  \newcommand\listfigurename{List of Figures}
\fi
\ifdefined\listtablename
  \renewcommand*\listtablename{List of Tables}
\else
  \newcommand\listtablename{List of Tables}
\fi
\ifdefined\figurename
  \renewcommand*\figurename{Figure}
\else
  \newcommand\figurename{Figure}
\fi
\ifdefined\tablename
  \renewcommand*\tablename{Table}
\else
  \newcommand\tablename{Table}
\fi
}
\@ifpackageloaded{float}{}{\usepackage{float}}
\floatstyle{ruled}
\@ifundefined{c@chapter}{\newfloat{codelisting}{h}{lop}}{\newfloat{codelisting}{h}{lop}[chapter]}
\floatname{codelisting}{Listing}

\makeatother
\makeatletter
\@ifpackageloaded{caption}{}{\usepackage{caption}}
\@ifpackageloaded{subcaption}{}{\usepackage{subcaption}}
\makeatother

\ifLuaTeX
  \usepackage{selnolig}  % disable illegal ligatures
\fi
\usepackage[]{natbib}
\usepackage{bookmark}

\IfFileExists{xurl.sty}{\usepackage{xurl}}{} % add URL line breaks if available
\hypersetup{
  pdftitle={Title},
  pdfauthor={Author 1; Author 2},
  pdfkeywords={3 to 6 keywords, that do not appear in the title},
  colorlinks=true,
  linkcolor={blue},
  filecolor={Maroon},
  citecolor={Blue},
  urlcolor={Blue},
  pdfcreator={LaTeX via pandoc}}

\usepackage{babel}
\usepackage{float}
\usepackage{booktabs}
\usepackage{fancybox}
\usepackage{calc}
\usepackage{mathtools}
\usepackage{url}
\usepackage{algorithm2e}
\SetAlFnt{\small}
\usepackage{algorithmic}
\algsetup{linenosize=\tiny}
\usepackage{amsmath}
\usepackage{amsthm}
\usepackage{amssymb}
\usepackage{graphicx}

\usepackage{bm}
\usepackage{amsfonts}

\usepackage{afterpage}

\newcommand{\R}{\mathbb{R}}

\newcommand{\ind}{\overset{\mathrm{ind}}{\sim}}

\newcommand{\x}{\mathbf{x}}

\theoremstyle{plain}
\newtheorem{theorem}{Theorem}[section]
\newtheorem{lemma}[theorem]{Lemma}\newtheorem{corollary}[theorem]{Corollary}

\theoremstyle{definition}

\theoremstyle{remark}

\newcommand{\anon}{1}

\begin{document}

\def\spacingset#1{\renewcommand{\baselinestretch}%
{#1}\small\normalsize} \spacingset{1}

%%%%%%%%%%%%%%%%%%%%%%%%%%%%%%%%%%%%%%%%%%%%%%%%%%%%%%%%%%%%%%%%%%%%%%%%%%%%%%

\if1\anon
{
  \title{\bf Wavelet Tree Ensembles for Triangulable Manifolds}
  \author{Hengrui Luo
    % \thanks{LM was supported by NSF grant DMS-2152999. HL was supported by U.S. Department of Energy under Contract DE-AC02-05CH11231 and NSF-DMS 2412403.}
    \hspace{.2cm}\\
    Department of Statistics, Rice University;\\
    Computational Research Division, Lawrence Berkeley National Laboratory\\
    and \\
    Akira Horiguchi \\
    Department of Statistics, University of California Davis \\
    and \\
    Li Ma \\
    Department of Statistics \& Data Science Institute, University of Chicago}
    \date{}
  \maketitle
} \fi

\if0\anon
{
  \bigskip
  \bigskip
  \bigskip
  \begin{center}
    {\LARGE\bf Wavelet Tree Ensembles for Triangulable Manifolds}
\end{center}
  \medskip
} \fi

\bigskip
\begin{abstract}
% The text of your abstract. 200 or fewer words.
We develop unbalanced Haar (UH) wavelet tree ensembles for regression on triangulable manifolds.
Given data sampled on a triangulated manifold, we construct UH wavelet trees whose atoms are supported on geodesic triangles and form an orthonormal system in $L^2(\mu_n)$, where $\mu_n$ is the empirical measure on the sample, which allows us to use UH trees as weak learners in additive ensembles.
Our construction extends classical UH wavelet trees from regular Euclidean grids to generic triangulable manifolds while preserving three key properties: (i) orthogonality and exact reconstruction at the sampled locations, (ii) recursive, data-driven partitions adapted to the geometry of the manifold via geodesic triangulations, and (iii) compatibility with optimization-based and Bayesian ensemble building. 
In Euclidean settings, the framework reduces to standard UH wavelet tree regression and provides a baseline for comparison.
We illustrate the method on synthetic regression on the sphere and on climate anomaly fields on a spherical mesh, where UH ensembles on triangulated manifolds substantially outperform classical tree ensembles and non-adaptive mesh-based wavelets. 
For completeness, we also report results on image denoising on regular grids.
A Bayesian variant (RUHWT) provides posterior uncertainty quantification for function estimates on manifolds.
Our implementation is available at 
\if1\anon
{
\url{http://www.github.com/hrluo/WaveletTrees}.
} \fi
\if0\anon
{
[URL hidden for blind version.]
} \fi
\end{abstract}

\noindent%
\vfill

\newpage
\spacingset{1} 

\section{Introduction}

\begin{table}[h!]
\centering
\resizebox{\linewidth}{!}{%
\begin{tabular}{lllclc}
\toprule
Method / family & Key reference & Domain & Data-driven & Base & Ensemble\\
 &  &  & partition? & measure & diversity\\
\midrule
Dyadic wavelets & \citet{daubechies1992ten} & Images & No & Lebesgue & None\\
 & \citet{donoho1995wavelet} &  &  &  & \\
\midrule
SHAH & \citet{fryzlewicz2016shah} & Images & Yes{*} & Counting & None\\
\midrule
WARP & \citet{li2021learning} & Images & Yes{*} & Counting & Cyclic spinning\\
\midrule
Random Forest & \citet{breiman2001random} & Euclidean & Yes & Empirical & Stochastic splits\\
BART & \citet{chipman2010bart} &  &  &  & \\
\midrule
UHWT & \textbf{This work} & Triangulable & Yes & Empirical & None\\
\midrule
UHWT boost & \textbf{This work} & Triangulable & Yes & Empirical & Boosting\\
\midrule
RRE (with UHWT) & \citet{blaser2016random} & Triangulable & Yes & Empirical & Haar rotations\\
\midrule
RotF (with UHWT) & \citet{rodriguez2006rotation} & Surface & Yes & Empirical & PCA rotations\\
\midrule
RotF (with sPCA) & \citet{rodriguez2006rotation,hrluo_2021c} & Sphere & Yes & Empirical & Bootstrap+sPCA\\
\bottomrule
\end{tabular}%
}
\caption{Structural properties of classical, adaptive and ensemble wavelet
frameworks. {*} means that the data-driven part only applies to split
dimensions. ``Triangulable'' refers to ``triangulable manifold.''
``RotF'' is an abbreviation for ``Rotation Forest.''}
\end{table}

Wavelets are effective for sparsely representing functions with sharp changes and have been broadly applied in scenarios with a low signal-to-noise ratio. 
This work is motivated by the observation that the existing literature on wavelet regression has largely concerned the effective representation of functions using a set of wavelet bases defined along a single recursive partition tree. At the same time, tree-based regression methods such as Classification and Regression Trees (CART) \citep{breiman1984classification} have been generalized to ensembles, most notably through random forests \citep{breiman2001random} and boosting \citep{Friedman01}, to mitigate high variance and overfitting, which can be serious when the signal-to-noise ratio is low. Wavelet transforms and ensembles thus present two complementary approaches for reducing variance. A natural question is whether incorporating both approaches simultaneously in a single framework can lead to further improvements in the resulting estimation and inference. 

Following this motivation, we introduce an additive boosted ensemble model of weak wavelet-based tree learners.
Whereas standard boosted regression trees regularize only at the leaf-level of the partition trees, each of our weak learners is a strongly regularized estimator along a binary recursive partition tree over the function domain, with regularization enforced under a wavelet basis along the partition tree. Standard wavelet transforms, however, are defined on dyadic trees that partition at fixed midpoints of the domain dimensions, which limits the diversity of the collection of trees and hence leads to large bias or approximation error in the resulting additive ensemble. To overcome this limitation, we employ an ``unbalanced'' Haar wavelet transform defined on partition trees that can be divided into a rich set of possible locations. 
Ensemble ideas for wavelet-based generative models also appear in \citet{xie2016inducing}, who introduce wavelets into a boosting framework and incrementally build complex models for wavelet coefficients.

We carry out a theoretical analysis to understand the difference in asymptotic behaviors of standard regression trees such as CART, which we refer to as leafwise estimators, and wavelet-based regression trees with unbalanced bases, which estimate coefficients and impose regularization on interior nodes of the partition tree. We derive oracle bounds showing that the stochastic term for the UH estimator scales with the number of active
coefficients, while the corresponding term for the leafwise estimator
scales with the total number of leaves. This clarifies when reconstruction
based on UH coefficients can substantially reduce variance,
particularly for signals that are sparse in the corresponding dictionary.

In addition, we construct a new regression framework based on orthogonal unbalanced Haar (UH) wavelet trees (UHWT) whose main target is regression on triangulable manifolds. The framework further connects classical wavelet and tree methodology and extends additive ensembles to non-Euclidean domains, while reducing to standard UHWT regression on regular Euclidean grids as a special case.
To this end, we extend UHWT from two-dimensional
images to higher-order tensors and triangulable manifold domains.
On the latter, we work with geodesic triangulations and data-adaptive recursive splits of triangles, which yield an orthonormal unbalanced
Haar system in $L^{2}(\mu_{n})$ on curved domains. This construction
permits a data-driven wavelet regression ensemble for generic non-Euclidean spaces.
We use UH trees as weak learners in additive ensembles
and employ random rigid motions, for example, rotations on the sphere 
or intrinsic reparameterizations, to generate diverse
trees while preserving the geometry of the domain. 
This unbalanced construction has concrete advantages
over more rigid multiresolution schemes, such as balanced Haar bases on fixed
spherical grids or fixed-depth trees. Balanced constructions split every cell
according to a predetermined schedule, regardless of the data, and they typically rely on global
coordinate systems that distort geometry on curved surfaces. In contrast, UH
trees choose split locations along edges using the empirical masses and local
variation of the data, so partitions can be simultaneously highly refined near sharp
transitions on the manifold and coarse in smooth or low-data regions. This
adaptivity is particularly important when triangle shapes and sampling
densities vary across the surface, as in our spherical examples.

While we mainly take an optimization approach to risk minimization in designing our algorithms for fitting our estimators, for completeness, we also present a probabilistic
formulation of the unbalanced wavelet-tree ensemble model, where a prior distribution over split dimensions and locations induces a
prior distribution over UH trees, and a wavelet likelihood factor
leads to a posterior distribution over tree structures with direct uncertainty quantification for shapes \citep{luo2024multiple}. 
The latter formulation leads to a Bayesian backfitting algorithm that parallels that for fitting 
Bayesian Additive Regression Trees (BART) \citep{chipman2010bart}. 

To our knowledge, this is the first framework that combines orthogonal
wavelets, recursive partitions that are aware of manifold geometry,
and modern ensemble techniques into a unified regression
method on triangulable manifolds.
The remainder of the paper develops UHWTs on grids and triangulated manifolds (Sections~\ref{sec:Wavelet-Regression}–\ref{sec:Wavelet-Trees-general domain}), their optimization-based and Bayesian ensembles (Section~\ref{sec:Wavelet-Regression-ensemble}). Section~\ref{sec:Theoretical-Results}
investigates the theoretical properties of our methods. Section~\ref{sec:Experiments}
presents empirical evaluations using diverse real and synthetic datasets.

\section{Wavelet Tree Regression}
\label{sec:Wavelet-Regression}

Signal denoising is a fundamental task in signal processing and can be
formulated as a regression problem. For a domain
$\Omega$, we aim to recover a signal $f\colon\Omega\to\mathbb{R}$ from observations
\begin{align}
y_{i}=f(\bm{x}_{i})+\epsilon_{i},\quad i=1,\ldots,n\label{eq:observation}
\end{align}
at locations $\bm{x}_{i}\in\Omega$ with zero mean additive noise
$\epsilon_{i}$ that is stable under orthogonal wavelet transforms.
Classical wavelet denoising proceeds by projecting the noisy signal
onto a wavelet system \citep{mallat1999wavelet}, shrinking the resulting coefficients, and
reconstructing a smoothed version of the signal
\citep{donoho1995wavelet}. In this section, we construct a wavelet
representation along a binary tree that links directly to regression
trees and lays the groundwork for both the optimization and Bayesian
approaches developed later. These two approaches share the same
 Haar representation and the same notion of wavelet
coefficients attached to nodes of a tree, but they differ in how
splits are chosen and how regularization is imposed.
The remainder of Section~\ref{sec:Wavelet-Regression} will assume that the domain $\Omega\subset [0,1]^D$ is a known, regular grid 
$\{1/n_{1},\ldots,1\}\times\{1/n_{2},\ldots,1\}\times\cdots\times\{1/n_{D},\ldots,1\}$
for some $D \in \mathbb{N}$, where usually $D=2$ (e.g., an image) or $D=3$ (e.g., a video or a spatial temporal field).
Section~\ref{sec:Wavelet-Trees-general domain} will then introduce our methodology for when $\Omega$ is a triangulable manifold (e.g., a sphere). 
For simplicity, the entire paper will assume that the locations $\bm{x}_{1},\ldots,\bm{x}_{n}$ are unique.

\subsection{Unbalanced Haar Wavelet Tree}
\label{sec:Haar-Wavelets}\label{sec:Wavelet-Regression-via}

The choice of wavelet functions determines how well the procedure
captures the structure of the data. To realize the potential of ensembling over wavelet regression as base learners, we need to enable enough diversity among the partition trees. As such, we must depart from standard wavelets that always partition in a balanced fashion at the midpoints of the domain dimensions, and consider unbalanced domain partitions. While in general this strategy can be incorporated for any wavelet basis, the Haar basis enjoys computational simplicity \citep{li2021learning}.
Thus we start by constructing unbalanced Haar (UH) functions through a 
recursive binary partitioning of the grid domain $\Omega$.

Starting with $A = \Omega$, we select---either greedily or probabilistically---a ``valid'' dimension $d \in \{1,\ldots,D\}$
and then select a ``valid'' location $\ell \in (a,b)$, 
where $a$ and $b$ are the
endpoints of the marginal interval of $A$ along dimension $d$. 
We call a dimension/location \textit{valid} if the node can be bisected along that dimension/location. 
The set $A$ is then bisected along dimension $d$ at location $\ell$ 
into two disjoint pieces, $A_{l}$ and $A_{r}$.

For this split $(A,d,\ell)$, we define the UH wavelet function by
\begin{align}
\psi_{A,d,\ell}(\bm{x}) =
\sqrt{\frac{|A_{l}|\;|A_{r}|}{|A|}}\left(
\frac{\mathbf{1}\{\bm{x}\in A_{l}\}}{|A_{l}|} -
\frac{\mathbf{1}\{\bm{x}\in A_{r}\}}{|A_{r}|}\right),\label{eq:UHbasis_function}
\end{align}
where $|\cdot|$ denotes the $\mu$-measure of a set for a fixed base measure $\mu$ on $\Omega$ (e.g., Lebesgue measure in the classical setting or the empirical measure $\mu_n$ in our finite-sample theory).

In classical wavelet regression \citep{mallat1989theory}, one chooses a {\em pre-defined} family of wavelets. In a data-driven regression tree, the same splitting construction is applied recursively to each piece, $A_{l}$ and $A_{r}$,
that can be split. In this way, a recursive partition
and a family of UH functions \eqref{eq:UHbasis_function} are built in parallel.
The collection of functions
generated in this way forms a dictionary of piecewise-constant
functions supported on sets of varying shapes and locations
\citep{friedman1974projection,donoho1997cart}. This
dictionary is not necessarily a complete basis of $L^{2}([0,1]^{D})$, but it is
sufficiently rich for regression on the observed design.

Given such a recursive binary partition, each split $(A,d,\ell)$ can be associated with an internal node of a binary tree rooted at $\Omega$; 
we will call it an \emph{Unbalanced Haar Wavelet Tree} (UHWT).
Each split is also associated with a wavelet coefficient defined as the $L^{2}(\mu)$ inner product
\begin{align}
w_{d,\ell}(A)\coloneqq
\langle\psi_{A,d,\ell},y\rangle_{L^{2}(\mu)}
&=\sqrt{\frac{|A_{l}^{(d,\ell)}|\;|A_{r}^{(d,\ell)}|}{|A|}}
\left(
  \frac{\int_{A_{l}^{(d,\ell)}}y(\bm{x})\,\mathrm d\mu(\bm{x})}{|A_{l}^{(d,\ell)}|}
  -
  \frac{\int_{A_{r}^{(d,\ell)}}y(\bm{x})\,\mathrm d\mu(\bm{x})}{|A_{r}^{(d,\ell)}|}
\right)\label{eq:wavelet_coef}
\end{align}
where $A_{l}^{(d,\ell)}$ and $A_{r}^{(d,\ell)}$ denote the left
and right children obtained by splitting $A$ at $(d,\ell)$.
When the domain $\Omega$ is a regular grid and we take $\mu$ to be the empirical
measure $\mu_n = n^{-1}\sum_{i=1}^n \delta_{\bm x_i}$, 
for any node $A$ we have
$\int_A y(\bm x)\,\mathrm d\mu_n(\bm x)
=
n_A^{-1}\sum_{\bm x_i\in A} y_i$ and $\mu_n(A)= n_A/n,$
where $n_A$ is the number of design points in $A$. Substituting into \eqref{eq:wavelet_coef}
implies 
$w_{d,\ell}(A) = \sqrt{\frac{n_l n_r}{n_A}}\;(\bar y_l - \bar y_r),$
where $\bar y_B = n_B^{-1}\sum_{\bm x_i\in B} y_i$ denotes the sample mean in node $B$.
The observed signal can then be reconstructed using an empirical inverse transform.
Denoising the signal, however, will require regularizing the coefficients, as
will be discussed in Section~\ref{sec:optimizationapproach}.

Related connections between Haar wavelets and regression trees include \cite{donoho1997cart} and \cite{fryzlewicz2007unbalanced}.
Our construction allows continuous split locations in multiple dimensions, and in the grid case the inverse UH transform coincides exactly with the leafwise CART estimator on the same tree (online Appendix~C; see also online Appendix~G.3).

Our construction allows continuous split
locations in multiple dimensions and, in the grid case, we show 
in online Appendix~C that
the associated regression fit coincides exactly with the leafwise
CART estimator on the same tree (see online Appendix~G.3 for a connection to \citet{fryzlewicz2007unbalanced}). %

As will be shown below, our UH construction is algebraically agnostic and can be extended to accommodate tensors;  
In contrast, the tensor-tree 
method of \citet{luo2024efficient} is explicitly tensor aware:
it evaluates splits using low-rank approximation errors, for example
based on CP or Tucker reconstructions, and fits low-rank models inside
leaves. 
Our UHWT tensor method instead emphasizes geometric
splits and Haar contrasts and is designed for signals with local
discontinuities or anisotropic edges.

\subsection{The optimization approach to fitting a UHWT}
\label{sec:optimizationapproach}

\begin{figure}[t]
\centering
\includegraphics[width=0.2\textwidth]{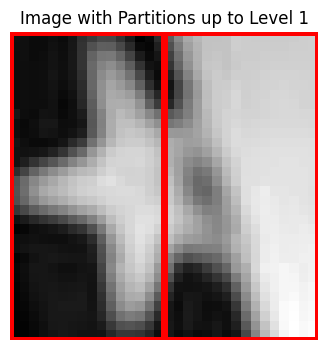}\includegraphics[width=0.2\textwidth]{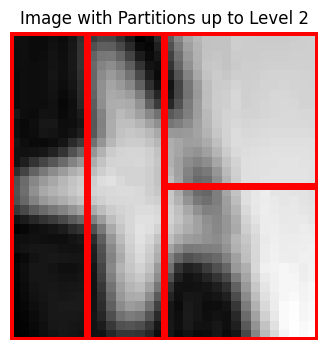}\includegraphics[width=0.2\textwidth]{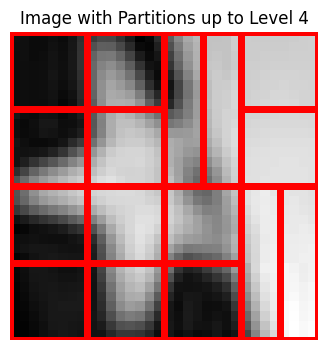}\includegraphics[width=0.2\textwidth]{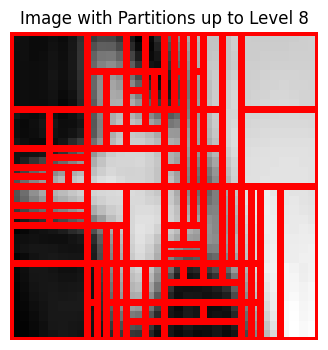}\includegraphics[width=0.2\textwidth]{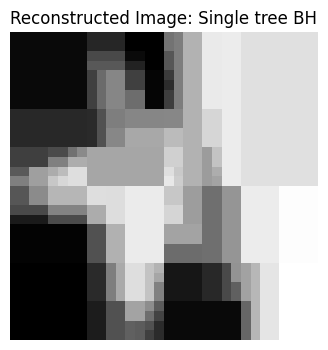}
~ \includegraphics[width=0.2\textwidth]{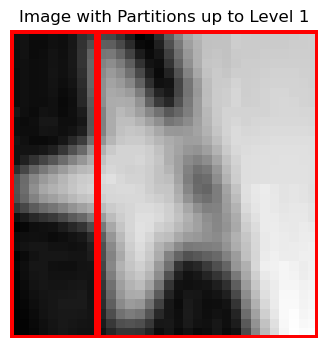}\includegraphics[width=0.2\textwidth]{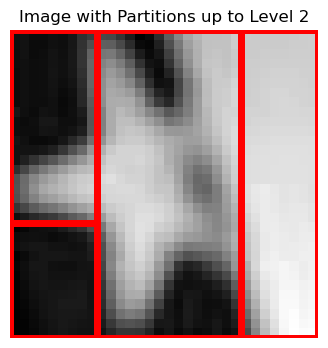}\includegraphics[width=0.2\textwidth]{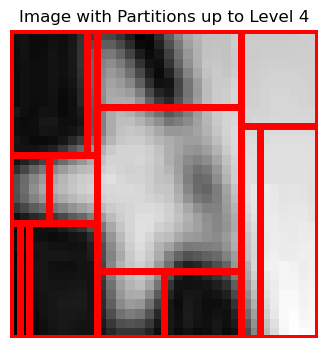}\includegraphics[width=0.2\textwidth]{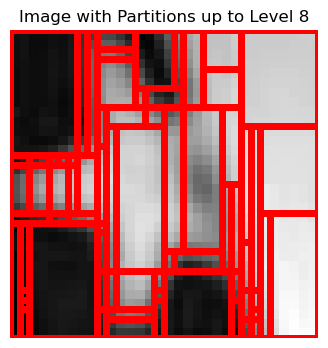}\includegraphics[width=0.2\textwidth]{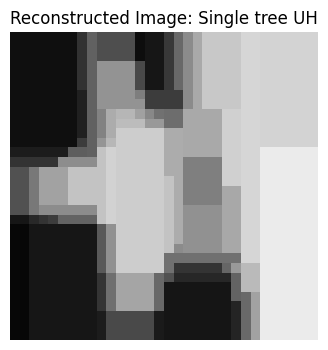}
\caption{Comparison between the Balanced Haar (BH) and UH optimization approach on a noisy
image that contains a star and a mostly two tone background. The first
row uses BH splits. The second row uses UH splits chosen by optimization;
these splits delineate the boundary of the star and adapt to the
volume of the main object. At depth 8 the partitions concentrate
along the edges of the star. The final reconstruction with a single
tree preserves both the boundary and the volume of the star while
reducing background interference.}
\label{fig:comparison_basis-1}
\end{figure}

In the optimization approach, each split in the UHWT is chosen greedily and deterministically from the data. 
(Online Appendix~C explores the connection of our method to CART in the grid case.)
At each node $A$, this approach chooses a split dimension and location pair $(d,\ell)$ 
that maximizes the absolute value of the coefficient $w_{d,\ell}(A)$ from \eqref{eq:wavelet_coef},
\begin{align}
(d^{*},\ell^{*})\in\arg\max_{(d,\ell)}\left|w_{d,\ell}(A)\right|,\label{eq:opt_approach}
\end{align}
subject to basic safeguards such as minimum sample sizes in the
children, as described in online Appendix~A,
and optional stopping rules. Because the UH scaling and wavelet functions
form an orthonormal system in $L^{2}(\mu_{n})$ (see Section~\ref{sec:Theoretical-Results} and online Appendix~B),
the immediate reduction in empirical squared error achieved by splitting at
$(d,\ell)$ is exactly $\left|w_{d,\ell}(A)\right|^{2}$. Thus the criterion
\eqref{eq:opt_approach} selects the split with the largest one-step
reduction in empirical risk on node $A$.

Each greedy split \eqref{eq:opt_approach} requires the computation of wavelet coefficients for many possible splits.
To reduce the amount of computation, we will use an early stopping mechanism:
if the wavelet coefficient for an internal node $A$ falls below the threshold $b\,\hat{\sigma}$ for $\hat{\sigma}$ defined in \eqref{eq:sigmahat} and some user-chosen constant $b \geq 0$, then $A$ will be set to a leaf node that is no longer allowed to split.
This is optional if we ultimately fit just a single tree, but it will prove to be computationally crucial when we construct ensembles of wavelet trees as in Section~\ref{sec:Wavelet-Regression-ensemble}.
Figure~2 in online Appendix~L.1.2 provides empirical results on an image dataset for using both the soft-thresholding and early-stopping mechanisms.

After the binary partition tree is generated, a denoised signal is obtained by an empirical inverse transform on
the learned binary-tree partition $\mathcal{P}_n$ using possibly regularized versions of the wavelet coefficients. 
The transform at any location $\bm{x}$ is described as follows.
Let $A_{0}\supset A_{1}\supset\cdots\supset A_{S(\bm{x})}$ be the
sequence of nodes along the path from the root $A_{0}$ to the leaf $A_{S(\bm{x})}$ containing
$\bm{x}$.
For each internal node $A_{i}$ on this path $1,2,\cdot,S(\bm{x})$, we store the coefficient $w_{d_{i},\ell_{i}}(A_{i})$. 
For some regularized versions of these coefficients, which we will denote as 
$\tilde{w}_{d_{i},\ell_{i}}(A_{i})$,
we define the denoised signal
\begin{align}
\hat{f}_{\mathcal{P}_{n}}^{\text{UH}}(\bm{x})
=
\frac{1}{n}\sum_{i=1}^{n}y_{i}
+
\sum_{i=0}^{S(\bm{x})-1}
u_{i}\frac{\tilde{w}_{d_{i},\ell(A_{i})}(A_{i})}{c(A_{i})},
\quad
c(A_{i})
=
|A_{i}|^{1/2}
\left(\frac{|A_{i,l}^{(d_{i})}|}{|A_{i,r}^{(d_{i})}|}\right)^{u_{i}/2},\label{eq:recon_mean}
\end{align}
where $u_{i}=1$ if $\bm{x}$ is in the left child of $A_{i}$ and
$u_{i}=-1$ if $\bm{x}$ is in the right child. 
The normalization
$c(A_{i})$ corresponds to the UH normalization and ensures
orthonormality in $L^{2}(\mu_{n})$.
If $\tilde{w}_{d,\ell}(A)=w_{d,\ell}(A)$ for all splits in $\mathcal{P}_n$, and if the leaves of $\mathcal{P}_n$ each contain exactly one training location, then the original observed (noisy) signal is recovered.

The amount and quality of denoising will depend on the early stopping and how the wavelet coefficients are regularized in \eqref{eq:recon_mean}.
Our optimization approach will involve two main ingredients. 
The first ingredient is an estimate of the noise level based on a
robust scale measure. Let $\mathcal{W}_{\text{deep}}$ denote a
collection of coefficients at deep levels of the tree where the
signal is expected to be weak. We define
\begin{align}\label{eq:sigmahat}
\hat{\sigma}
=
\frac{\operatorname{MAD}\{w_{d,\ell}(A): (A,d,\ell)\in\mathcal{W}_{\text{deep}}\}}{0.6745},
\end{align}
where MAD is the median absolute deviation recommended by \citet{nason2008wavelet}. This estimator is widely
used in wavelet denoising and is stable in the presence of a small
number of large coefficients (See Sec \ref{sec:Theoretical-Results}).
The second ingredient is a separate threshold applied at
reconstruction. Once the tree is grown and the coefficients
$w_{d,\ell}(A)$ are stored at each internal node, we compute 
soft-thresholded coefficients
\begin{align}\label{eq:soft-threshold}
\tilde{w}_{d,\ell}(A)
=
\operatorname{sign}\bigl(w_{d,\ell}(A)\bigr)\,
\bigl(|w_{d,\ell}(A)|-\tau\bigr)_{+},
\end{align}
with threshold $\tau=a\,\hat{\sigma}\sqrt{2\log n}$ for some user-chosen constant $a \geq 0$. 
More sophisticated ways of thresholding can be used; we will leave this exploration for future work.

\subsection{Bayesian modeling and sampling on UHWTs}
\label{sec:bayesianapproach}

\begin{figure}[ht]
\centering

\includegraphics[width=0.216\textwidth]{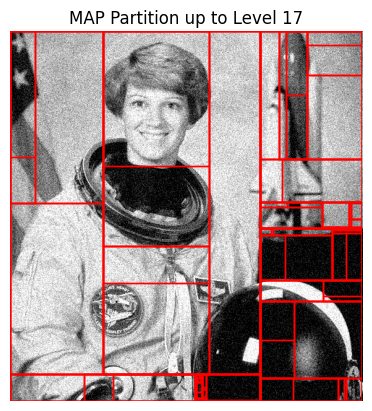}\includegraphics[width=0.78\textwidth]{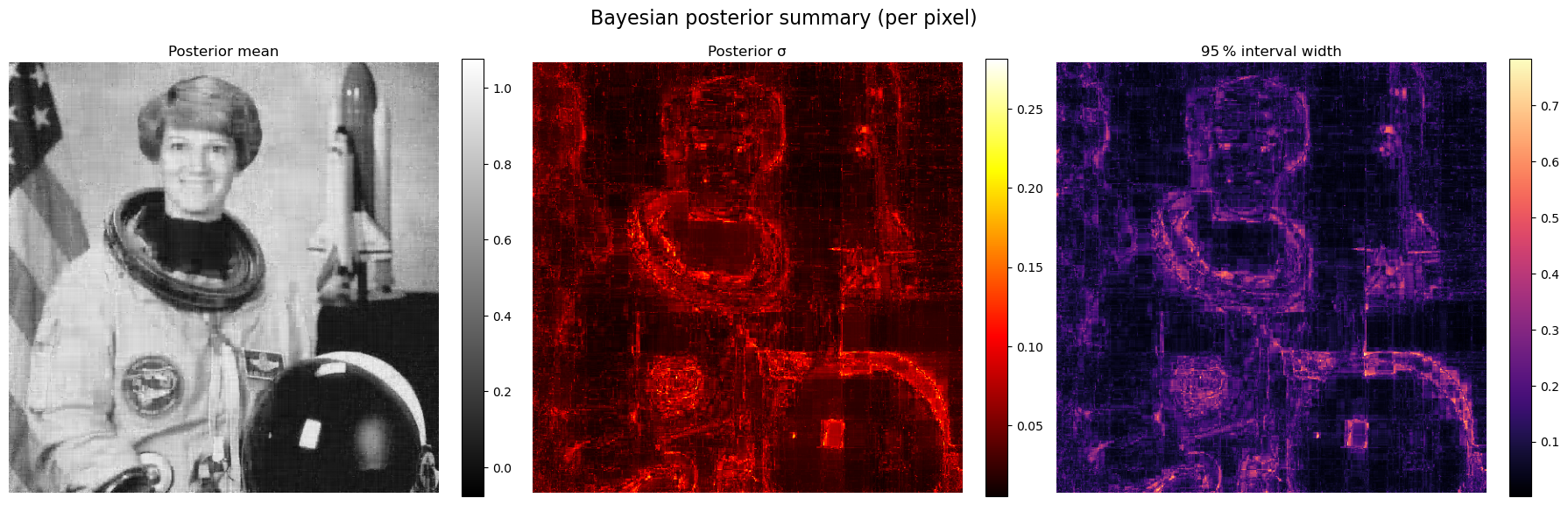}

\caption{First panel: MAP partition overlaid on the input image. 
Posterior mean (second panel), standard
deviation (third panel), and 95\% credible interval width (fourth panel) computed
from 500 full posterior draws of a Bayesian additive UH-tree model
with 200 trees and geometric early-stopping prior probability $0.5^{\text{node depth}}$.
}
\label{fig:comparison_basis-2} 
\end{figure}

The optimization approach will produce a point estimate for the underlying signal function without any associated uncertainty.
In contrast, a Bayesian modeling and sampling approach will induce stochastic sampling from the corresponding posterior distribution, directly allowing uncertainty quantification as shown in Figure~\ref{fig:comparison_basis-2}. 
The prerequisite is to specify a prior distribution on the space of UHWTs. 
Earlier literature introduced a number of priors for wavelets defined on fixed trees splitted at mid-points \citep{li2021learning}. 
Here we allow unbalanced splits, which will be critical for our later ensembling.

Here we describe our prior, which we call a random UHWT (RUHWT) model.
We start at the root node $A_0 = \{\Omega\}$. 
For any node $A$ that can be split, we draw a split dimension $D(A)$ from a set $\mathcal{D}(A)$ of admissible split dimensions according to probability mass function $\{\lambda_{d}(A)\colon d \in \mathcal{D}(A)\}$ and then draw a split location
$L(A)$ from a distribution $B_{A,d}$ on the
unit interval. 
The following theorems will allow $B_{A,d}$ to have full support on $[0,1]$, but 
for a grid-domain application, each $B_{A,d}$ is a discrete distribution on the set $\{i/n_A\colon i=1,\ldots,(n_A-1)\}$; 
we could then specify median splits by setting each $B_{A,d}$ to be a point mass at $0.5$ (assuming $n_A$ is divisible by two).

Given a tree $T$ and its UH system, we specify a Bayesian wavelet
model. In the simplest case,
we assume that the UH coefficients are independent and identically
distributed under a parametric family $p_{i}(w,z\mid\bm{\phi})$ 
--- see \eqref{eq:wavelet-model} ---
where $z$ are local latent variables and $\bm{\phi}$ are
hyperparameters.
The following theorem expresses the marginal posterior for $T$ 
in terms of updated split probabilities that combine
the prior weights $\lambda_{d}(A)$, $B_{A,d}$ and local marginal
likelihoods based on the same coefficients $w_{d,\ell}(A)$ that
drive the optimization rule. 

\begin{theorem}
\label{thm:posterior-nolatent}
Suppose $T$ has a RUHWT prior with split dimension probabilities
$\{\lambda_{d}(A)\colon A\in\mathcal{A},d\in\mathcal{D}(A)\}$ and
split location distributions $\{B_{A,d}\colon A\in\mathcal{A},d\in\mathcal{D}(A)\}$
on $[0,1]$, with zero probability of creating a child with no
training locations. Suppose that, given $T$ and hyperparameters $\bm{\phi}$, the wavelet coefficients satisfy
\begin{align} \label{eq:wavelet-model}
(w_{i},z_{i})\;\ind\;p_{i}(w,z\mid\bm{\phi}), \quad \text{for all }A_i \in T.
\end{align}
Then the marginal
posterior of $T$ is a RUHWT with posterior splitting probabilities
\begin{align}
P(D(A)=d\mid \bm{y})
&\propto
\lambda_{d}(A)\int_0^1
M_{d,\ell}(A)\,
\Phi\bigl(A_{l}^{(d,\ell)}\bigr)\,
\Phi\bigl(A_{r}^{(d,\ell)}\bigr)\,\mathrm{d}B_{A,d}(\ell),\label{eq:posterior_dimension}
\end{align}
and conditional split location distribution
\begin{align}
\mathrm{d}B_{A,d}(\ell\mid D(A)=d,\bm{y})
&\propto
\mathrm{d}B_{A,d}(\ell)\,
M_{d,\ell}(A)\,
\Phi\bigl(A_{l}^{(d,\ell)}\bigr)\,
\Phi\bigl(A_{r}^{(d,\ell)}\bigr),\label{eq:posterior_location}
\end{align}
where
$M_{d,\ell}(A_i)
=
\int p_{i}(w_{d,\ell}(A_i),z\mid\bm{\phi})\mathrm{d}z$
is the marginal likelihood contribution of the coefficient at node
$A_i$ if it is split at $(d,\ell)$, and $\Phi(A)$ satisfies
\begin{align}
\Phi(A)
=
\sum_{d\in D(A)}\lambda_{d}(A)\int_0^1
M_{d,\ell}(A)\,
\Phi\bigl(A_{l}^{(d,\ell)}\bigr)\,
\Phi\bigl(A_{r}^{(d,\ell)}\bigr)\,\mathrm{d}B_{A,d}(\ell)\label{eq:phi2}
\end{align}
when $A$ is not atomic and $\Phi(A)=1$ when $A$ is atomic.
\end{theorem}

Theorem~\ref{thm:posterior-nolatent} provides a posterior distribution on the space of axis-aligned, binary-tree partitions; such a partition will establish a one-to-one correspondence between the leaves of the tree and the training locations $\{\bm{x}_i\}$ if the prior distributions on the splitting dimensions and locations are positive for admissible split dimensions and locations. 
A signal can then be constructed from the partition's associated UH coefficients, but if the coefficients are not regularized, then the constructed signal will be exactly the original observed signal.
We can denoise the signal by introducing latent state variables at each node that, for example, determine whether the node will be split or how the resulting wavelet coefficient will contribute to the reconstructed signal (this is analogous to our optimization approach's splitting safeguards, early stopping mechanism, and soft thresholding). 
Theorem~E.1 (stated and explained in online Appendix~E for space reasons) extends Theorem~\ref{thm:posterior-nolatent} to incorporate latent states.

Theorems~\ref{thm:posterior-nolatent} and E.1
generalize recursive partition posteriors in \citet{li2021learning}
by allowing data-driven split locations and UH systems on irregular
designs. 
The same coefficients $w_{d,\ell}(A)$ that determine the
greedy optimization rule now enter through the marginal likelihood
factors $M_{d,\ell}(A)$ and guide the posterior choice of splits.
In the optimization framework, regularization is applied through
MAD-based thresholds applied to coefficients on a fixed tree. In the
Bayesian framework, regularization is applied through 
depth-dependent split probabilities and the latent states. 

When the prior on
splits is diffuse and the likelihood is sharply peaked at the
largest $|w_{d,\ell}(A)|$ (we provide spike-and-slab and Laplace priors in implementation), the posterior split probabilities
concentrate near the maximizing pair of the optimization rule, so
the greedy algorithm appears as a limiting case. 
This parallel viewpoint
will be important when we construct ensembles and extend the methods
to triangulable manifolds in later sections.

\section{Wavelet Trees for Triangulable Manifold Domains}
\label{sec:Wavelet-Trees-general domain}

We now describe our second major step in generalizing UHWTs. The goal is to define
wavelet trees on spheres or more general triangulable manifolds 
while retaining the key features of the Euclidean construction:
orthogonality in $L^{2}(\mu_{n})$, exact reconstruction at the
sampled locations, and the ability to support both optimization-based
and Bayesian tree building. These properties are crucial for the
ensemble methods on manifolds that will be developed in later
sections.

A curved surface has no coordinate axes that provide natural directions along which to split. 
Our approach thus rests on the fact that any compact orientable two-dimensional
manifold can be approximated, and in fact partitioned, by a finite
simplicial complex of geodesic triangles, which exists on a \emph{triangulable manifold} 
\citep{edelsbrunner1994triangulating}. 
In most practical scenarios, the manifold of interest is triangulable. We therefore assume that the
domain $\mathcal{M}$ is a compact orientable two-dimensional
manifold endowed with an area measure $\mathcal{H}^{2}$,
and that we are given an initial triangulation
whose faces are geodesic triangles. On the sphere, we take the initial
triangulation to be a geodesic polyhedron, such as the icosahedron
\citep{schroder1995spherical}; on a generic surface, we use any
admissible simplicial complex. The same ideas apply to \emph{any manifold} 
that admits such a triangulation.

A geodesic triangle can be bisected
by connecting a point $\xi$ on an edge $e\in\{1,2,\cdots,D\}$
to the opposite vertex via a geodesic arc. 
This split's UH coefficient is defined by
\begin{align}\label{eq:UH_traingulable}
    w(A,e,\xi)
    &=
    \sqrt{\frac{\mu_n(A_{1})\,\mu_n(A_{2})}{\mu_n(A)}}\,
    \bigl(\bar{y}_{A_{1}}-\bar{y}_{A_{2}}\bigr),
    \qquad
    \bar{y}_{A_{j}}
    =
    \frac{1}{\mu_n(A_{j})}\int_{A_{j}} y\,\mathrm d\mu_{n},
\end{align}
where $A_{1}$ and $A_{2}$ are the two child triangles that result from the split,
and $\mu_{n}=n^{-1}\sum_{i=1}^{n}\delta_{\bm{x}_{i}}$ is the
empirical measure of the observed locations
$\bm{x}_{1},\ldots,\bm{x}_{n}\in\mathcal{M}$.
The associated UH wavelet function $\psi_{A,e,\xi}$ is defined as in
\eqref{eq:UHbasis_function} with $|B|$ interpreted as $\mu_n(B)$.
With this choice, the UH scaling and wavelet functions form an orthonormal
system in $L^{2}(\mu_{n})$ on the triangulated manifold, exactly as in the
rectangular case. We still use the geometric area $|A| = \int_A \mathrm d\mathcal H^{2}$
to control triangle shape regularity and to specify admissible splits, but
orthonormality is always defined with respect to $\mu_n$.
Thus each triple $(A,e,\xi)$
defines a UH atom on the manifold in exactly the same way that
$(A,d,\ell)$ defines a UH atom in the Euclidean setting. For visualization, we discuss $D=3$ hereafter, but emphasize that our construction works for any $D\in \mathbb{N}$.

Compared to a balanced Haar construction on a fixed spherical grid, this
edge–cut representation together with the unbalanced normalization have two
advantages on manifolds. First, the split point $\xi$ can move freely along an
edge, so triangles can be cut in a strongly unbalanced way when the response
varies sharply along a boundary while remaining intact when the response is nearly
constant. Second, because the normalization is defined with respect to
$\mu_n$, the associated wavelet atoms still form an orthonormal system in
$L^2(\mu_n)$, so each coefficient $w(A,e,\xi)$ measures the empirical energy
of the contrast between the two child triangles. In practice, this leads each
UH tree to concentrate depth where the regression function exhibits localized
structure on the manifold and, when aggregated in an ensemble, to yield sparse
representations and improved predictive performance on the spherical regression
problems in Section~\ref{sec:Experiments}.

The optimization and Bayesian tree building rules extend to this
setting. In the optimization approach, at a node $A$ we
choose the admissible edge-cut pair $(e,\xi)$ 
that maximizes $|w(A,e,\xi)|$. 
Splitting along $(e,\xi)$ decreases empirical squared
error by $|w(A,e,\xi)|^{2}$, so the greedy rule has the same interpretation
as in the rectangular case. In the Bayesian approach, the RUHWT
prior can now be defined on sequences of edge-cut pairs
$(e,\xi)$, with $\lambda_{e}(A)$ and $B_{A,e}$ replacing
$\lambda_{d}(A)$ and $B_{A,d}$ in
Theorems~\ref{thm:posterior-nolatent} and E.1. The
same coefficients $w(A,e,\xi)$ enter the marginal likelihood factors
$M_{A,e,\xi}$, so the optimization and Bayesian constructions share
the same empirical contrasts. 

Figure~\ref{fig:triangulation_comp-1} shows
four splitting rules for triangles. 
The rule called \texttt{balance} is a classical
mid-edge bisection: the longest edge of the parent triangle is
bisected at its geodesic midpoint, and the far vertex is joined to
this midpoint by a geodesic segment. The rule called \texttt{adapt}
is also a mid-edge bisection, but it chooses the edge
whose geodesic midpoint produces the largest magnitude UH coefficient
among all three edges. The rule
called \texttt{adapt\_vertex} is fully point adaptive. Each interior
data point in $A$ is projected orthogonally onto each edge, and the
pair consisting of an edge and a projection point that maximizes
$|w(A,e,\xi)|$ is selected; the geodesic segment connecting this
projection point to the opposite vertex defines the two children.
The rule called \texttt{balance4} is a canonical four-child split
that connects the three edge midpoints to form four triangles inside
$A$ \citep{schroder1995spherical}. 
The \texttt{adapt} and \texttt{adapt\_vertex}
rules favor anisotropic adaptation; repeated splits refine precisely
those edges where the empirical function varies most, while nearby
regions with little variation remain coarse. The \texttt{balance}
and \texttt{balance4} rules favor shape regularity. In all cases,
after the tree is grown we apply the same reconstruction-time
thresholding of UH coefficients described in \eqref{eq:UH_traingulable}, with thresholds calibrated
from a MAD estimate of noise on fine scale coefficients. This keeps
the partition from reflecting noise or irregular sampling artifacts
and ensures that deep refinements correspond to genuine local
variation.

\begin{figure}[h]
\centering
\includegraphics[width=0.25\textwidth]{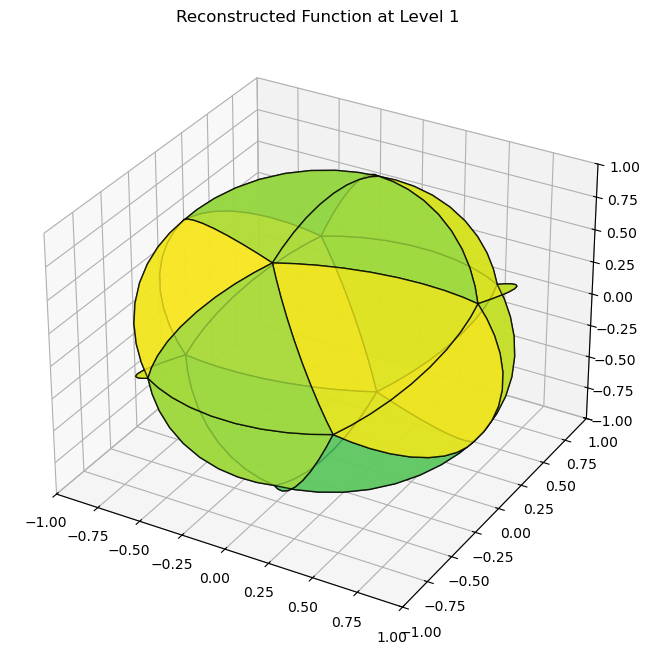}
\includegraphics[width=0.25\textwidth]{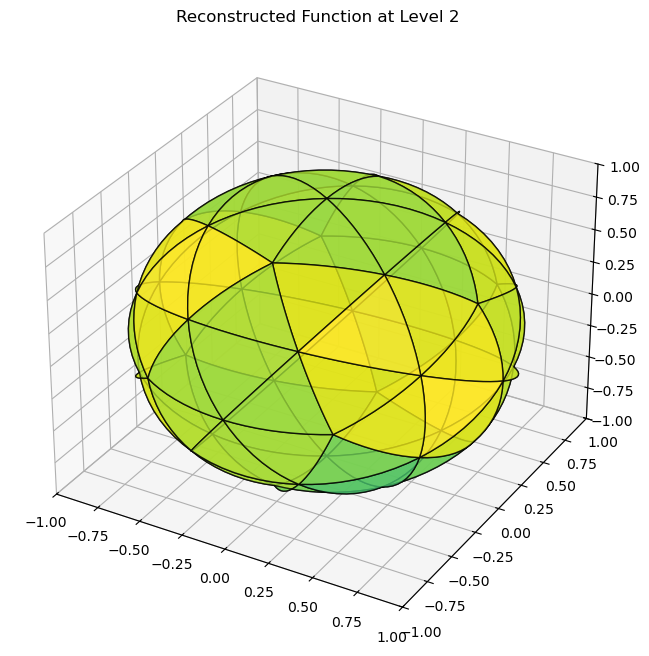}
\includegraphics[width=0.25\textwidth]{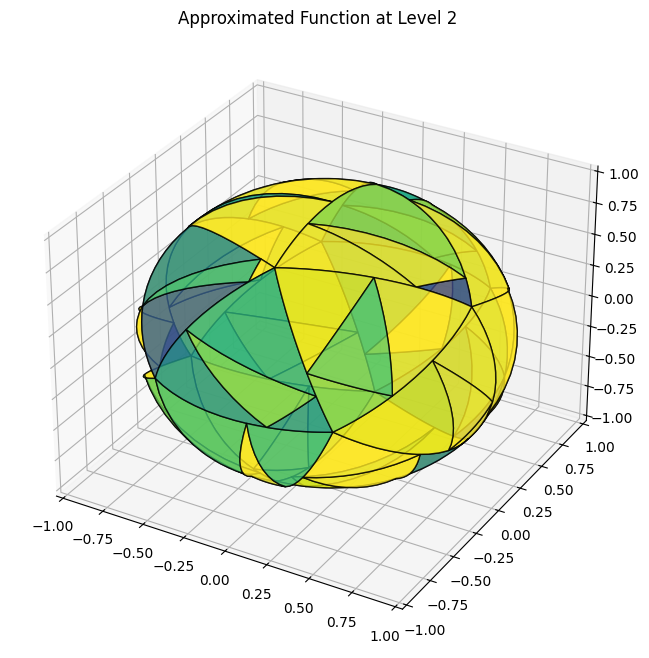}\\
\includegraphics[width=0.9\textwidth]{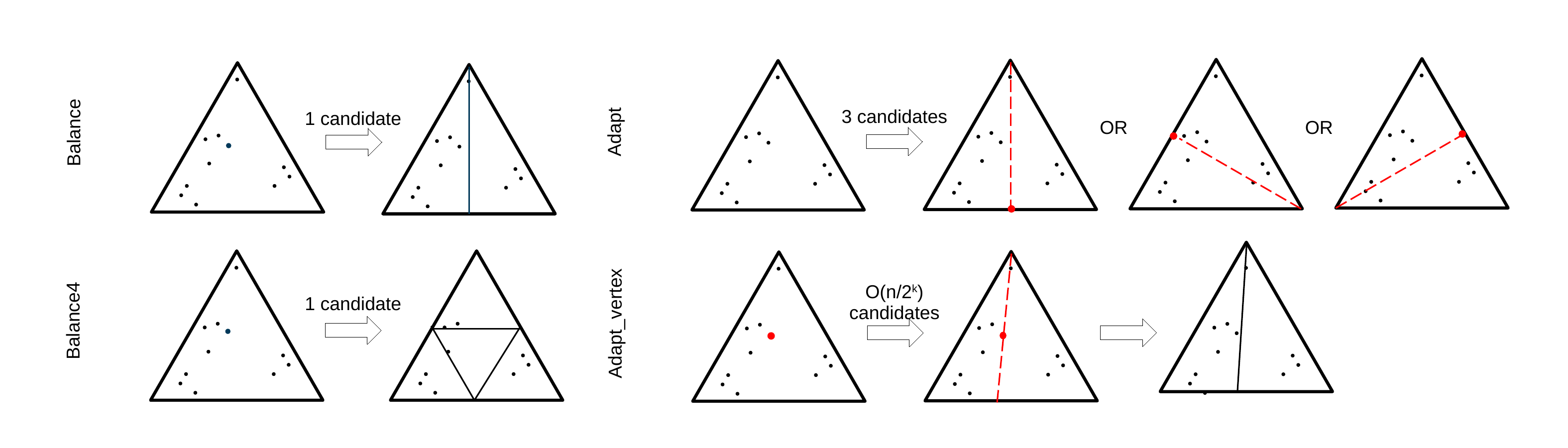}
\caption{Examples of triangulable manifolds 
and four approaches (\texttt{balance}, \texttt{balance4}, \texttt{adapt},
\texttt{adapt\_vertex}) to subdivide a triangle. Points represent
data; the red point indicates the pivot point, and dashed lines show
the cutting edge. The \texttt{balance} and \texttt{balance4} schemes
do not depend on data. The \texttt{adapt} scheme uses data to choose
which edge to bisect at its midpoint. The \texttt{adapt\_vertex}
scheme uses data to choose both the edge and the cut point.}
\label{fig:triangulation_comp-1}
\end{figure}

Table~1 in online Appendix~J 
summarizes the asymptotic training complexity for these four splitting
rules (online Appendix~L.3.2 has an empirical comparison). 
The main computational cost arises from determining which
sample points lie in each child triangle. 
A triangulation and barycentric membership test suffice for this.

To summarize, the triangulable manifold constructions extend classical tree methods to more general domains, and also extend the UH tree framework from (hyper)rectangular grid data to triangulable manifolds. 
Both the optimization and Bayesian constructions
extend naturally: in each case the same UH coefficients guide split
selection, and regularization is imposed either by MAD-based
thresholding at reconstruction or by biased priors and posterior
shrinkage in the RUHWT model. 
(Online Appendix~K compares our UHWT model on a triangulable manifold to other existing tree-based methods for manifold domains.)
These constructions will form the
basis for the ensemble methods developed in the next section.

\section{Additive Unbalanced Haar Wavelet Tree Ensembles}
\label{sec:Wavelet-Regression-ensemble}
Ensembles improve prediction by combining many weak learners. 
Random Forests typically grow deep trees, which
can be less compatible with sparse multiscale representations. 
In contrast, classical boosting and BART regularize by restricting the number of leaf nodes in each tree.
Unlike tree fits used in previous methods, UHWTs provide orthogonal expansions with wavelet coefficient-based shrinkage. We thus aim to construct both a boosting algorithm and a Bayesian method analogous to BART, using the UHWT as the weak learner. 

A single UH tree $T$ yields the orthogonal expansion
\begin{align}
g(\bm{x})
&=
\bar{y}_{\Omega}+\sum_{j\in J(T)}\tilde{w}_{j}\,\psi_{j}(\bm{x}).
\label{eq:single_tree_expansion}
\end{align}
Here $\bar{y}_{\Omega}\coloneqq n^{-1}\sum_{i=1}^{n}y_i$ is the global mean;
$\mathcal{P}(T)$ denotes the recursive partition induced by $T$; $J(T)$
indexes the internal nodes/splits of $T$
and hence the non-constant UH atoms; $\{\psi_j\}_{j\in J(T)}$ are those UH
atoms, normalized to be orthonormal in $L^2(\mu_n)$ as in
Section~\ref{sec:optimizationapproach}; and $\tilde w_j$ are (possibly
shrunk) empirical coefficients.

\subsection{Boosting for UHWTs}
\label{subsec:Boosting-the-wavelet}

In general, forward stagewise boosting \citep{Friedman01} builds an ensemble model $f_G(\cdot), G \in \mathbb{N}$ by iteratively adding ``weak'' learners $g_t(\cdot), t=1,\ldots,G$ to improve predictions.
Our weak learners will belong to the class $\mathcal{G}$ of UH tree fits built by the optimization approach of Section~\ref{sec:optimizationapproach} (using the UH coefficients
\eqref{eq:wavelet_coef} and the greedy split rule \eqref{eq:opt_approach}).
The fits are then further regularized by a pre-chosen learning rate $\eta\in(0,1]$.
Hence, our procedure is forward stagewise boosting specialized to UH trees;
the difference from CART-based boosting is that each $g_t$ comes with an
orthonormal UH expansion and admits coefficient shrinkage.

Let $\{(\bm{x}_i,y_i)\}_{i=1}^n$ be training data on $\Omega$ with empirical
measure $\mu_n$, and let $y$ denote the observed signal viewed as a function
on the sample. We use the squared-error loss $\frac{1}{2n}\sum_{i=1}^n(y_i-f(\bm{x}_i))^2 = \frac{1}{2}\,\|y-f\|_{L^2(\mu_n)}^2$
and initialize $f_0(\bm{x})\equiv \bar y_\Omega$.
At each stage $t=1,\ldots,G$, 
our ensemble model $f_t$ will be updated as
\begin{align}
f_t(\bm{x})
&=
f_{t-1}(\bm{x})+\eta\,g_t(\bm{x}),
\label{eq:boosting_update}
\end{align}
where the weak learner $g_t\in\mathcal{G}$ is fit to residuals
$r_{t-1,i} \coloneqq y_i-f_{t-1}(\bm{x}_i), i=1,\ldots,n,$
by (approximately) solving
\begin{align}
g_t
&\in
\arg\min_{g\in\mathcal{G}}
\frac{1}{2n}\sum_{i=1}^n\bigl(r_{t-1,i}-g(\bm{x}_i)\bigr)^2.
\label{eq:stagewise_objective}
\end{align}
A UHWT weak learner built on the residuals has the form
\begin{align}
g_t(\bm{x})
&=
\sum_{j\in J(T_t)}\tilde c^{(t)}_{j}\,\psi^{(t)}_{j}(\bm{x}),
\quad
\tilde c^{(t)}_{j}
=
\operatorname{sign}\Bigl(c^{(t)}_{j}\Bigr)\Bigl(\left|c^{(t)}_{j}\right|-\tau_t\Bigr)_+, 
\quad
c^{(t)}_{j}
=
\langle r_{t-1},\psi^{(t)}_{j}\rangle_{L^2(\mu_n)}.
\label{eq:g_t_residual}
\end{align}
Here $T_t$ is the fitted UH tree at stage $t$; $J(T_t)$ indexes its internal
nodes/splits; $\psi_j^{(t)}$ are the corresponding UH atoms (orthonormal in
$L^2(\mu_n)$); $\langle\cdot,\cdot\rangle_{L^2(\mu_n)}$ is the empirical inner
product; and $(u)_+\coloneqq \max\{u,0\}$.
The threshold $\tau_t$ is calibrated from a MAD estimate of noise on fine-scale
residual coefficients (as in Section~\ref{sec:optimizationapproach}).

To connect \eqref{eq:stagewise_objective} to UH orthogonality, fix a candidate
tree $T$ and write $g$ supported on $T$ as
$g(\bm{x}) = \sum_{j\in J(T)}\theta_j\,\psi_j(\bm{x})$.
Substituting this into \eqref{eq:stagewise_objective}, orthonormality in $L^2(\mu_n)$ yields
\begin{align}
\frac{1}{2}\|r_{t-1}-g\|_{L^2(\mu_n)}^2
&=
\frac{1}{2}\|r_{t-1}\|_{L^2(\mu_n)}^2
-\sum_{j\in J(T)}\theta_j\,\langle r_{t-1},\psi_j\rangle_{L^2(\mu_n)}
+\frac{1}{2}\sum_{j\in J(T)}\theta_j^2.
\label{eq:orthogonal_expansion_residual}
\end{align}
For fixed $T$, the minimizer over $\theta$ is
$\theta_j^\ast = \langle r_{t-1},\psi_j\rangle_{L^2(\mu_n)}$ for all $j\in J(T)$,
and the associated risk reduction has the same form as in orthogonal wavelet
shrinkage \citep{donoho1995wavelet}.

Using the UHWT fitting (See online Appendix~G.4), on the sphere we draw $R_t\in\mathrm{SO}(3)$ from Haar measure and build $T_t$ (hence $\psi_j^{(t)}$)
on the rotated locations $R_t\bm{x}_i$. The residuals $r_{t-1,i}$ remain evaluated
at $\bm{x}_i$, but the learner is pulled back by composition:
\begin{align}
g_t(\bm{x})
&=
\sum_{j\in J(T_t)}\tilde c^{(t)}_j\,\psi^{(t)}_j(R_t\bm{x}),
\qquad
f_t(\bm{x})
=
f_{t-1}(\bm{x})+\eta\,g_t(\bm{x}).
\label{eq:boosting_sphere}
\end{align}
On a general triangulable manifold $\mathcal{M}$, we replace $R_t$ in \eqref{eq:boosting_sphere} by a random
intrinsic reparameterization $\Phi_t$.
On spheres, we provide two options: Euclidean PCA and sPCA \citep{hrluo_2021c}, as shown later in Section \ref{sec:Experiments}.

\subsection{Bayesian backfitting sampling for RUHWT ensembles}
\label{subsec:Backfitting-the-wavelet}

We now describe a Bayesian ensemble that builds on the RUHWT model of
Section~\ref{sec:bayesianapproach} and uses Bayesian backfitting
\citep{hastie2000bayesian,chipman2010bart}. 
Let $m$ denote the number of component trees in the ensemble. We model
\begin{align}
f_m(\bm{x})
&=
\mu+\sum_{t=1}^{m} g_t(\bm{x};\theta_t),
\label{eq:additive_model}
\end{align}
where $\mu$ is a global intercept and each $\theta_t$ collects the $t$th tree
structure, its UH coefficients, and any local scales under the RUHWT prior.
Conditional on $f_m$,
\begin{align}
y_i
&=
f_m(\bm{x}_i)+\epsilon_i,
\qquad
\epsilon_i\sim N(0,\sigma^2),
\quad i=1,\ldots,n.
\label{eq:observation_model}
\end{align}
We place a prior $p(\mu,\sigma^2)$ on global parameters and independent RUHWT
priors on $\theta_1,\ldots,\theta_m$ as in Section~\ref{sec:bayesianapproach}.
A Gibbs sampler cycles through components and conditionally updates $\theta_t$ given the others
by fitting a RUHWT to the current partial residual.

Given a pre-defined number of sweeps, at component $t\in\{1,\ldots,m\}$ and sweep $s$ we define
\begin{align}
r^{(s)}_{t,i}
&=
y_i-\mu^{(s)}-\sum_{k\neq t} g_k\big(\bm{x}_i;\theta_k^{(s)}\big),
\qquad i=1,\ldots,n,
\label{eq:backfit_residual}
\end{align}
and update $\theta_t$ from its conditional posterior
\begin{align}
p\big(\theta_t\mid r_t^{(s)},\sigma^2\big)
&\propto
p(\theta_t)\,
\prod_{i=1}^n
\exp\biggl\{-\frac{1}{2\sigma^2}\Bigl(r^{(s)}_{t,i}-g_t(\bm{x}_i;\theta_t)\Bigr)^2\biggr\}.
\label{eq:conditional_posterior_theta}
\end{align}
We use Metropolis-within-Gibbs proposals that locally modify $\theta_t$ and
accept/reject via the usual Metropolis ratio. (The ``latent-state Markov tree''
regularization referenced in Theorem~E.1 in online Appendix~E is the same as in
Section~\ref{sec:bayesianapproach} and is distinct from the UH trees $T_t$.)

Orthogonality yields simple coefficient updates in a fixed tree. For example,
if the coefficient on atom $\psi_j^{(t)}$ has prior $w_j^{(t)}\sim N(0,\tau_t^2)$,
then under \eqref{eq:observation_model} its conditional posterior satisfies
\begin{align}
\operatorname{Var}\Bigl(w_j^{(t)}\mid T_t,r_t^{(s)},\sigma^2\Bigr)
&=
\biggl(\frac{1}{\tau_t^2}+\frac{1}{\sigma^2}\biggr)^{-1},
\\
\mathbb{E}\Bigl[w_j^{(t)}\mid T_t,r_t^{(s)},\sigma^2\Bigr]
&=
\operatorname{Var}\Bigl(w_j^{(t)}\mid T_t,r_t^{(s)},\sigma^2\Bigr)\,
\frac{\langle r_t^{(s)},\psi_j^{(t)}\rangle_{L^2(\mu_n)}}{\sigma^2}.
\label{eq:posterior_coef_normal}
\end{align}
More flexible priors (e.g., spike-and-slab or global--local shrinkage) can be
handled via the latent-state construction in Theorem~E.1 in online Appendix~E.

\subsection{Regularization}
\label{subsec:Regularization-comparison}

Regularization enters (i) during tree growth, (ii) during reconstruction from a
fixed tree, and (iii) at the ensembling stage.
When growing a tree, the optimization scheme regularizes through split
selection and stopping rules \eqref{eq:opt_approach}; the Bayesian scheme
regularizes through RUHWT priors and the latent states of
Theorem~E.1 (Section~\ref{sec:bayesianapproach}).
During reconstruction, regularization acts directly on UH coefficients.
In the optimization framework, we threshold coefficients as in
\eqref{eq:soft-threshold} before reconstructing (cf.\ \eqref{eq:recon_mean});
in the Bayesian framework, posterior means/draws of coefficients are shrunk by
their priors, the likelihood, and latent states.

At the ensembling stage, boosting reduces 
the loss $\frac{1}{n}\sum_{i=1}^n(y_i-f(\bm{x}_i))^2$
via updates
\eqref{eq:boosting_update}, where each $g_t$ in \eqref{eq:g_t_residual} is a
residual-fitted UH tree with coefficient shrinkage, and $\eta$ limits each
tree's contribution. Bayesian backfitting uses the additive model
\eqref{eq:additive_model}--\eqref{eq:posterior_coef_normal}, updating each component
on the partial residual \eqref{eq:backfit_residual}; as residuals shrink, later
posterior draws concentrate near small trees with small coefficients, so the
ensemble behaves like a structured shrinkage prior over many UH coefficients.

Classical tree ensembles regularize mainly via tree size (e.g., pruning) or leaf
priors \citep{hrluo_2022e,hastie1990shrinking,breiman1984classification,zhang2025minimax}. In contrast, UH ensembles always operate in orthogonal UH systems induced
by the learned partitions and regularize explicitly at the split, coefficient,
and ensemble stage using the same wavelet contrasts that drive
\eqref{eq:opt_approach}, \eqref{eq:recon_mean}, and \eqref{eq:additive_model}.

\section{Theoretical Results}
\label{sec:Theoretical-Results}

The recursive constructions on rectangular and triangulable domains
use the same principle: 
starting with a base partition $\mathcal{P}_0$ 
(typically with $|\mathcal{P}_0|=1$ cell for a rectangular domain or $|\mathcal{P}_0|=20$ cells for a spherical domain using an icosahedron base partition),
each split of a cell creates two
children and introduces exactly one new UH ``detail'' supported on the
parent that is orthonormal in $L^{2}(\mu)$, where $\mu=\mu_n$ is the
empirical measure on the sample points. After all splits, we obtain a
partition $\mathcal{P}$ with $|\mathcal{P}|$ cells and $M=|\mathcal{P}|-|\mathcal{P}_0|$
non-constant UH details.

Let $\mathcal{P}=\{A_{1},\dots,A_{|\mathcal{P}|}\}$ be a partition obtained
by admissible splits. Each split
$s\colon A\to(A_{1},A_{2})$ defines an UH atom
\[
\psi_{s}
=
\sqrt{\frac{\mu(A_{1})\mu(A_{2})}{\mu(A)}}\left(\frac{\mathbf{1}_{A_{1}}}{\mu(A_{1})}-\frac{\mathbf{1}_{A_{2}}}{\mu(A_{2})}\right),
\quad
\langle\psi_{s},1\rangle_{L^{2}(\mu)}=0,\quad
\|\psi_{s}\|_{L^{2}(\mu)}=1.
\]
Collecting all such atoms for $\mathcal{P}$ gives
$\{\psi_{j}\}_{j=1}^{M}$, which, together
with the constant, span the piecewise-constant functions on
$\mathcal{P}$. 
For each $j=1,\ldots,M$, let $w_{j}=\langle Y,\psi_{j}\rangle_{L^{2}(\mu)}$.
For any $S\subset\{1,\dots,M\}$, let $f_{S}$ be the $L^{2}(\mu)$-projection of $f$ onto
$\mathrm{span}\{\psi_{j}\colon j\in S\}$, which means that at most $|S|$ out of all possible wavelet basis functions have non-zero coefficients.

\begin{theorem}[Oracle bounds on a fixed partition]
\label{thm:fixed-partition}
Let $(\mathcal{D},\nu)$ be either a rectangular grid with its
reference measure or a triangulable domain with the empirical measure
on the sample points. Let the recursive rule of this paper produce a
partition $\mathcal{P}$ and corresponding orthonormal UH details
$\{\psi_{j}\}_{j=1}^{M}$ as above. 
Fix $\mathcal{P}$ and $0<\delta<1$, and set
\begin{align}\label{eq:tau}
\tau = \sigma\sqrt{\frac{2}{n}\log\frac{2M}{\delta}}.
\end{align}
The soft-thresholded UH estimator
\[
\hat{f}_{\mathcal{P},\tau}^{\mathrm{UH}}(\bm{x})
=
\bar{Y}
+
\sum_{j=1}^{M}
\operatorname{sgn}(w_{j})\bigl(|w_{j}|-\tau\bigr)_{+}\psi_{j}(\bm{x}),
\qquad
\bar{Y}=\int Y\,d\mu,
\qquad
w_j = \langle Y,\psi_{j}\rangle_{L^{2}(\mu)},
\]
satisfies, with probability at least $1-\delta$,
\begin{align}
\|\hat{f}_{\mathcal{P},\tau}^{\mathrm{UH}}-f\|_{L^{2}(\mu)}^{2}
\;\le\;
4
\min_{S\subset\{1,\dots,M\}}
\Bigl\{
\|f-f_{S}\|_{L^{2}(\mu)}^{2}
+
|S|\tau^2
\Bigr\}.
\label{eq:UH-oracle-clean}
\end{align}

On the same partition, the leafwise (CART-style) estimator
\[
\hat{f}_{\mathcal{P}}^{\mathrm{leaf}}(\bm{x})
=
\frac{1}{N_{A}}
\sum_{i: \bm{X}_{i}\in A}Y_{i},
\qquad \bm{x}\in A\in\mathcal{P},
\]
satisfies, with probability at least $1-\delta$,
\begin{align}
\|\hat{f}_{\mathcal{P}}^{\mathrm{leaf}}-f\|_{L^{2}(\mu)}^{2}
\;\le\;
\|f-f_{\mathcal{P}}\|_{L^{2}(\mu)}^{2}
+
4\,\frac{\sigma^{2}}{m_{\min}(\mathcal{P})}
\log\biggl(\frac{2|\mathcal{P}|}{\delta}\biggr),
\label{eq:leaf-clean}
\end{align}
where $f_{\mathcal{P}}$ is the piecewise-constant projection of $f$
on $\mathcal{P}$ and $m_{\min}(\mathcal{P})=\min_{A\in\mathcal{P}}N_{A}$.
If $f$ is Hölder-$\alpha$ on each $A$ with constant $L_{f}$, then
$\|f-f_{\mathcal{P}}\|_{L^{2}(\mu)}^{2}
\le
L_{f}^{2}\sum_{A\in\mathcal{P}}h_{A}^{2\alpha}\mu(A)$.
\end{theorem}

Inequality \eqref{eq:UH-oracle-clean} shows that, on any fixed tree
partition, UH soft thresholding automatically adapts to the sparsity of
$f$ in the UH dictionary: if $f$ is well approximated by a small
set $S$ of coefficients, the risk scales like $|S|/n$ (up to
logarithms). In contrast, \eqref{eq:leaf-clean} shows that the
stochastic term for leafwise means depends on the smallest leaf size,
and effectively on the total number of leaves. This supports our use
of UH thresholding inside each tree in Sections~\ref{sec:Wavelet-Regression}
and \ref{sec:Wavelet-Regression-ensemble}: once a tree is grown, we
reconstruct via \eqref{eq:recon_mean} but only with coefficients
that pass a threshold, which is exactly the regime where
\eqref{eq:UH-oracle-clean} is favorable.

Finally, we state a simplified consequence of
Theorem~\ref{thm:fixed-partition} that highlights the benefits of UH
thresholding when the true signal is sparse in the UH dictionary.

\begin{corollary}[Sparse detail vectors]
\label{cor:sparse}
Assume the setting of Theorem~\ref{thm:fixed-partition} on a fixed
partition $\mathcal{P}$, and that $f$ is exactly supported
on a set $S^{\star}$ of UH atoms with $s=|S^{\star}|\ll M$, so
$f=f_{S^{\star}}$. Assume also that every cell contains at least a constant fraction $c \in (0,1/2)$ of
$n/|\mathcal{P}|$ points. Then, for the UH estimator
$\widehat{f}_{\mathcal{P},\tau}^{\mathrm{UH}}$ with
$\tau$ as in \eqref{eq:tau} and the leafwise estimator
$\widehat{f}_{\mathcal{P}}^{\mathrm{leaf}}$ in
\eqref{eq:leaf-clean}, 
with probability at least $1-\delta$,
\[
\|\widehat{f}_{\mathcal{P},\tau}^{\mathrm{UH}}-f\|_{L^{2}(\mu)}^{2}
\;\le\;
8\,s\,\frac{\sigma^{2}}{n}\log\frac{2M}{\delta},
\qquad
\|\widehat{f}_{\mathcal{P}}^{\mathrm{leaf}}-f\|_{L^{2}(\mu)}^{2}
\;\le\;
\frac{4}{c}\,|\mathcal{P}|\,\frac{\sigma^{2}}{n}
\log\frac{2|\mathcal{P}|}{\delta}.
\]
\end{corollary}

This corollary formalizes the advantage of reconstruction-time UH
thresholding in scenarios where the signal is concentrated on a few
nodes, such as localized bumps or spikes (including functions like
Easom-type examples mentioned in
Section~\ref{sec:Wavelet-Regression}). On a fixed tree,
$\widehat{f}_{\mathcal{P},\tau}^{\mathrm{UH}}$ pays variance
proportional to the number of large coefficients $s$, whereas
leafwise means pay variance proportional to the total number of
leaves $|\mathcal{P}|$. In our ensemble constructions, each boosting
step or backfitting component performs exactly this kind of
coefficient selection on its own residuals, so the same sparse-versus-dense
gap appears componentwise. Taken together,
Theorem~\ref{thm:fixed-partition} and Corollary~\ref{cor:sparse} provide the
theoretical backing for both pillars: UHWTs offer a
sparsity-adaptive, orthogonal representation on adaptive partitions,
and this structure survives when we move from single trees to additive
ensembles and from Euclidean grids to triangulable manifolds.

\section{Numerical Experiments\label{sec:Experiments} }

This section will illustrate our methodology using simulated data on the image and sphere domains, and real data on the sphere domain. 
We will relegate our numerical experiments
on tensor data to online Appendix~L.2, and our
real-data image and additional spherical simulations to respective
online Appendix~L.1 and online Appendix~L.3.

\subsection{BH versus UH for boosted ensembles on images}

\label{subsec:numerical-image}

\begin{figure}[h]
\centering \includegraphics[width=1\textwidth]{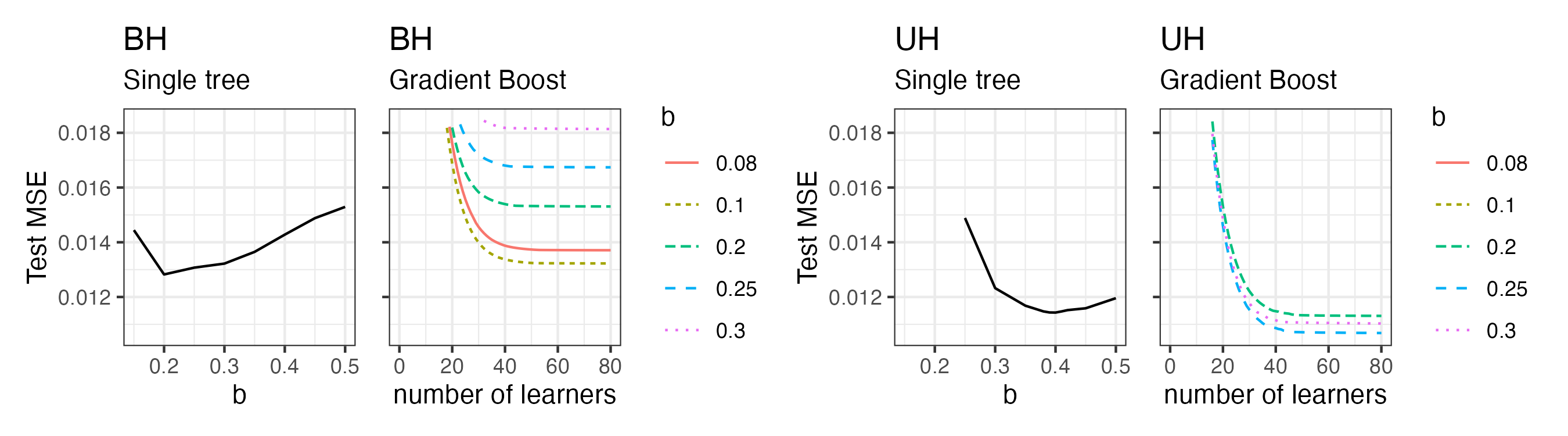}
\caption{Metrics of methods fit to a $512\times512$ pixel astronaut image
with additive Gaussian noise whose standard deviation equals the noiseless
image's standard deviation. For each boosted ensemble, the learning
rate is 0.1. The hyperparameter $b$ is as defined at the beginning
of Sec~\ref{subsec:numerical-image}. %
For the fourth panel, the test MSEs for $b\in\{0.08,0.1\}$ were larger
than $0.018$ and hence are not shown.}
\label{fig:astronaut_numerical} 
\end{figure}

This experiment illustrates the need for weak learners in a boosted
ensemble to be able to split away from the midpoint. For each tree,
we will stop splitting any node whose wavelet coefficient \eqref{eq:wavelet_coef}
is smaller than the threshold $\tau=b\,\hat{\sigma}$
\citep{donoho1995noising} for a user-defined hyperparameter $b\geq0$.
For both BH and UH trees and various values of $b$, we fit a 
single-tree learner and boosting ensembles to a
$512\times512$ pixel astronaut image with additive Gaussian noise
whose standard deviation equals the noiseless image's standard deviation.

Figure~\ref{fig:astronaut_numerical} shows the test MSE for the
fitted models. The left two panels show that boosting using BH weak
learners does not seem to reduce test MSE compared to a single BH
tree: the smallest BH single-tree test MSE is $\approx0.0129$
(at $b=0.2$), which is smaller than any of the boosting test MSEs
with BH weak learners (we computed these test
MSEs also for $b=0.09$ and $b=0.11$, but omitted their MSE curves
to avoid visual clutter). This supports the notion that boosting will not be effective if the weak learners always split at the midpoint. On the other hand, the
third panel shows that the smallest UH single-tree test MSE is approximately
$0.0114$ (at $b\approx0.4$), which is noticeably larger than some
of the shown (fourth panel) boosting test MSEs with UH weak learners,
namely with $b\in\{0.25,0.3\}$; this implies that boosting can improve
upon single-tree learners if learners are allowed to split away from
the midpoint. Furthermore, the optimal value of $b$ for a single-tree
learner is larger than that for a boosted ensemble with learning rate
0.1; the learning rate allows the trees to grow deeper without as
much risk of overfitting as when using a single-tree learner.

\subsection{Numerical results for sphere: Beyond Matérn kernels}

\label{sec:sphere-numerical}

Here we illustrate our spherical methodology proposed in Section~\ref{subsec:Boosting-the-wavelet}
on simulated data. The training data consist of $n=300$
i.i.d.\ data points $(\bm{x}_{i},y_{i})$, $i=1,\dots,n$, where the locations
$\bm{x}_{i}\in\mathbb{S}_{2}$ are drawn uniformly on a unit sphere, and
the observed $y_{i}$ are noisy versions of the unseen signal defined
in the caption of Figure~\ref{fig:sphere-mse}. Predictive performance
is assessed by the MSE of the trained model on $15300$
test sample points.

Although \citet{blaser2016random} mentioned the idea 
of incorporating random rotations into a boosted ensemble, 
we are not aware of any public implementation.
We thus implement a \textit{Random-Rotation Boosting} approach whose base learner will be our sphere wavelet, as described in Section~\ref{subsec:Boosting-the-wavelet}.
For comparison, we will also use a Fréchet tree \citep{capitaine2024frechet}
as a base learner, which we believe has never been publicly implemented
before; we will call this approach \textit{Random-Rotation Fréchet Boosting}.

We also compare our approach against existing tree-ensemble methods
that allow inputs that lie on a sphere, namely Fréchet Forest \citep{capitaine2024frechet},
Random-Rotation Random Forest (called Random-Rotation Ensembles in
\citet{blaser2016random}), and spherical PCA \citep{hrluo_2021c} with Rotation-Forest
\citep{rodriguez2006rotation}. For Fréchet Forest, we used the \texttt{FrechForest}
R package implementation. 
We also compare against three state-of-the-art nonparametric residual deep
Gaussian process (GP) models \citep{wyrwal2025residual}, which can
model manifold- and scalar-valued functions with inputs supported
on a manifold. \citet{wyrwal2025residual} can be considered as a
deep ensemble extension of \citet{Gramacy2005} to spherical domains,
with more sophisticated covariance and possible manifold uncertainty quantification \citep{luo2024multiple}.

\afterpage{
\begin{figure}[h!]
\centering %
\includegraphics[width=1\textwidth]{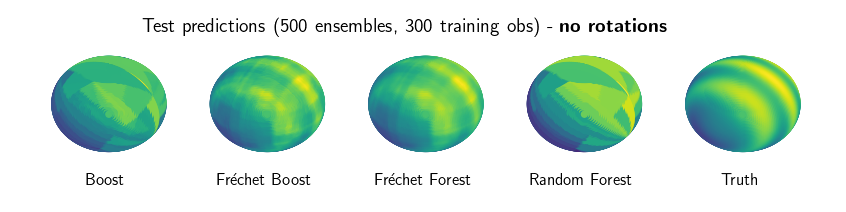}
\includegraphics[width=1\textwidth]{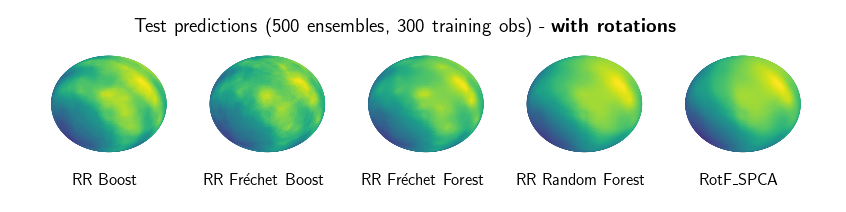}
\includegraphics[width=1\textwidth]{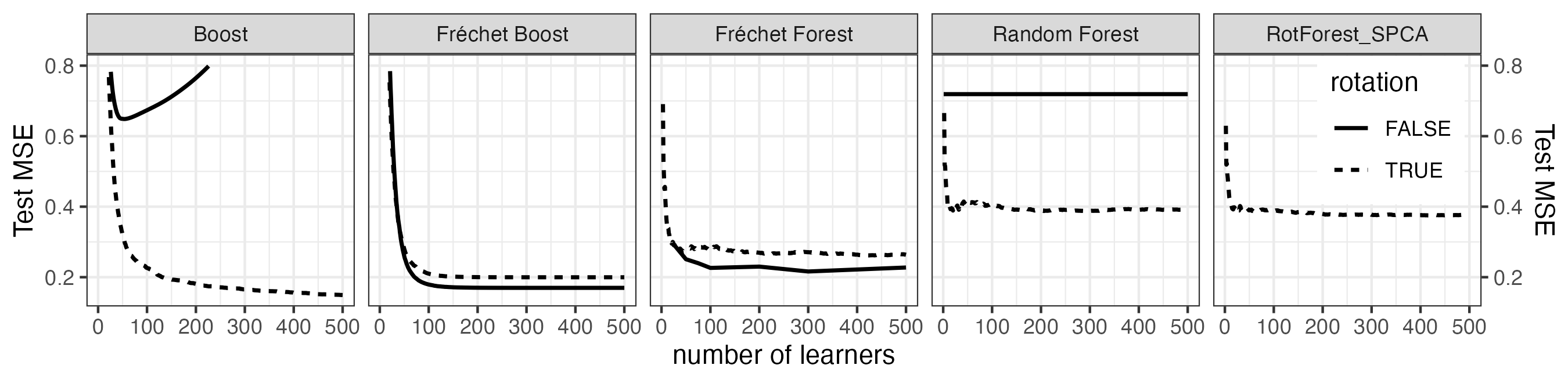}
\begin{tabular}{rlll}
\toprule 
 & 500 ensembles & 1000 ensembles & 1500 ensembles \tabularnewline
\midrule 
Test MSE for RR Boost  & $0.134$  & $0.129$  & $\mathbf{0.128}$ \tabularnewline 
\bottomrule
\end{tabular}
\caption{Empirical comparison of various tree ensemble models (\textquotedblleft RR\textquotedblright{}
refers to the method version that randomly rotates covariates before
fitting any base learner) trained on $n=300$ points $(\bm{x},y)$, where
each location $\bm{x}=(x_{1},x_{2},x_{3})$ was drawn uniformly on the
unit sphere, and the observed response $y$ at $\bm{x}$
is the signal $f(\bm{x})=2\tanh(x_{1})+\cos(10x_{2})+2x_{3}$ plus i.i.d.\ Gaussian
noise with zero mean and standard deviation equal to $0.1$ of the
standard deviation of $f$. 
Trees are allowed to grow to full depth. 
Boosted methods use learning rate $0.05$. 
Top two rows: predictions of the ensemble methods. 
The true signal is shown in the top-right sphere. 
Third row: MSE on 15300 test points.}
\label{fig:sphere-mse} 
\end{figure}

\begin{table}[h]
\centering %
\begin{tabular}{rllll}
\toprule 
Test MSE & 1 layer & 2 layers & 3 layers & 4 layers \tabularnewline
\midrule 
{\tiny +hodge+spherical\_harmonic\_features}  & 0.467 (0)  & 0.466 (0.000)  & 0.466 (0.000)  & 0.466 (0.000) \tabularnewline
{\tiny +inducing\_points}  & 0.466 (0.001)  & 0.208 (0.091)  & 0.209 (0.093)  & 0.207 (0.093) \tabularnewline
{\tiny +spherical\_harmonic\_features}  & 0.467 (0)  & \textbf{0.182 (0.010)}  & 0.192 (0.027)  & 0.199 (0.017) \tabularnewline
\bottomrule
\end{tabular}\caption{Mean and standard deviation of test MSE for three residual deep GP
methods over 10 starting seeds (with variance provided in brackets), with different combinations of parameters as in their codebase. Data is described in Figure~\ref{fig:sphere-mse}. Smallest MSE is bolded.}
\vskip 0.1cm
\label{tab:deepGP} 
\end{table}
}

Figure~\ref{fig:sphere-mse} show the tree-ensemble predictions.
The rotations smooth out the artificial boundaries or striations found in the predictions without rotations.
Interestingly,
the rotations increase the test MSE for the ensembles whose base learner is a Fréchet tree,
but decrease the test MSE for the ensembles whose base learner is
our wavelet tree using \texttt{adapt}. The rotations seem to
increase the bias and variance of each base learner for the sake of
increasing diversity between learners; we explore this phenomenon
in online Appendix~L.3.1.
Our Random-Rotation Boosting approach produces the smallest test MSE
among the tested methods, including the residual deep GP models whose
MSEs are shown in Table~\ref{tab:deepGP}; this was also observed for 
different signals, as described in online Appendices~L.3.3 and L.3.4, 
as well as in the real-data example in Section~\ref{sec:GISS}.
Furthermore, a spherical wavelet
tree can be applied to generic topological triangulations and hence
can be used for a larger class of input domains than a Fréchet
tree can, which requires inputs to lie on a metric space.

Finally, Table~\ref{tab:sphere-softthreshold-2} shows the predictive improvement attained by incorporating a modest amount of soft thresholding into the Random-Rotation Boosting wavelets. 
Here we use the same functions, but we increase the sample size from 300 to 1000 to allow the trees to grow deep enough for soft thresholding to nontrivially affect the final prediction. We also increase the noise standard deviation to 0.3 of the standard deviation of $f$.

\begin{table}[h]
    \centering
    \begin{tabular}{rlllll}
    \toprule 
    Function & soft 0.0 & soft 0.1 & soft 0.2 & soft 0.3 & soft 0.4 \tabularnewline
    \midrule 
    {\tiny from Figure~\ref{fig:sphere-mse}} & $0.0227$ & $0.0218$ & $\mathbf{0.0217}$ & $0.0224$ & $0.0239$ \tabularnewline 
    {\tiny from online Appendix~L.3.3} & $0.1033$ & $0.1006$ & $0.0989$ & $\mathbf{0.0981}$ & $0.0983$ \tabularnewline 
    {\tiny from online Appendix~L.3.4} & $1.39 \times 10^{-3}$ & $\mathbf{1.37 \times 10^{-3}}$ & $1.60 \times 10^{-3}$ & $2.08 \times 10^{-3}$ & $2.68 \times 10^{-3}$ \tabularnewline 
    \bottomrule
    \end{tabular}
    \caption{Test MSE for RR Boost trained on data from functions described in the stated figures, except that $n=1000$ and the noise standard deviation is equal to $0.3$ of that of $f$.}
    \label{tab:sphere-softthreshold-2}
\end{table}

\subsection{Real data: GISS}

\label{sec:GISS}

Here we apply these spherical methods to real data that lie on the
Earth. This example uses data from GISS Surface Temperature Analysis
\citep{GISTEMPv4,lenssen2024nasa}, which is an estimate
of global surface temperature change at various longitude-latitude
coordinates of the Earth. 
Using data from the year 2023, we aim to predict the temperature change on a testing 
set (7766 observations) by training on 1941 other observations. 
Figure~\ref{fig:sphere-mse-noaa} shows the ensembles' predictions 
and MSEs on the test data. 
Similarly to the simulated examples in Section~\ref{sec:sphere-numerical} in the main text and online Appendix~L.3,
Figure~\ref{fig:sphere-mse-noaa} shows that the smallest test MSE 
is produced by the Random-Rotation Boosted ensemble and appears to further decrease if the number of learners increases, whereas the other methods appear to stabilize in their predictive performance.
We also see that the prediction error appears
to be largest at land-ocean borders, particularly near central America
and northeastern Europe.

\begin{figure}[h!]
\centering \includegraphics[width=0.95\textwidth]{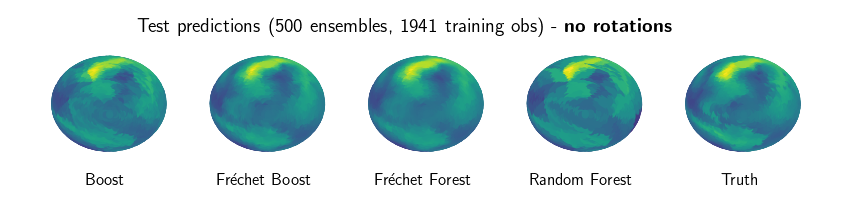}
\includegraphics[width=0.95\textwidth]{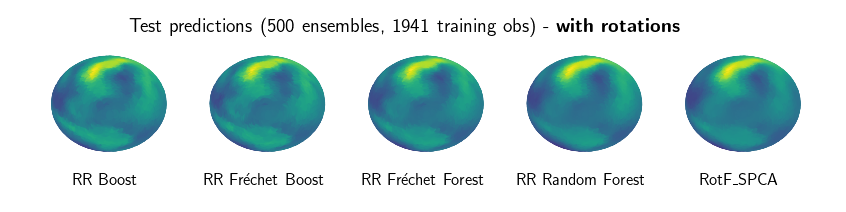}
\includegraphics[width=1\textwidth]{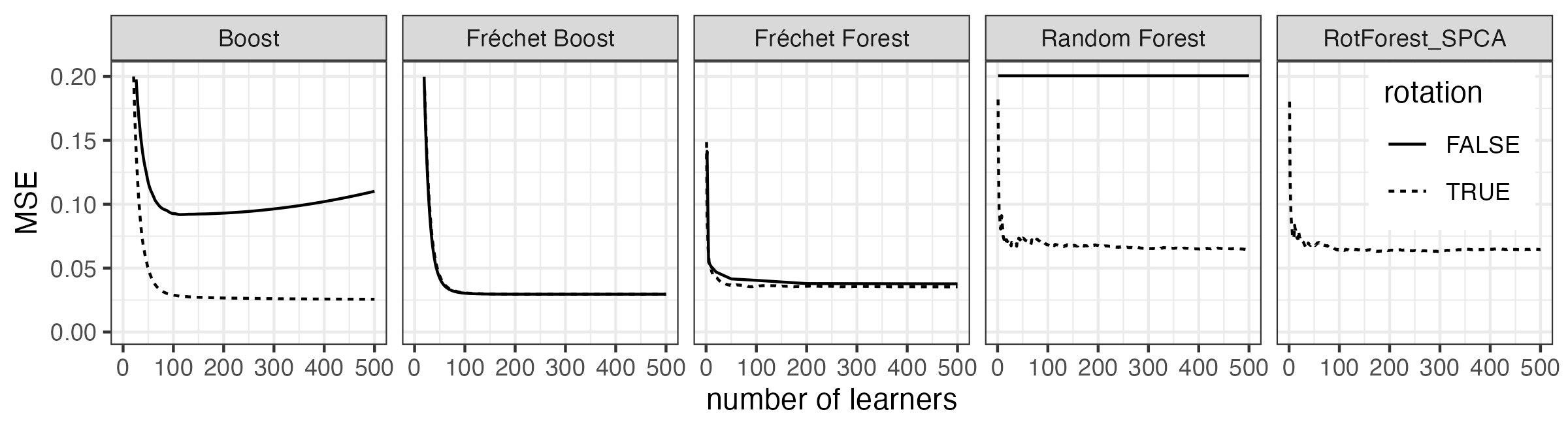}
\caption{Empirical comparison of tree ensemble models (\textquotedblleft RR\textquotedblright{}
refers to the method version that randomly rotates covariates before
fitting a base learner) trained on $n=1941$ points $(\bm{x},y)$ from the GISS dataset.
Trees are allowed to grow to full depth. 
Boosted methods use learning rate $0.05$. 
Top two rows: predictions of the ensemble methods. 
The 7766 test points are shown in the top-right sphere.
Third row: MSE on 7766 test points.}
\label{fig:sphere-mse-noaa} 
\end{figure}

\section{Conclusion}

\label{sec:conclusion}

This work develops a wavelet regression framework built on two
pillars. The first is an adaptive UHWT 
representation for regression whose transform is constructed directly
in $L^{2}(\mu_{n})$, so each admissible split is scored by its
empirical UH coefficient and the selected atoms remain orthonormal
(Section~\ref{sec:Wavelet-Regression}). 
The second pillar extends these constructions from regular grids to
higher-order tensors and triangulable manifolds
(Section~\ref{sec:Wavelet-Trees-general domain}). 
On spheres and general triangulated surfaces, we
replace axis-aligned splits by geodesic triangle refinements that
preserve local orthogonality in $L^{2}(\mu_{n})$. Random rotations
on $\mathbb{S}^{2}$ and intrinsic reparameterizations on general
manifolds generate diverse weak learners.

These two pillars come together in the ensemble constructions of
Section~\ref{sec:Wavelet-Regression-ensemble}. In the boosting
procedure, each weak learner is an UH tree iteratively fitted to residuals
and regularized by early stopping, a learning rate, 
and reconstruction-time soft thresholding. 
In the Bayesian backfitting procedure, additive UH trees with
RUHWT priors provide a fully probabilistic ensemble and uncertainty
quantification for triangulable manifolds \citep{luo2024multiple}, where each component sees only a partial residual and is shrunk by both its prior and its role in the sum.

Theoretical results in Section~\ref{sec:Theoretical-Results}
support these design choices. Oracle inequalities for fixed
partitions show that UH thresholding adapts to sparsity in the UH
coefficients, with risk controlled by the number of large coefficients
rather than the total number of leaves. 
A unified oracle bound covers both grids and triangulable
domains once orthonormality with respect to the underlying measure is
ensured. Together with the empirical results in
Section~\ref{sec:Experiments} (and in the Supplementary Materials), these findings suggest that UH wavelet
trees and their ensembles offer a practical and theoretically grounded
tool for multiscale regression on complex domains, complementing
classical grid-based wavelets and traditional tree ensembles.

\bibliography{multidim.bib}

\newpage

\begin{center}

{\Large\bf SUPPLEMENTARY MATERIAL}

\end{center}

\appendix

References from the main text will be preceded by ``MT.''

% \tableofcontents

\section{Related Algorithms}
\label{sec:Related-Algorithms}
\begin{algorithm}[h!]
{\scriptsize\rule[0.5ex]{1\columnwidth}{1pt}}\par
\begin{raggedright}
\textbf{Input:} Data $\{(\bm{x}_{i},y_i)\}_{i=1}^{n}$ on domain $\Omega\subset\mathbb{R}^{D}$.
Hyperparameters: maximum depth $L_{\max}$, minimum leaf size $n_{\min}\in\mathbb{N}$,
and contrast threshold $\tau$ (e.g.\ $\tau=\tau_{n,\delta}$).\\
\textbf{Output:} A UH tree $T$ (i.e., a recursive partition $\mathcal{P}(T)$) whose internal
nodes store a split $(\hat d,\hat\ell)$ and UH detail coefficients; reconstruction uses
(MT-5).
\end{raggedright}

\noindent\textbf{Notation.} For any node/cell $A\subset\Omega$, let
$I_A\coloneqq\{i:\bm{x}_i\in A\}$ and $N_A\coloneqq|I_A|$. For a split candidate
$(d,\ell)$ (coordinate $d\in\{1,\dots,D\}$ and threshold $\ell\in\mathbb{R}$), let
$A_l(d,\ell)$ and $A_r(d,\ell)$ be the two children induced by the split (as in
Section~MT-2.2), and let $w_{d,\ell}(A)$ be the UH contrast
defined in (3).

\noindent\textbf{Procedure.}
\begin{enumerate}
\item \emph{Initialize.} Set the root node $A_0\leftarrow\Omega$ and $T\leftarrow\{A_0\}$.
\item \emph{While} there exists a leaf $A\in T$ with $\mathrm{depth}(A)<L_{\max}$, process $A$ (e.g.\ breadth-first):
\begin{enumerate}
\item \emph{Admissible splits.} For each $d\in\{1,\dots,D\}$, form the set of admissible
thresholds $\Lambda_A(d)$ consisting of values of $\{x_{i,d}: i\in I_A\}$ whose split
produces children with $N_{A_l(d,\ell)}\ge n_{\min}$ and $N_{A_r(d,\ell)}\ge n_{\min}$.
\item \emph{Score candidates.} Compute $w_{d,\ell}(A)$ from (MT-3) for all admissible $(d,\ell)$.
\item \emph{Split or stop.} If no admissible $(d,\ell)$ exists or $\max_{(d,\ell)}|w_{d,\ell}(A)|<\tau$, mark $A$ as a leaf.
Otherwise choose and store
\begin{align}
(\hat d,\hat\ell)\in\arg\max_{(d,\ell)}|w_{d,\ell}(A)|,
\end{align}
split $A$ into $A_l(\hat d,\hat\ell)$ and $A_r(\hat d,\hat\ell)$, and add both children to $T$.
\end{enumerate}
\item \textbf{Return} $T$.
\end{enumerate}

\noindent\textbf{Prediction.} Reconstruct $\widehat f$ from the UH coefficients on $T$ using (MT-5).
{\scriptsize\rule[0.5ex]{1\columnwidth}{1pt}}\par
\caption{Optimization approach for fitting a single UHWT (aligned with (MT-3), (MT-4), and (MT-5)).}
\label{alg:uhwt-opt}
\end{algorithm}

\begin{algorithm}[h!]
\caption{\textsc{SpherePartitionerUH}$(X,r;\Theta)$: UH tree on $\mathbb{S}^{2}$ data (notation aligned with (MT-3) and (MT-18)).}
\label{alg:sphere-uh}

\scriptsize

{\scriptsize\rule[0.5ex]{1\columnwidth}{1pt}}\par
\begin{raggedright}
\textbf{Input:} Locations $X\in\mathbb{R}^{n\times 3}$ with rows on $\mathbb{S}^{2}$;
responses $r\in\mathbb{R}^{n}$ (e.g., residuals or observed responses).\\
\phantom{\textbf{Input:}} Hyperparameters $\Theta=(L_{\max},n_{\min},\tau,\mathtt{rule})$ with
$L_{\max}$ maximum depth, $n_{\min}$ minimum leaf size, $\tau\ge 0$ contrast threshold,
and $\mathtt{rule}\in\{\mathtt{balance},\mathtt{balance4},\mathtt{adapt},\mathtt{adapt\_vertex}\}$. Initial icosahedron.\\
\textbf{Output:} For all $j=1,\ldots,20$, a UH tree $T^{(j)}$ (i.e., a partition $\mathcal{P}(T^{(j)})$) on each face where $T^{(j)}$'s
internal nodes store split decisions $(\hat d,\hat\ell)^{(j)}$ and UH detail coefficients
(from (MT-3)).
\end{raggedright}

\noindent\textbf{Notation.} For a node/cell $A$, let $I_A\subset\{1,\dots,n\}$ be its index set,
$N_A\coloneqq|I_A|$, and $\bar r_A\coloneqq N_A^{-1}\sum_{i\in I_A} r_i$. For a candidate
$(d,\ell)$ with $d\in\{1,2,3\}$ and $\ell\in\mathbb{R}$, define children index sets
\begin{align}
I_{A_l(d,\ell)}
&\coloneqq \{\,i\in I_A:\ X_{i,d}\le \ell\,\},\qquad
I_{A_r(d,\ell)} \coloneqq I_A\setminus I_{A_l(d,\ell)},
\end{align}
and write $N_{A_\bullet}\coloneqq |I_{A_\bullet}|$ and $\bar r_{A_\bullet}$ for the corresponding
sizes and means. The UH contrast for $(A;d,\ell)$ is
\begin{align}
w_{d,\ell}(A)
&\coloneqq \sqrt{\frac{N_{A_l(d,\ell)}\,N_{A_r(d,\ell)}}{N_A}}
\bigl(\bar r_{A_l(d,\ell)}-\bar r_{A_r(d,\ell)}\bigr),
\end{align}
which matches (MT-3).

\noindent\textbf{Procedure.} For $j=1,\ldots,20$:
\begin{enumerate}
\item \emph{Initialize.} Denote by $A_0^{(j)}$ the $j$th face of the icosahedron. Set $T^{(j)}\leftarrow\{A_0^{(j)}\}$ and $\mathrm{depth}(A_0^{(j)})=0$. %
\item \emph{While} there exists a leaf $A\in T^{(j)}$ with $\mathrm{depth}(A)<L_{\max}$, process $A$:
\begin{enumerate}
\item \emph{Admissible thresholds.} For each $d\in\{1,2,3\}$, define $\Lambda_A(d)$ by
\begin{align}
\Lambda_A(d)\coloneqq
\begin{cases}
\{\text{empirical median of }\{X_{i,d}\}_{i\in I_A}\}, & \mathtt{rule}=\mathtt{balance},\\
\{\text{empirical quartiles of }\{X_{i,d}\}_{i\in I_A}\}, & \mathtt{rule}=\mathtt{balance4},\\
\{\text{midpoints of consecutive distinct values in }\{X_{i,d}\}_{i\in I_A}\}, & \mathtt{rule}=\mathtt{adapt},\\
\{\text{distinct values in }\{X_{i,d}\}_{i\in I_A}\}, & \mathtt{rule}=\mathtt{adapt\_vertex}.
\end{cases}
\end{align}
Discard any $\ell\in\Lambda_A(d)$ violating $N_{A_l(d,\ell)}\ge n_{\min}$ and $N_{A_r(d,\ell)}\ge n_{\min}$.
\item \emph{Score candidates.} Compute $w_{d,\ell}(A)$ for all admissible $(d,\ell)$.
\item \emph{Split or stop.} If no admissible $(d,\ell)$ exists or $\max_{(d,\ell)}|w_{d,\ell}(A)|<\tau$, mark $A$ as a leaf.
Otherwise choose and store
\begin{align}
(\hat d,\hat\ell)\in\arg\max_{(d,\ell)}|w_{d,\ell}(A)|,
\end{align}
split $A$ into $A_l(\hat d,\hat\ell)$ and $A_r(\hat d,\hat\ell)$, and add both children to $T^{(j)}$.
\end{enumerate}
\item \textbf{Return} $T^{(j)}$.
\end{enumerate}

\noindent\textbf{Prediction.} For $\bm{x}\in A_0^{(j)}$, use UH synthesis on $T^{(j)}$, as in (MT-5). (When used inside boosting,
$r$ is the stagewise residual and the resulting fit corresponds to the weak learner in (MT-16)--(MT-18).)
{\scriptsize\rule[0.5ex]{1\columnwidth}{1pt}}\par

\end{algorithm}

\textbf{Implementation notes for Algorithm~\ref{alg:sphere-uh}.}
Algorithm~\ref{alg:sphere-uh} is the greedy optimization procedure of Algorithm~\ref{alg:uhwt-opt}
specialized to $\Omega=\mathbb{S}^{2}$ and evaluated in rotated ambient coordinates. Candidate splits
remain axis-aligned in the rotated coordinate system (i.e., $X_{i,d}\le \ell$), producing spherical caps/belts
$A_l(d,\ell)$ and $A_r(d,\ell)$ within the current cell. The only implementation choice beyond the abstract
description is the candidate set $\Lambda_A(d)$ (controlled by \texttt{rule}); the split score and stopping rule
use the same UH contrasts and threshold test $\max_{(d,\ell)}|w_{d,\ell}(A)|<\tau$ as in the main text.

\begin{algorithm}[h!]
\caption{Bayesian MCMC sampler for a single RUHWT (image-specialized; notation aligned with Section~MT-2.3 and with boosting learning rate $\eta$).}
\label{alg:bayesian-rudp}

\scriptsize

{\scriptsize\rule[0.5ex]{1\columnwidth}{1pt}}\par
\begin{raggedright}
\textbf{Input:} Image array $Y\in\mathbb{R}^{n_r\times n_c}$ on $\Omega\subset\mathbb{R}^2$;
maximum depth $L_{\max}$; iterations $K$; depth-based split prior
$p_{\mathrm{split}}(d)=0.99\cdot 0.499^{\,d}$; noise level $\sigma$ (MAD estimate if not supplied).\\
\textbf{Output:} A UH tree $T$ (last accepted state) whose internal nodes store split $(d,\ell)$ and UH details.
\end{raggedright}

\noindent\textbf{Notation.} For a node $A$, let $\mathrm{depth}(A)$ be its depth. When $A$ is internal with stored
$(d,\ell)$, write $w_{d,\ell}(A)$ for its UH contrast (as in (MT-3), specialized to the image grid).

\noindent\textbf{Initialization.} Set root $A_0=\Omega$, $\mathrm{depth}(A_0)=0$, and initialize $T\leftarrow\{A_0\}$.

\noindent\textbf{Local scores on a subtree rooted at $A$.}
\begin{itemize}
\item \emph{Local prior:}
\begin{align}
\log\pi_{\mathrm{prior}}(A)
\;=\;
\sum_{U\in\mathrm{Desc}(A)\cup\{A\}}
\begin{cases}
\log\!\bigl(1-p_{\mathrm{split}}(\mathrm{depth}(U))\bigr), & U\ \text{leaf},\\
\log\!\bigl(p_{\mathrm{split}}(\mathrm{depth}(U))\bigr), & U\ \text{internal}.
\end{cases}
\end{align}
\item \emph{Local likelihood (Gaussian):}
\begin{align}
\log\mathcal{L}(A)
\;=\;
\sum_{\substack{U\in\mathrm{Descendants of }(A)\cup\{A\}\\ U\ \mathrm{internal}}}\log\phi_\sigma\!\bigl(w_{d_U,\ell_U}(U)\bigr)
\;-\;
\sum_{\substack{U\ \mathrm{leaf}}}\frac{\mathrm{SSE}(U)}{2\sigma^2},
\end{align}
where $\phi_\sigma$ is the $N(0,\sigma^2)$ density and $\mathrm{SSE}(U)$ is the within-leaf sum of squares about the leaf mean.
\item \emph{Local log-posterior:} $\log\Pi(A)\coloneqq\log\pi_{\mathrm{prior}}(A)+\log\mathcal{L}(A)$.
\end{itemize}

\noindent\textbf{One MCMC step (repeat for $t=1,\dots,K$).}
\begin{enumerate}
\item Choose move type $\texttt{move}\in\{\textsf{GROW},\textsf{PRUNE}\}$ uniformly, for candidate sets:
\begin{align}
\mathcal{C}_{\textsf{GROW}}(T)=\{U:\ U\text{ leaf and }\mathrm{depth}(U)<L_{\max}\},\\
\mathcal{C}_{\textsf{PRUNE}}(T)=\{U:\ U\text{ internal and both children are leaves}\}.
\end{align}
If $\mathcal{C}_\texttt{move}(T)=\varnothing$, switch \texttt{move} to the other move.
\item Pick target node $A$ uniformly from $\mathcal{C}_\texttt{move}(T)$ and compute $L_{\mathrm{cur}}=\log\Pi(A)$.
\item Propose a local change at $A$ (We implemented \textsf{SWAP} option as well):
\begin{itemize}
\item If $\texttt{move}=\textsf{GROW}$: propose a split $(d,\ell)$ for $A$ (image-row or image-column direction),
create two children, and assign their depths $\mathrm{depth}(A)+1$.
\item If $\texttt{move}=\textsf{PRUNE}$: remove both children of $A$ and mark $A$ as a leaf (clear $(d,\ell)$).
\end{itemize}
\item Compute $L_{\mathrm{prop}}=\log\Pi(A)$ after the change and the Hastings correction:
\begin{align}
\Delta_{\mathrm{prop}}
=
\begin{cases}
\log\!\dfrac{|\mathcal{C}_{\textsf{GROW}}(T)|}{|\mathcal{C}_{\textsf{PRUNE}}(T')|}, & \texttt{move}=\textsf{GROW},\\[0.2cm]
\log\!\dfrac{|\mathcal{C}_{\textsf{PRUNE}}(T)|}{|\mathcal{C}_{\textsf{GROW}}(T')|}, & \texttt{move}=\textsf{PRUNE},
\end{cases}
\end{align}
where $T'$ is the proposed tree (used only to count candidates).
\item Metropolis--Hastings accept with probability
$\min\{1,\exp(\Lambda)\}$, $\Lambda=(L_{\mathrm{prop}}-L_{\mathrm{cur}})+\Delta_{\mathrm{prop}}$.
\end{enumerate}

\noindent\textbf{Return.} The final accepted $T$ (and optionally the stored trace of trees).
{\scriptsize\rule[0.5ex]{1\columnwidth}{1pt}}\par

\end{algorithm}

\begin{algorithm}[h!]
{\scriptsize\rule[0.5ex]{1\columnwidth}{1pt}}\par
\begin{raggedright}
\textbf{Input:} %
Training data $\{(y_i, \bm{x}_i)\}_{i=1}^n$ with $\bm{x}_i\in\mathbb{S}^2$ and $y_i \in \mathbb{R}$; 
ensemble size $m$; base-tree parameters $\Theta$ for Algorithm~\ref{alg:sphere-uh}. Initial icosahedron.\\
\phantom{\textbf{Input:}} Rotation generator $\mathtt{HaarSO}(3)$ (Gaussian--QR with sign fix; see below).\\
\textbf{Output:} Rotations $\{R_t\}_{t=1}^m\subset\mathrm{SO}(3)$ and UH trees $\{(T_t^{(1)}, \ldots, T_t^{(20)})\}_{t=1}^{m}$, where for each $j=1,\ldots,20$ the superscript $(j)$ denotes the $j$th face of the icosahedron.
\end{raggedright}

\noindent\textbf{Procedure.}
\begin{enumerate}
    \item For $t=1,\dots,m$:
    \begin{enumerate}
    \item Draw $R_t\sim \mathtt{HaarSO}(3)$.
    \item Fit $(T_t^{(1)}, \ldots, T_t^{(20)}) \leftarrow \textsc{SpherePartitionerUH}\big(\{R_m\bm{x}_i\}_{i=1}^n,\,(y_i)_{i=1}^n;\Theta\big)$.
    \end{enumerate}
    \item \textbf{Return} $\{R_t\}_{t=1}^m$ and $\{(T_t^{(1)}, \ldots, T_t^{(20)})\}_{t=1}^m$.
\end{enumerate}
\noindent\textbf{Prediction.} For $\bm{x}\in\mathbb{S}^2$,
\begin{align}
\widehat f(\bm{x})
\;=\;
\frac{1}{m}\sum_{t=1}^m \widehat f_{T_t^{(j)}}(R_t\bm{x}),
\end{align}
where $j=j_{R_t}$ is the index of the icosahedron's face containing $R_t\bm{x}$, and $\widehat f_{T_t^{(j)}}$ denotes the UH reconstruction on the tree $T_t^{(j)}$  (cf.\ (MT-5)).
{\scriptsize\rule[0.5ex]{1\columnwidth}{1pt}}\par
\caption{RRE: Random Rotation Ensemble with UH base trees on $\mathbb{S}^{2}$ (notation aligned with Section 4).}
\label{alg:rre}
\end{algorithm}

\begin{algorithm}[h!]
{\scriptsize\rule[0.5ex]{1\columnwidth}{1pt}}\par
\begin{raggedright}
\textbf{Input:} $(\{\bm{x}_i\}_{i=1}^n,y)$ with $\bm{x}_i\in\mathbb{S}^2$;
stages $G$; learning rate $\eta\in(0,1]$; base parameters $\Theta$;
rotation generator $\mathtt{HaarSO}(3)$.\\
\textbf{Output:} Predictor $\widehat f_G$ of the form (MT-18).
\end{raggedright}

\noindent\textbf{Procedure.}
\begin{enumerate}
\item Initialize $f_0(\bm{x})\equiv \bar y_\Omega$ and residuals $r_{0,i}\coloneqq y_i-f_0(\bm{x}_i)$.
\item For $t=1,\dots,G$:
\begin{enumerate}
\item Draw $R_t\sim \mathtt{HaarSO}(3)$.
\item Fit a UHWT weak learner on rotated inputs:
\begin{align}
T_t \leftarrow \textsc{SpherePartitionerUH}\big(\{R_t\bm{x}_i\}_{i=1}^n,\,\{r_{t-1,i}\}_{i=1}^n;\Theta\big),
\end{align}
and let $g_t(\bm{x})$ be the corresponding UH reconstruction on $T_t$ evaluated at $R_t\bm{x}$ (cf.\ (MT-16)--(MT-18)).
\item Update
\begin{align}
f_t(\bm{x}) \leftarrow f_{t-1}(\bm{x}) + \eta\, g_t(\bm{x}),\qquad
r_{t,i}\leftarrow y_i-f_t(\bm{x}_i).
\end{align}
\end{enumerate}
\item \textbf{Return} $f_G$.
\end{enumerate}
{\scriptsize\rule[0.5ex]{1\columnwidth}{1pt}}\par
\caption{RRB: Random Rotation Boosting with UH weak learners on $\mathbb{S}^{2}$ (matched to (MT-15)--(MT-18)).}
\label{alg:rrb}
\end{algorithm}

\paragraph{Haar-uniform rotations on $\mathrm{SO}(3)$.}
To draw $R\sim\mathrm{Haar}(\mathrm{SO}(3))$, sample $G\in\mathbb{R}^{3\times 3}$ with i.i.d.\ $\mathcal{N}(0,1)$
entries, compute the QR factorization $G=Q\,R_{\mathrm{up}}$ with $Q^\top Q=I_3$, and enforce $Q\in\mathrm{SO}(3)$ by a sign fix:
\begin{align}
\widetilde Q \;=\; Q\cdot \mathrm{diag}\bigl(1,1,\mathrm{sign}(\det Q)\bigr),\qquad \det(\widetilde Q)=+1.
\end{align}
Then $\widetilde Q$ is Haar on $\mathrm{SO}(3)$.
 
If a base tree is fit on rotated locations $\{R\bm{x}_i\}$, then test points must be mapped consistently to $R\bm{x}$
before evaluating the UH reconstruction on that tree, matching (MT-18).

\section{Orthonormal UH Basis}
\label{sec:UH_orthonormal}
Let $\mu_n = \frac{1}{n}\sum_{i=1}^n \delta_{\bm x_i}$ be the empirical measure,
and let $\mathcal{T}^\circ$ be a finite binary partition of $\Omega$ with root cell $A_0$.
For each internal node $A \in \mathcal{T}^\circ$, let $A_L$ and $A_R$ denote its left and right children, and write
$\mu_n(A) = \int_A \mathrm d\mu_n = \frac{1}{n} \sum_{i=1}^n \mathbf{1}\{\bm x_i \in A\}$ for the empirical mass of $A$.

Define the scaling function and wavelet functions by
\begin{align}
\phi_{A_0}(\bm x)
&:= \frac{\mathbf{1}\{\bm x \in A_0\}}{\sqrt{\mu_n(A_0)}},
\label{eq:uh-scaling-def}
\\
\psi_A(\bm x)
&:= \sqrt{\frac{\mu_n(A_L)\,\mu_n(A_R)}{\mu_n(A)}}
\left(
  \frac{\mathbf{1}\{\bm x \in A_L\}}{\mu_n(A_L)}
  - \frac{\mathbf{1}\{\bm x \in A_R\}}{\mu_n(A_R)}
\right),
\qquad A \in \mathcal{T}^\circ.
\label{eq:uh-wavelet-def}
\end{align}
We equip $L^2(\mu_n)$ with the inner product
$\langle f,g\rangle_n := \int f(\bm x) g(\bm x)\, \mathrm d\mu_n(\bm x)$.
\begin{lemma}[Orthonormality of the UH system]
\label{lem:uh-orthonormal}
The collection
\[
\mathcal{U} := \{\phi_{A_0}\} \cup \{\psi_A : A \in \mathcal{T}^\circ\}
\]
is an orthonormal system in $L^2(\mu_n)$.
\end{lemma}

\begin{proof}
We first check that each element has unit $L^2(\mu_n)$ norm and then verify pairwise
orthogonality.

\emph{Norms.}
For the scaling function,
\[
\|\phi_{A_0}\|_n^2
= \int \phi_{A_0}^2(\bm x)\,\mathrm d\mu_n(\bm x)
= \frac{1}{\mu_n(A_0)}\int \mathbf{1}\{\bm x \in A_0\}\,\mathrm d\mu_n(\bm x)
= \frac{\mu_n(A_0)}{\mu_n(A_0)} = 1.
\]
Fix an internal node $A \in \mathcal{T}^\circ$ with children $A_L,A_R$.
By definition \eqref{eq:uh-wavelet-def}, $\psi_A$ is supported on $A_L \cup A_R$
and takes the constant values
\[
c_L := \sqrt{\frac{\mu_n(A_L)\,\mu_n(A_R)}{\mu_n(A)}}\,\frac{1}{\mu_n(A_L)}
\quad\text{on }A_L,
\qquad
c_R := -\sqrt{\frac{\mu_n(A_L)\,\mu_n(A_R)}{\mu_n(A)}}\,\frac{1}{\mu_n(A_R)}
\quad\text{on }A_R.
\]
Hence
\begin{align*}
\|\psi_A\|_n^2
&= \int \psi_A^2(\bm x)\,\mathrm d\mu_n(\bm x) \\
&= \int_{A_L} c_L^2\,\mathrm d\mu_n(\bm x)
   + \int_{A_R} c_R^2\,\mathrm d\mu_n(\bm x) \\
&= c_L^2\,\mu_n(A_L) + c_R^2\,\mu_n(A_R) \\
&= \frac{\mu_n(A_L)\,\mu_n(A_R)}{\mu_n(A)}
    \left(\frac{1}{\mu_n(A_L)} + \frac{1}{\mu_n(A_R)}\right) \\
&= \frac{\mu_n(A_R)}{\mu_n(A)} + \frac{\mu_n(A_L)}{\mu_n(A)} \\
&= 1,
\end{align*}
since $\mu_n(A) = \mu_n(A_L)+\mu_n(A_R)$.

\emph{Orthogonality with the scaling function.}
Each wavelet $\psi_A$ has zero $\mu_n$-mean:
\begin{equation} \label{eq:zero-mean}
   \int \psi_A(\bm x)\,\mathrm d\mu_n(\bm x)
   = c_L\,\mu_n(A_L) + c_R\,\mu_n(A_R) 
   = 0. 
\end{equation}
Therefore
\[
\langle \phi_{A_0},\psi_A\rangle_n
= \frac{1}{\sqrt{\mu_n(A_0)}}
  \int_{A_0} \psi_A(\bm x)\,\mathrm d\mu_n(\bm x)
= 0.
\]

\emph{Orthogonality between wavelets.}
Let $A,B \in \mathcal{T}^\circ$, $A \neq B$.
Because $\mathcal{T}^\circ$ is a binary tree, nodes $A$ and $B$ are either disjoint,
or one node is strictly contained in the other.

If $A$ and $B$ are disjoint, then
$\psi_A \psi_B \equiv 0$ and hence $\langle \psi_A,\psi_B\rangle_n = 0$.
Suppose instead that $B \subset A$ (the case $A \subset B$ is symmetric).
Since $B$ lies entirely inside one child of $A$, say $B \subseteq A_L$,
the wavelet $\psi_A$ is constant on $B$ with value $c_L$ as above.
Thus
\[
\langle \psi_A,\psi_B\rangle_n
= \int_B \psi_A(\bm x)\psi_B(\bm x)\,\mathrm d\mu_n(\bm x)
= c_L \int_B \psi_B(\bm x)\,\mathrm d\mu_n(\bm x) = 0,
\]
because $\psi_B$ has zero $\mu_n$-mean on its support, as shown above in \eqref{eq:zero-mean}.

This exhausts all possibilities, so every pair of distinct elements of $\mathcal{U}$
is orthogonal, and each has unit norm. Therefore $\mathcal{U}$ is orthonormal
in $L^2(\mu_n)$.
\end{proof}

\section{Connection to CART}
\label{sec:connection-to-CART}

This section explores the similarity between CART and the optimization approach to fitting a UHWT in the grid case.
We first consider how the splits are chosen.
On a grid, and when $|\cdot|$ in (MT-3) counts
design points, this reduction in squared error matches the CART
criterion up to constants. The greedy rule
(MT-4) therefore coincides with the usual CART
split rule. 
For a candidate split $s\colon A\to(A_{1},A_{2})$, the empirical
squared-error decrease for a leafwise estimator is
\begin{align}
\Delta_{Y}(s)
&=
\sum_{i: X_{i}\in A}(Y_{i}-\bar{Y}_{A})^{2}
-
\sum_{i: X_{i}\in A_{1}}(Y_{i}-\bar{Y}_{A_{1}})^{2}
-
\sum_{i: X_{i}\in A_{2}}(Y_{i}-\bar{Y}_{A_{2}})^{2}
\nonumber\\
&=
\frac{N_{A_{1}}N_{A_{2}}}{N_{A}}
(\bar{Y}_{A_{1}}-\bar{Y}_{A_{2}})^{2},
\label{eq:delta-cart}
\end{align}
which, on a grid, is exactly the square of the UH coefficient
$|w_{d,\ell}(A)|$ used in the optimization rule
(MT-4). Thus the greedy UH split and the greedy
CART split share the same basic objective.

The identity \eqref{eq:delta-cart} shows that UH splitting is fully
aligned with CART when the same candidate set is used: both rules
pick the split that maximizes the drop in empirical risk, and the UH
coefficient is simply a more convenient parameterization of that
drop. 

In the analysis of \cite{donoho1997cart}, CART
appears as a method for selecting an orthonormal basis from a
predefined library. Here the same equality between reduction in
squared error and squared UH coefficient is used in the opposite
direction: the CART partition is interpreted as a sequence of UH
atoms and thereby defines an orthogonal wavelet system in
$L^{2}(\mu_{n})$ on the empirical design. Compared with the
univariate UH construction in \citet{fryzlewicz2007unbalanced}, our
optimization approach works in multiple dimensions, allows continuous
split locations, and links the UH coefficients directly to the
regression fit of a multivariate tree.

However, the reconstruction is different due to the regularization (MT-7).
If the wavelet coefficients are not regularized (i.e., if $\tilde{w}_{d,\ell}(A)=w_{d,\ell}(A)$ for all nodes $A$), 
then (MT-5) yields the orthogonal projection of
$y$ onto the span of the selected UH atoms in $L^{2}(\mu_{n})$. 

\begin{theorem}
\label{thm:grid-coincide}
Let $(\bm{x}_{i},y_{i})$ for $i=1,\dots,n$ be training data.
Let $\mathcal{Q}$ be a binary-tree partition of $\Omega$. For
$\bm{x}\in\Omega$ let $A_{\bm{x}}$ denote the leaf of $\mathcal{Q}$
that contains $\bm{x}$ and define
$\hat{f}_{\mathcal{Q}}^{\text{leaf}}(\bm{x})
=
n_{A_{\bm{x}}}^{-1}\sum_{\bm{x}_{i}\in A_{\bm{x}}}y_{i}$,
where $n_{A_{\bm{x}}}$ is the number of training locations in
$A_{\bm{x}}$. If $\Omega$ is a grid and $|\cdot|$ in
(MT-3) equals the number of training locations in
its argument, then $\hat{f}_{\mathcal{Q}}^{\text{leaf}}$ equals
$\hat{f}_{\mathcal{Q}}^{\text{UH}}$ defined in
(MT-5) when $\tilde{w}_{d,\ell}(A)=w_{d,\ell}(A)$.
\end{theorem}

Therefore,  
the main difference between our optimization approach and tree-based regression algorithms such as CART %
lies in how the empirical estimate is regularized. 
CART regularizes by pruning subtrees, which can be interpreted as shrinking all
coefficients in those subtrees to zero before applying (MT-5). 
In contrast, our approach regularizes by thresholding UH coefficients 
at each node using a threshold calibrated by a MAD estimate of the noise. 

To motivate the use of wavelet regularization,
we note that existing consistency results for
CART \citep{scornet2015consistency,wager2018estimation,klusowski2024large}
often rely on no-interaction or smoothness conditions on $f$ that are not
well matched to the piecewise-constant and edge-driven structure of
images and tensors. In contrast, UH-based estimators admit
consistency for broader function classes in one dimension and for
certain multiscale constructions
\citep{fryzlewicz2007unbalanced,kolaczyk2005multiscale}. 
Empirically, \cite{agarwal2022hierarchical} also find (for non-grid domains) that CART's predictive ability often improves when additional ($L_2$) shrinkage is applied to an existing CART tree, even when it has already been pruned.

\section{Proof of Theorem~\ref{thm:grid-coincide}}
\label{sec:Proof-of-CART-coincide} 
\begin{proof}
We will prove this claim by induction.
For the base case, we have $\hat{f}_{\mathcal{Q}}^{\text{leaf}} = \hat{f}_{\mathcal{Q}}^{\text{UH}}$ when $\mathcal{Q} = \{\Omega\}$, in which case the functions $\hat{f}_{\mathcal{Q}}^{\text{leaf}}$ and $\hat{f}_{\mathcal{Q}}^{\text{UH}}$ have constant value $n^{-1}_\Omega \sum_{\bm{x}_i \in \Omega} y_i$, i.e., the sample mean of the observed response values.

Now suppose $\mathcal{P}$ is a binary-tree partition such that $\hat{f}_{\mathcal{P}}^{\text{leaf}} = \hat{f}_{\mathcal{P}}^{\text{UH}}$.
For a leaf node $A \in \mathcal{P}$, suppose $\mathcal{Q}$ is a binary-tree partition that is exactly $\mathcal{P}$ except that $A$ is an internal node in $\mathcal{Q}$ with leaf nodes as children.
By this assumption, we have
\begin{align}
    \hat{f}_{\mathcal{Q}}^{\text{leaf}}(\bm{x}) \;=\; \hat{f}_{\mathcal{P}}^{\text{leaf}}(\bm{x}) \;=\; \hat{f}_{\mathcal{P}}^{\text{UH}}(\bm{x}) \;=\; \hat{f}_{\mathcal{Q}}^{\text{UH}}(\bm{x}) \quad\text{for any } \bm{x} \notin A.
\end{align}
Now denote the left and right child of $A$ as $A_L$ and $A_R$, respectively.
We have 
\begin{align}\label{eq:function-difference}
    \hat{f}_{\mathcal{Q}}^{\text{UH}}(\bm{x}) - \hat{f}_{\mathcal{Q}}^{\text{leaf}}(\bm{x})
    \;=\; n_L \bar{y}_L\Bigg[-\frac{n_R}{n_L} \frac{1}{n_A} + \frac{|A_R|}{|A_L|}\frac{1}{|A|}\Bigg] + n_R \bar{y}_R \Bigg[\frac{1}{n_A} - \frac{1}{|A|}\Bigg] \quad\text{for any } \bm{x} \in A_L.
\end{align}
The RHS vanishes if $|\cdot|$ is the number of training locations in the argument $\cdot$, in which case $\hat{f}_{\mathcal{Q}}^{\text{leaf}}(\bm{x}) = \hat{f}_{\mathcal{Q}}^{\text{UH}}(\bm{x})$ for any $\bm{x} \in A_L$.
Similar reasoning will imply that $\hat{f}_{\mathcal{Q}}^{\text{leaf}}(\bm{x}) = \hat{f}_{\mathcal{Q}}^{\text{UH}}(\bm{x})$ for any $\bm{x} \in A_R$.
Therefore, $\hat{f}_{\mathcal{Q}}^{\text{leaf}} = \hat{f}_{\mathcal{Q}}^{\text{UH}}$.
\end{proof}

\section{Posterior distribution with latent states}
\label{sec:posterior-latent}

The following theorem, Theorem~\ref{thm:posterior}, extends Theorem~MT-2.1 to incorporate latent states.
(Our Theorem~\ref{thm:posterior} is also written in a way that emphasizes that a tree can be grown breadth-first or depth-first.)
For tractability, the distribution of the latent states is modeled as a top-down Markov tree so that
the state of a node can depend on the state of its parent but not on the state of any other node. 
For example, an early-stopping mechanism can be used: a node at depth $j$ is no longer allowed to split with probability $\rho_j$, where $\rho_j=1$ if its parent is not allowed to split and $\rho_j<1$ if its parent is allowed to split; 
this mechanism can greatly reduce the bottom-up computation of $\Phi_{s}(A)$ in Theorem~\ref{thm:posterior}.
Alternatively, a spike-and-slab framework enables a less stringent regularization approach:
if a node is assigned to be a spike in order to indicate that its UH coefficient will be omitted when constructing the signal, then the probability that a node is a spike is allowed to depend on whether or not the node's parent is a spike.

\begin{theorem}
\label{thm:posterior}
Suppose $T$ has a RUHWT prior with split dimension probabilities
$\{\lambda_{d}(A)\colon A\in\mathcal{A},d\in\mathcal{D}(A)\}$ and
split location distributions $\{B_{A,d}\colon A\in\mathcal{A},d\in\mathcal{D}(A)\}$
on $[0,1]$, with zero probability of creating a child with no
training locations. 
Suppose that, given $T$ and hyperparameters $\bm{\phi}$, the coefficients are
conditionally independent given latent states
$\mathcal{S}=\{S_{i} \in \{1,\ldots,K\}\colon A_{i} \in T\}$:
\begin{align}
(w_{i},z_{i})\mid S(A_{i})=s
\;\ind\;p_{i}^{(s)}(w,z\mid\bm{\phi}),\quad A_{i} \in T.
\end{align}
Assume that the distribution of the latent states follows a top-down
Markov tree whose transition kernels 
$P(S(A)=s'\mid S(A_p)=s) \eqqcolon \rho_{j}(s,s')$, where $A_p$ is $A$'s parent, can 
depend on $A$'s depth $j$ in the tree $T$. Then the joint marginal posterior of
$(T,\mathcal{S})$ admits the following sequential generative description. 

We begin with $k=0$ and $T^{(0)} = \{\Omega\}$.
Given a tree $T^{(k)}$, a node $A \in T^{(k)}$ whose latent state has not yet been generated, and the latent state of $A$'s parent $A_p$, 
the state of $A$ is drawn from
\begin{align}
P(S(A)=s'\mid S(A_{p})=s,T^{(k)},\bm{y})
&=
\rho_{j}(s,s')
\sum_{d\in D(A)}\lambda_{d}(A)\int_0^1
M_{d,\ell}^{(s')}(A)\,
\frac{\Phi_{s'}(A_{l}^{(d,\ell)})\Phi_{s'}(A_{r}^{(d,\ell)})}{\Phi_{s}(A)}\,\mathrm{d}B_{A,d}(\ell),\label{eq:posterior_state}
\end{align}
where $j$ is the depth of $A$.
Given $S(A)=s'$ and $j<J$,
the split dimension and location of $A$ are drawn from
\begin{align}
P(D(A)=d\mid S(A)=s',T^{(k)},\bm{y})
&\propto
\lambda_{d}(A)\int_0^1
M_{d,\ell}^{(s')}(A)\,
\Phi_{s'}(A_{l}^{(d,\ell)})\Phi_{s'}(A_{r}^{(d,\ell)})\,\mathrm{d}B_{A,d}(\ell),\label{eq:posterior_dimension_latent}\\
\mathrm{d}B_{A,d}(\ell\mid D(A)=d,S(A)=s',T^{(k)},\bm{y})
&\propto
\mathrm{d}B_{A,d}(\ell)\,
M_{d,\ell}^{(s')}(A)\,
\Phi_{s'}(A_{l}^{(d,\ell)})\Phi_{s'}(A_{r}^{(d,\ell)}),\label{eq:posterior_location_latent}
\end{align}
where
$M_{d,\ell}^{(s)}(A_i)\coloneqq\int p_{i}^{(s)}(w_{d,\ell}(A_i),z\mid\bm{\phi})\mathrm{d}z$,
and the state-specific marginal likelihood satisfies
\begin{align}
\Phi_{s}(A)
=
\sum_{s'}\rho_{j}(s,s')
\sum_{d\in D(A)}\lambda_{d}(A)\int_0^1
M_{d,\ell}^{(s')}(A)\,
\Phi_{s'}(A_{l}^{(d,\ell)})\Phi_{s'}(A_{r}^{(d,\ell)})\,\mathrm{d}B_{A,d}(\ell)\label{eq:phi2_latent}
\end{align}
if $A$ is not atomic and $\Phi_{s}(A)=1$ if $A$ is atomic.
\end{theorem}

\section{Proof of Theorems MT-2.1 and \ref{thm:posterior}}

Theorem~MT-2.1 is a special case of Theorem~\ref{thm:posterior} and hence will be proven using the following proof of Theorem~\ref{thm:posterior}.

\label{sec:Proof-of-Posterior} 
\begin{proof}
For notational simplicity, let $T = T^{(k)}$.
This proof will be a slight modification of the proof of Theorem~2
in \citet{li2021learning}. The key difference lies in what is meant
by an atom, which we define as any measurable subset of $\Omega$
that contains exactly one training location. Under this definition,
if $A$ is atomic, then $\Phi_{s}(A)=P(\bm{y}(A)|A\in T,S(A_{p})=s)=1$.
Now, for a node $A$, suppose we have shown that $\Phi_{s}(B)=P(\bm{y}(B)|B\in T,S(B_{p})=s)$
for all possible descendants $B$ of $A$.
If $A$ is at depth $j$, it follows that 
\begin{align*}
& P(\bm{y}(A)|A\in T,S(A_{p})=s) \\
& =\sum_{s'}\sum_{d\in D(A)}\int_{\ell}P(\bm{y}(A)|A\in T,S(A)=s',S(A_{p})=s,D(A)=d,L(A)=\ell)\\
 & \quad\times\underset{\eqqcolon\rho_{j}(s,s')}{\underbrace{P(S(A)=s'| A\in T,S(A_{p})=s)}}\times\underset{\eqqcolon\lambda_{d}(A)}{\underbrace{P(D(A)=d| A\in T,S(A_{p})=s)}}\times\mathrm{d}B_{A,d}(\ell)\\
 & =\sum_{s'}\rho_{j}(s,s')\sum_{d}\lambda_{d}(A)\int_{\ell}M_{d,\ell}^{(s')}(A)\Phi_{s'}(A_{l}^{(d,\ell)})\Phi_{s'}(A_{r}^{(d,\ell)})\,\mathrm{d}B_{A,d}(\ell),
\end{align*}
which leads to the definition of $\Phi_{s}(A)$ in Theorem~\ref{thm:posterior}.

The remainder of the proof follows by modifying the proof of Theorem~2
in \citet{li2021learning} by appropriately including split-location
probabilities in the relevant places in the proof. 

Now we prove the claim \eqref{eq:posterior_state}.
We have
\begin{align*}
    &P(S(A)=s', D(A)=d, L(A)=\ell, \bm{y}(A) \mid S(A_p)=s, T) \\    
    &=\underset{=\rho_j(s,s') \lambda_d(A) B_{A,d}(\ell)}{\underbrace{P(S(A)=s', D(A)=d, L(A)=\ell | S(A_p)=s, T)}} \\
    & \times \underset{=M_{d,\ell}^{(s')}(A) \Phi_{s'}(A_l^{(d,\ell)}) \Phi_{s'}(A_r^{(d,\ell)})}{\underbrace{P(\bm{y}(A) | S(A_p)=s, T, S(A)=s', D(A)=d, L(A)=\ell)}}.
\end{align*}
Marginalizing the above over $d, \ell$, we get 
\begin{align*}
    P(S(A)=s', \bm{y}(A) \mid S(A_p)=s, T) 
    &= \rho_j(s,s') \sum_{d \in D(A)} \lambda_d(A) \int_0^1 M_{d,\ell}^{(s')}(A) \Phi_{s'}(A_l^{(d,\ell)}) \Phi_{s'}(A_r^{(d,\ell)}) \,\mathrm{d}B_{A,d}(\ell).
\end{align*}
From this and the identity 
\begin{align*}
    P(S(A)=s' \mid S(A_p)=s, T, \bm{y}) 
    = \frac{P(S(A)=s', \bm{y}(A) \mid S(A_p)=s, T) }{P(\bm{y}(A) \mid S(A_p)=s, T)}
\end{align*}
where the denominator is just $\Phi_s(A)$,
we get the claim \eqref{eq:posterior_state}.

Finally, we prove the remaining claims \eqref{eq:posterior_dimension_latent} and \eqref{eq:posterior_location_latent}.
We have 
\begin{align*}
    P(D(A)=d, L(A)=\ell \mid S(A)=s', T, \bm{y}) 
    &= \frac{P(D(A)=d, L(A)=\ell, \bm{y}(A) \mid S(A)=s', T)}{P(\bm{y}(A) \mid S(A)=s', T)}\\
    &\propto \lambda_d(A) B_{A,d}(\ell) M_{d,\ell}^{(s')}(A) \Phi_{s'}(A_l^{(d,\ell)}) \Phi_{s'}(A_r^{(d,\ell)}).
\end{align*}
\eqref{eq:posterior_dimension_latent} follows by marginalizing the above over $\ell$.
\eqref{eq:posterior_location_latent} follows from 
\begin{align*}
    P(L(A)=\ell \mid D(A)=d, S(A)=s', T, \bm{y})
    &= \frac{P(D(A)=d, L(A)=\ell \mid S(A)=s', T, \bm{y}) }{P(D(A)=d \mid S(A)=s', T, \bm{y})}\\
    &\propto B_{A,d}(\ell) M_{d,\ell}^{(s')}(A) \Phi_{s'}(A_l^{(d,\ell)}) \Phi_{s'}(A_r^{(d,\ell)}).
\end{align*}

\end{proof}

\section{Proofs for Function Reconstruction}

\label{sec:reconstruction}

\subsection{Euclidean}

\label{sec:reconstruction-euclidean}

The following theorem provides an exact reconstruction of the original
noisy image. To construct the estimated image, we replace the unregularized
wavelet coefficients $w_{d_{i},\ell(A_{i})}(A_{i})$ in \eqref{eq:IHUH}
by the regularized versions. 
\begin{theorem}
\label{thm:img-coef}Given a collection $\{\bm{x}_{i},y(\bm{x}_{i})\}_{i=1}^{n},\bm{x}_{i}\in\Omega$
with $n>1$ and $y(\bm{x}_{i})\in[0,1]$, let $M_{\Omega}$ be the
mean of all observed values $y(\bm{x}_{i})$ in the domain $\Omega$,
and let $L$ be the depth of the tree that characterizes the unbalanced
splits. Given a UHWT, suppose $A_{0},A_{1},\ldots,A_{S}$ is the decreasing
sequence of nodes that contain atom $\x$, where $S\leq L-1$. For
each $i=0,1,\ldots,S$, let $d_{i}$ be the axis along which node
$A_{i}$ is split, and let $u_{i}$ be $-1$ if $\bm{x}$ is in the
right child of the $i$th split, and $1$ if $\bm{x}$ is in the left
child of the $i$th split. Then $y(\bm{x})$ equals 
\begin{align}
M_{\Omega}+\sum_{i=0}^{S-1}u_{i}\frac{w_{d_{i},\ell(A_{i})}(A_{i})}{c(A_{i})}\text{ where }c(A_{i})=|A_{i}|^{1/2}\left(\frac{|A_{i,l}^{(d_{i})}|}{|A_{i,r}^{(d_{i})}|}\right)^{u_{i}/2},\label{eq:IHUH}
\end{align}
where $w_{d_{i},\ell(A_{i})}$ is the UH coefficient (MT-3)
for the partition component for the node $A_{i}$. 
\end{theorem}

\begin{proof}
For any node $A$, let $\overline{y(A)}\coloneqq|A|^{-1}\sum_{\bm{x}\in A}y(\bm{x})$
be the mean of all pixel values in $A$. For all $k=0,1,\ldots,S$,
define 
\begin{align*}
Q_{k} & \coloneqq\overline{y(A_{k})}+\sum_{i=k}^{S-1}u_{i}\frac{w_{d_{i},\ell(A_{i})}(A_{i})}{c(A_{i})}.
\end{align*}
(Here a summation vanishes if its starting index is larger than its
ending index.) Because $Q_{0}$ equals \eqref{eq:IHUH} and $Q_{S}=\overline{y(A_{S})}=y(\bm{x})$,
it suffices to show that $Q_{0}=Q_{1}=\cdots=Q_{S}$.

For all $k=0,1,\ldots,(S-1)$, we have $\overline{y(A_{k})}=\frac{1}{|A_{k}|}\sum_{\bm{a}\in A_{k,l}^{(d_{k})}}y(\bm{a})+\frac{1}{|A_{k}|}\sum_{\bm{a}\in A_{k,r}^{(d_{k})}}y(\bm{a})$
and 
\begin{align*}
\overline{y(A_{k})}+\ell_{k}\frac{w_{d_{k}}^{\text{UH}}(A_{k})}{c(A_{k})} & =\left[\frac{1}{|A_{k}|}+\ell_{k}\frac{1}{c(A_{k})}\left(\frac{1}{|A_{k,l}^{(d_{k})}|}-\frac{1}{|A_{k}|}\right)^{1/2}\right]\sum_{\bm{a}\in A_{k,l}^{(d_{k})}}y(\bm{a})\\
 & +\left[\frac{1}{|A_{k}|}-\ell_{k}\frac{1}{c(A_{k})}\left(\frac{1}{|A_{k,r}^{(d_{k})}|}-\frac{1}{|A_{k}|}\right)^{1/2}\right]\sum_{\bm{a}\in A_{k,r}^{(d_{k})}}y(\bm{a})\\
 & =\begin{cases}
\overline{y(A_{k,l}^{(d_{k})})}\text{ if }\ensuremath{\ell_{k}=1}\\
\overline{y(A_{k,r}^{(d_{k})})}\text{ if }\ensuremath{\ell_{k}=-1}
\end{cases}\\
 & =\overline{y(A_{k+1})}
\end{align*}
By adding the term $\sum_{i=k+1}^{S-1}u_{i}\frac{w_{d_{i},\ell(A_{i})}(A_{i})}{c(A_{i})}$
to every line in the preceding panel, we get $Q_{k}=Q_{k+1}$, which
implies $Q_{0}=Q_{1}=\cdots=Q_{S}$. 
\end{proof}

\subsection{General domain\label{subsec:General-domain}}

Theorem~\ref{thm:img-coef-general} below extends Theorem~\ref{thm:img-coef}
to 1-4 splits and
to a more general domain that has a notion of ``area'' or ``volume.''
Such a domain might not have a natural notion of which child is the
``left'' or ``right'' one, particularly if there is no natural
coordinate axis. Here we provide a labeling for such a case. (The
end of this subsection shows that this labeling is valid even though
it does not depend on any intrinsic properties of the nodes and is
chosen entirely by the user.)

Define the \emph{signed indicator}: 
\begin{align}
\ell_{j}=\begin{cases}
+1, & \text{if }x\in A_{j,\ell}\\
-1, & \text{if }x\in A_{j,r}
\end{cases}\,\text{(in a 1--2 split)}\quad\text{or}\quad\ell_{j}=\begin{cases}
+1, & \text{if }x\in\{A_{j,0},A_{j,2}\}\\
-1, & \text{if }x\in\{A_{j,1},A_{j,3}\}
\end{cases}\,\text{(in a 1--4 split)}
\end{align}
A 1--2 split we label the two children 
\begin{align}
A_{\ell}\quad\text{(the positive half-plane)},\qquad A_{r}\quad\text{(the negative one)}.
\end{align}
A 1--4 split still has one parent $A$ but now four grandchildren:
\begin{align}
A_{0},\;A_{1},\;A_{2},\;A_{3}.
\end{align}
To preserve the orthogonality of the UH wavelets, we \emph{group}
those four grandchildren into two aggregate halves: 
\begin{align}
A^{+}=A_{0}\cup A_{2},\qquad A^{-}=A_{1}\cup A_{3}.
\end{align}
Everything that was ``left vs right'' in the 1--2 split becomes
``$A^{+}$ vs $A^{-}$'' in the 1--4 split. 
\begin{itemize}
\item $\bar{y}_{A_{\ell}}\longrightarrow{\displaystyle \bar{y}_{A^{+}}=\frac{|A_{0}|\bar{y}_{A_{0}}+|A_{2}|\bar{y}_{A_{2}}}{|A^{+}|}}$ 
\item $\bar{y}_{A_{r}}\longrightarrow{\displaystyle \bar{y}_{A^{-}}=\frac{|A_{1}|\bar{y}_{A_{1}}+|A_{3}|\bar{y}_{A_{3}}}{|A^{-}|}}$ 
\item $|A_{\ell}|\longrightarrow|A^{+}|=|A_{0}|+|A_{2}|$ 
\item $|A_{r}|\longrightarrow|A^{-}|=|A_{1}|+|A_{3}|$ 
\item $\ell_{j}=\begin{cases}
+1 & \text{if }x\in\{A_{0},A_{2}\}\\
-1 & \text{if }x\in\{A_{1},A_{3}\}
\end{cases}$ 
\end{itemize}
\begin{theorem}
\label{thm:img-coef-general} Suppose we have a training collection
$\{\bm{x}_{a},y(\bm{x}_{a})\}_{a=1}^{n},\bm{x}_{a}\in\Omega$ with
$n>1$ and $y(\bm{x}_{a})\in\mathbb{R}$. For any node $A$, define
$y(A)\coloneqq\sum_{\bm{x}_{a}\in A}y(\bm{x}_{a})$, i.e., $y(A)$
is the sum of all training response values whose locations are in
$A$. Given a UHWT, for any node $A$ in the UHWT with child sets
$A^{+}$ and $A^{-}$ also in the UHWT, let 
\begin{align}
w(A)=\sqrt{\frac{|A^{+}|\cdot|A^{-}|}{|A|}}\left[y(A^{+})-y(A^{-})\right]
\end{align}
where $|A|=|A|_{n}$ denotes the number of training points in node
$A$. Given a location $\bm{x}^{*}\in\{\bm{x}_{i}\}_{i=1}^{n}$, suppose the
UHWT has a node $A_{S}$ at depth $S$ that contains $\bm{x}^{*}$ and
no other training location, and thus let $A_{S},A_{S-1},\ldots,A_{1},A_{0}$
be the telescoping sequence of nodes along the UHWT tree, where $A_{0}=\Omega$.
For each $i=0,1,\ldots,S$, let $u_{i}$ be $-1$ if $\bm{x}^{*}$
is in the ``minus'' child set $A_{i}^{-}$ of the $i$th split,
and $1$ if $\x^{*}$ is in the ``plus'' child set $A_{i}^{+}$
of the $i$th split. For all $k=0,1,\ldots,S$, define 
\begin{align*}
Q_{k} & \coloneqq\frac{y(A_{k})}{|A_{k}|}+\sum_{i=k}^{S-1}u_{i}\frac{w(A_{i})}{c(A_{i})},\text{ where }c(A_{i})=|A_{i}|^{1/2}\left(\frac{|A_{i}^{+}|}{|A_{i}^{-}|}\right)^{u_{i}/2}.
\end{align*}
Then $y(\bm{x}^{*})=Q_{0}=Q_{1}=\cdots=Q_{S}$. 
\end{theorem}

\begin{proof}
Because $Q_{S}=\frac{y(A_{S})}{|A_{S}|}=y(\bm{x}^{*})$, it suffices
to show that $Q_{0}=Q_{1}=\cdots=Q_{S}$.

For any $i=0,1,\ldots,S-1$, we have 
\begin{align*}
\frac{w(A_{i})}{c(A_{i})}=\frac{1}{|A_{i}|}\left(\frac{|A_{i}^{-}|}{|A_{i}^{+}|}\right)^{u_{i}/2}\left[\left(\frac{|A_{i}^{-}|}{|A_{i}^{+}|}\right)^{1/2}y(A_{i}^{+})-\left(\frac{|A_{i}^{+}|}{|A_{i}^{-}|}\right)^{1/2}y(A_{i}^{-})\right].
\end{align*}

For all $k=0,1,\ldots,(S-1)$, we have $y(A_{k})=y(A_{k}^{+})+y(A_{k}^{-})$
and thus 
\begin{align*}
\frac{y(A_{i})}{|A_{i}|}+u_{i}\frac{w(A_{i})}{c(A_{i})} & =\frac{1}{|A_{i}|}\left[1+u_{i}\left(\frac{|A_{i}^{-}|}{|A_{i}^{+}|}\right)^{(1+u_{i})/2}\right]y(A_{i}^{+})\\
 & +\frac{1}{|A_{i}|}\left[1-u_{i}\left(\frac{|A_{i}^{+}|}{|A_{i}^{-}|}\right)^{(1-u_{i})/2}\right]y(A_{i}^{-})\\
 & =\begin{cases}
y(A_{i}^{+})/|A_{i}^{+}| & \text{ if }\ensuremath{u_{i}=1}\\
y(A_{i}^{-})/|A_{i}^{-}| & \text{ if }\ensuremath{u_{i}=-1}
\end{cases}\\
 & =\frac{y(A_{i+1})}{|A_{i+1}|}.
\end{align*}
In this panel, the left-most term minus the right-most term equals
$Q_{i}-Q_{i+1}$. Thus we get $Q_{i}=Q_{i+1}$. Because this holds
for all $i$, we get $Q_{0}=Q_{1}=\cdots=Q_{S}$. 
\end{proof}
We remark that the "plus" and "minus" labeling of children does
not depend on any intrinsic properties of the nodes and is chosen
entirely by the user. We prove this here. Given the telescoping sequence
$A_{S},\ldots,A_{1},A_{0}$ from the previous theorem, let $B_{k}=A_{k}$
for all $k$, but suppose we swap the plus and minus labels of the
children of the node $B_{j}$ at some depth $j=0,1,\ldots,S-1$ so
that $A_{j}^{+}=B_{j}^{-}$ and $A_{j}^{-}=B_{j}^{+}$, but $A_{k}^{+}=B_{k}^{+}$
and $A_{k}^{-}=B_{k}^{-}$ and for all $k\neq j$. We then would have
$\ell_{j}^{A}=-\ell_{j}^{B}$, but $\ell_{k}^{A}=\ell_{k}^{B}$ for
all $k\neq j$. %
Comparing $\{Q_{k}^{B}\}_{k=0}^{S}$ to $\{Q_{k}^{A}\}_{k=0}^{S}$,
we note that any difference could only stem from a difference between
\begin{align}
\ell_{j}^{B}\frac{w(B_{j})}{c(B_{j})}\quad\text{vs}\quad\ell_{j}^{A}\frac{w(A_{j})}{c(A_{j})}.\label{eq:comparisonAB}
\end{align}
But we have 
\begin{align*}
c(B_{j})=|B_{j}|^{1/2}\left(\frac{|B_{j}^{+}|}{|B_{j}^{-}|}\right)^{\ell_{j}^{B}/2}=|A_{j}|^{1/2}\left(\frac{|A_{j}^{-}|}{|A_{j}^{+}|}\right)^{(-\ell_{j}^{A})/2}%
=c(A_{j})
\end{align*}
and 
\begin{align*}
\ell_{j}^{B}w(B_{j})&=\ell_{j}^{B}\sqrt{\frac{|B_{j}^{+}|\cdot|B_{j}^{-}|}{|B_{j}|}}\left[y(B_{j}^{+})-y(B_{j}^{-})\right]=(-\ell_{j}^{A})\sqrt{\frac{|A_{j}^{-}|\cdot|A_{j}^{+}|}{|A_{j}|}}\left[y(A_{j}^{-})-y(A_{j}^{+})\right]\\
&=\ell_{j}^{A}w(A_{j}).
\end{align*}
Thus the two terms in \eqref{eq:comparisonAB} are equal to each other,
and so $Q_{k}^{B}=Q_{k}^{A}$ for all $k$.

\subsection{Connection to SHAH Algorithm}

\label{sec:Connection-to-SHAH}

Our proposed UHWT algorithm has a connection to the SHAH algorithm
from \citet{fryzlewicz2016shah}. The SHAH algorithm takes a bottom-up
approach that merges the two neighboring pixel pairs with the minimum
absolute value of detailed coefficients, which has the same complexity
as single linkage hierarchical clustering of $O(n^{2})$, slower than
our Algorithm~\ref{alg:uhwt-opt}'s additional $O(D\cdot L\cdot n/2^{L})$
complexity if $L$ grows with $n$ (Also see Section~\ref{sec:Complexity}).

To see this connection, suppose the SHAH algorithm merges two nodes
$j$ and $k$ (not necessarily atoms), where their respective squared
weights $w_{j}^{2}$ and $w_{k}^{2}$ equal the number of cells in
the respective nodes $j$ and $k$. Then the merged node by definition
has squared weight equal to $w_{j}^{2}+w_{k}^{2}$, which by assumption
equals the sum of the number of cells in the two nodes $j$ and $k$.
By definition of SHAH algorithm, the merged node (which we index as
node $\ell$) has term 
\begin{align*}
X_{\ell}\coloneqq\frac{w_{j}X_{j}+w_{k}X_{k}}{\sqrt{w_{j}^{2}+w_{k}^{2}}}.
\end{align*}
then $w_{j}X_{j}$ equals the sum of the pixel values in node $j$,
and $w_{k}X_{k}$ equals the sum of the pixel values in node $k$.
Thus the numerator in the above panel equals the sum of the pixel
values in the merged node. Furthermore, the denominator equals the
square root of number of atoms in the merged node.

Suppose any atom has weight equal to 1 and is an atomic node, the
detailed coefficient 
\begin{align*}
\tilde{d}_{\ell}\coloneqq\frac{w_{j}}{\sqrt{w_{j}^{2}+w_{k}^{2}}}X_{k}-\frac{w_{k}}{\sqrt{w_{j}^{2}+w_{k}^{2}}}X_{j}=\frac{1}{\sqrt{n_{\ell}}}\left(\frac{\sqrt{n_{j}}}{\sqrt{n_{k}}}\sum_{\text{atom \ensuremath{i} in node }k}y_{i}-\frac{\sqrt{n_{k}}}{\sqrt{n_{j}}}\sum_{\text{atom \ensuremath{i} in node }j}y_{i}\right),
\end{align*}
which is exactly $w_{d,\ell}(A)$ in (MT-3). By
the SHAH algorithm, two neighboring nodes are chosen to merge if the
edge between them has the minimum absolute value of detail among any
pair of neighboring nodes. This is equivalent to splitting a node
at the location with the maximum absolute value of detail. Therefore,
our algorithm can be considered a generalization of SHAH.

\section{Proof of Theorem~MT-5.1}
\begin{proof}
Since $\mu=\mu_{n}$ on the grid, $\{\psi_{j}\}$ is orthonormal in
$L^{2}(\mu)$ and $w_{j}=\theta_{j}+z_{j}$ with $\theta_{j}=\langle f,\psi_{j}\rangle$
and $z_{j}=\langle\varepsilon,\psi_{j}\rangle$ sub-Gaussian with
proxy $\sigma/\sqrt{n}$. A union bound gives $\max_{j}|z_{j}|\le\tau$
with probability at least $1-\delta$. On this event, the standard
orthogonal soft-thresholding argument in Lemma~\ref{lem:soft-thresh}
yields, for any $S\subset\{1,\dots,M\}$, the inequality
\begin{align}
\sum_{j=1}^{M}\big(\operatorname{sgn}(w_{j})(|w_{j}|-\tau)_{+}-\theta_{j}\big)^{2}\ \le\ 4\sum_{j\notin S}\theta_{j}^{2}+4|S|\tau^{2}.
\end{align}
Adding the constant component and using orthonormality in $L^{2}(\mu)$
gives the UH-estimator inequality (MT-26). 
For the leafwise inequality (MT-27), 
for each $A \in \mathcal{P}$ we define
\[
\bar{f}_A \coloneqq \frac{1}{N_A} \sum_{i\colon \mathbf{X}_i \in A} f(\mathbf{X}_i)
\qquad\text{and}\qquad
\bar{\varepsilon}_{A} \coloneqq \frac{1}{N_A} \sum_{i\colon \mathbf{X}_i \in A} \varepsilon_i,
\]
decompose $\widehat{f}_{\mathcal{P}}^{\mathrm{leaf}}-f$ on $A$ as $(\bar{f}_A-f)+\bar{\varepsilon}_{A}$,
and union-bound control $\bar{\varepsilon}_{A}$ to yield the stated variance term. The Hölder bound on $\|f-f_{\mathcal{P}}\|$
is standard. 

In both the grid construction and the triangulated construction, the
paper proves that each admissible split produces two child indicator
functions which are balanced and normalized with respect to the
current measure. This guarantees orthogonality of the new detail to
the space generated so far, and it also guarantees unit norm in $L_{2}(\nu)$.
Since every split raises the number of cells by one, after $|\mathcal{P}|-|\mathcal{P}_0|$
splits we have exactly $|\mathcal{P}|-|\mathcal{P}_0|$ such details. The reconstruction
identities given in the paper for the grid case and for the triangulated
case show that these details together with the global mean span the
whole space of piecewise-constant functions on the final partition.
Once this orthonormal system is available, the coefficientwise soft-thresholding argument of the paper applies without change in either
domain and yields the oracle inequality with $M=|\mathcal{P}|-|\mathcal{P}_0|$.
The proof of the leafwise bound is also independent of the domain
since it only uses the union bound for the empirical noise means on
each cell of the same partition. 
\end{proof}

\begin{lemma}[Soft-thresholding inequality] \label{lem:soft-thresh}
Let $\theta=(\theta_{j})_{j=1}^{M}\in\mathbb{R}^{M}$. Suppose for
some $\tau>0$ that $z=(z_{j})_{j=1}^{M}\in\mathbb{R}^{M}$ satisfies
the coordinatewise bound 
\begin{align}
\|z\|_{\infty}\;\coloneqq\;\max_{1\le j\le M}|z_{j}|\;\le\;\tau.\label{eq:inf-bound}
\end{align}
Define the soft--threshold map $\eta_{\tau}(y)=\mathrm{sign}(y)\,(|y|-\tau)_{+}$.
Define $\widehat{\theta}=(\widehat{\theta}_{j})_{j=1}^{M}$ by soft--thresholding
the noisy coefficients $\theta+z$: 
\begin{align}
\widehat{\theta}_{j}\;=\;\eta_{\tau}(\theta_{j}+z_{j})\qquad(j=1,\dots,M).\label{eq:def-that}
\end{align}
Then the following deterministic inequalities hold: 
\begin{align}
\sum_{j=1}^{M}\big(\widehat{\theta}_{j}-\theta_{j}\big)^{2} 
& \le4\sum_{j=1}^{M}\min\big\{\theta_{j}^{2},\ \tau^{2}\big\},\label{eq:soft-master}\\
 & \le4\bigg(\sum_{j\notin S}\theta_{j}^{2}\;+\;|S|\,\tau^{2}\bigg)\quad\text{for every subset }S\subset\{1,\ldots,M\}.\label{eq:soft-oracle}
\end{align}
\end{lemma} 
\begin{proof}
We prove \eqref{eq:soft-master} by a coordinatewise argument, then
sum over $j$; \eqref{eq:soft-oracle} is an immediate consequence.

Fix an index $j$ and set 
\begin{align}
a\;=\;\theta_{j},\qquad b\;=\;z_{j},\qquad u\;=\;a+b,\qquad\widehat{a}\;=\;\eta_{\tau}(u).
\end{align}
Our goal is to show the scalar inequality 
\begin{align}
(\widehat{a}-a)^{2}\;\le\;4\min\{a^{2},\ \tau^{2}\}\quad\text{under the bound }|b|\le\tau.\label{eq:scalar-goal}
\end{align}
Summing \eqref{eq:scalar-goal} over $j$ yields \eqref{eq:soft-master}.

We consider two disjoint cases:

\textit{Case 1:} $|u|\le\tau$. By the triangle inequality, 
\begin{align}
|a|\;\le\;|u|+|b|\;\le\;\tau+\tau\;=\;2\tau,\label{eq:a-upper-2tau}
\end{align}
hence $a^{2}\le4\tau^{2}$. Because the inequality $|u|\le\tau$ implies
$\widehat{a}=\eta_{\tau}(u)=0$, we get 
\begin{align}
(\widehat{a}-a)^{2}=(0-a)^{2}=a^{2}\;\le\;\min\{a^{2},\ 4\tau^{2}\}\;\le\;4\min\{a^{2},\ \tau^{2}\}.\label{eq:case1-end}
\end{align}
 \textit{Case 2:} $|u|>\tau$. Then $\widehat{a}=\eta_{\tau}(u)=u-\tau\,\mathrm{sign}(u)$
and 
\begin{align}
\widehat{a}-a=\{u-\tau\,\mathrm{sign}(u)\}-a=b-\tau\,\mathrm{sign}(u).\label{eq:case2-diff}
\end{align}
We split Case 2 into two subcases. 

\textit{Case 2a:} $|a|\ge\tau$. Then $\min\{a^{2},\tau^{2}\}=\tau^{2}$.
From \eqref{eq:case2-diff} and the triangle inequality, 
\begin{align}
|\widehat{a}-a|\;=\;|b-\tau\,\mathrm{sign}(u)|\;\le\;|b|+\tau\;\le\;\tau+\tau\;=\;2\tau.\label{eq:bound-by-2tau}
\end{align}
Hence 
\begin{align}
(\widehat{a}-a)^{2}\;\le\;4\tau^{2}\;=\;4\min\{a^{2},\tau^{2}\}.\label{eq:case2a-end}
\end{align}

\textit{Case 2b:} $|a|<\tau$. Then $\min\{a^{2},\tau^{2}\}=a^{2}$.
Using \eqref{eq:case2-diff}, the triangle inequality, and the identity
$|\eta_{\tau}(u)|=|u|-\tau$ (which holds because $|u|>\tau$ by assumption), we obtain 
\begin{align}
|\widehat{a}-a| & \le|\widehat{a}|+|a|=(|u|-\tau)+|a|\label{eq:use-modulus}\\
 & \le(|a|+|b|-\tau)+|a| \qquad\qquad\text{(by $|u|=|a+b|\le|a|+|b|$)}\label{eq:triangle-u}\\
 & \le|a|+|a|\;=\;2|a| \quad\qquad\qquad\text{(since $|b|\le\tau$)}\label{eq:use-b-bound}
\end{align}
and therefore 
\begin{align}
(\widehat{a}-a)^{2}\;\le\;4a^{2}\;=\;4\min\{a^{2},\tau^{2}\}.\label{eq:case2b-end}
\end{align}

In Case 1 we have \eqref{eq:case1-end}. In Case 2a we have \eqref{eq:case2a-end}.
In Case 2b we have \eqref{eq:case2b-end}. These three inequalities
exhaust all possibilities for $(a,b)$ under $|b|\le\tau$ and each
yields \eqref{eq:scalar-goal}. Summing the scalar bound over $j=1,\dots,M$
gives \eqref{eq:soft-master}.

Fix an arbitrary subset $S\subset\{1,\dots,M\}$. For each $j$, 
\begin{align}
\min\{\theta_{j}^{2},\tau^{2}\}\;\le\;\begin{cases}
\tau^{2}, & j\in S,\\
\theta_{j}^{2}, & j\notin S.
\end{cases}
\end{align}
Combine with \eqref{eq:soft-master} to obtain \eqref{eq:soft-oracle}. 
\end{proof}

\section{Proof of Corollary~MT-5.2}
\begin{proof}
By Theorem~MT-5.1, for $\tau=\sigma\sqrt{2n^{-1}\log(2M/\delta)}$
the UH soft-thresholded estimator satisfies, with probability at least
$1-\delta$, 
\begin{align}
\|\widehat{f}_{P,\tau}^{\mathrm{UH}}-f\|_{L^{2}(\mu)}^{2}\le4\inf_{S\subset\{1,\dots,M\}}\left\{ \|f-f_{S}\|_{L^{2}(\mu)}^{2}+2|S|\,\frac{\sigma^{2}}{n}\log\frac{2M}{\delta}\right\} .
\end{align}
Taking $S=S^{\star}$ and using the assumption $f=f_{S^{\star}}$ yields 
\begin{align}
\|\widehat{f}_{P,\tau}^{\mathrm{UH}}-f\|_{L^{2}(\mu)}^{2}\le8\,s\,\frac{\sigma^{2}}{n}\log\frac{2M}{\delta}.
\end{align}

For the leafwise estimator, Theorem~MT-5.1 gives,
with probability at least $1-\delta$, 
\begin{align}
\|\widehat{f}_{\mathcal{P}}^{\mathrm{leaf}}-f\|_{L^{2}(\mu)}^{2}\le\|f-f_{\mathcal{P}}\|_{L^{2}(\mu)}^{2}+4\,\frac{\sigma^{2}}{m_{\min}(\mathcal{P})}\log\frac{2|\mathcal{P}|}{\delta}.
\end{align}
Since $f$ is piecewise-constant on $\mathcal{P}$ by assumption ($f\in\mathrm{span}\{1,\psi_{j}\}_{j\in S^{\star}}$),
we have $f=f_{\mathcal{P}}$ and the approximation term vanishes. The $c$-balance
condition $m_{\min}(\mathcal{P})\ge c\,n/|\mathcal{P}|$ implies 
\begin{align}
\|\widehat{f}_{\mathcal{P}}^{\mathrm{leaf}}-f\|_{L^{2}(\mu)}^{2}\le\frac{4}{c}\,|\mathcal{P}|\,\frac{\sigma^{2}}{n}\log\frac{2|\mathcal{P}|}{\delta}.
\end{align}
This proves the two bounds and the stated comparison of variance scales. 
\end{proof}

\section{Complexity\label{sec:Complexity}}
We use the following notations in this section: 
\begin{itemize}
\item $n$: number of observations, $\bm{x}_{i}\in\mathbb{R}^{D}$, all assumed
unit-norm. 
\item $D$: embedding dimension ($D=3$ for the sphere). 
\item $L$: maximal tree depth; with regular refinement $L=\lceil\log_{2}n\rceil$. 
\item $n_{\ell}$: total number of points stored in depth $\ell$ of a tree. 
\item $M$: number of trees in an ensemble. 
\end{itemize}
For every method we give \emph{time} $T(\cdot)$ and \emph{space}
$S(\cdot)$ in big-$O$ / big-$\Theta$ form and lower-order constants
are omitted.

\begin{table}[h!]
\centering\label{tbl:complexity} %
\begin{tabular}{lcc}
\hline 
{\small{}{}{}{}{}{}{}{}Method}  & \multicolumn{1}{c}{{\small{}{}{}{}{}{}{}{}Time $T(n,D)$}} & \multicolumn{1}{c}{{\small{}{}{}{}{}{}{}{}Space $S(n)$}}\tabularnewline
\hline 
{\small{}{}{}{}{}{}{}{}single-tree\,\texttt{balance} or \texttt{balance4}}  & {\small{}{}{}{}{}{}{}{}$\Theta\bigl(D\,nL\bigr)$}  & {\small{}{}{}{}{}{}{}{}$\Theta(n)$}\tabularnewline
{\small{}{}{}{}{}{}{}{}single-tree\,\texttt{adapt} or \texttt{adapt\_vertex}}  & {\small{}{}{}{}{}{}{}{}$\Theta\bigl(D\,n^{2}\bigr)$}  & {\small{}{}{}{}{}{}{}{}$\Theta(n)$}\tabularnewline
{\small{}{}{}{}{}{}{}{}Random Rotation Forest\,(\texttt{balance})}  & {\small{}{}{}{}{}{}{}{}$M\!\bigl(O(D^{3})+O(nD^{2})+\Theta(D\,nL)\bigr)$}  & {\small{}{}{}{}{}{}{}{}$\Theta(Mn)$}\tabularnewline
{\small{}{}{}{}{}{}{}{}Random Rotation Forest\,(\texttt{adapt})}  & {\small{}{}{}{}{}{}{}{}$M\!\bigl(O(D^{3})+O(nD^{2})+\Theta(D\,n^{2})\bigr)$}  & {\small{}{}{}{}{}{}{}{}$\Theta(Mn)$}\tabularnewline
{\small{}{}{}{}{}{}{}{}Rotation Forest (Euclidean\ PCA)}  & {\small{}{}{}{}{}{}{}{}$M\!\bigl(O(bD^{3}/K^{2})+\Theta(D\,nL)\bigr)$}  & {\small{}{}{}{}{}{}{}{}$\Theta(Mn)$}\tabularnewline
{\small{}{}{}{}{}{}{}{}SPCA\,Rotation Forest}  & {\small{}{}{}{}{}{}{}{}$M\!\bigl(O(n)+\Theta(D\,nL)\bigr)$}  & {\small{}{}{}{}{}{}{}{}$\Theta(Mn)$}\tabularnewline
{\small{}{}{}{}{}{}{}{}Gaussian Process on $\mathbb{S}^{D-1}$}  & {\small{}{}{}{}{}{}{}{}$\Theta(n^{3})$}  & {\small{}{}{}{}{}{}{}{}$\Theta(n^{2})$}\tabularnewline
\hline 
\end{tabular}\caption{Asymptotic training complexity (See SM~\ref{sec:Complexity}
for detailed derivations). $L=\lceil\log_{2}n\rceil$, $M$ = ensemble
size, $b$ the bootstrap fraction and $K$ the block count in Euclidean
Rotation Forest. For $D=3$ the terms $O(D^{3})$ and $O(nD^{2})$
are constant-factor overheads.}
\label{tbl:complexity1} 
\end{table}

\subsection{Single-tree splitters}

\textbf{Balanced midpoint split (method }\texttt{\textbf{balance}}\textbf{).}
The procedure satisfies 
\begin{align}
T_{\mathrm{bal}}(n,D)=\Theta\!\bigl(D\,nL\bigr),\qquad S_{\mathrm{bal}}(n)=\Theta(n).
\end{align}

At any node containing $n_{\ell}$ points one performs the 
\begin{enumerate}
\item longest-edge search: $O(1)$ arithmetic, 
\item one midpoint evaluation: $O(D)$, 
\item a single membership scan of the local list: $O(D\,n_{\ell})$. 
\end{enumerate}
Because the lists of siblings form a partition of the parent list
we have $\sum_{\ell=0}^{L}n_{\ell}=n$. Summation over depths therefore
yields $DnL$. Each sample index is stored exactly once, hence linear
memory.

\textbf{Balanced four-way split (method}\texttt{\textbf{balance4}}\textbf{).}
In the worst case 
\begin{align}
T_{\mathrm{bal4}}(n,D)=\Theta\!\bigl(D\,nL\bigr),\qquad S_{\mathrm{bal4}}(n)=\Theta(n).
\end{align}

Every point is still inspected once per tree level; the factor $4$ in
the number of children changes only the leading constant.

\textbf{Adaptive edge split (methods }\texttt{\textbf{adapt}}\textbf{
and }\texttt{\textbf{adapt\_vertex}}\textbf{).}
In the worst case 
\begin{align}
T_{\mathrm{adapt}}(n,D)=T_{\mathrm{adapt\_{vertex}}}(n,D)=\Theta\!\bigl(D\,n^{2}\bigr),\qquad S_{\mathrm{adapt}}(n)=\Theta(n).
\end{align}

Inside a node with $n_{\ell}$ points, the algorithm selects \emph{every}
point as pivot and, for each of three edges, rescans the full list.
Cost per node $=O(3D\,n_{\ell}^{2})$. If early splits are rejected
the root remains with $n_{\ell}=n$, giving $Dn^{2}$. Memory is unchanged.

\subsection{Rotation ensembles}

All ensembles apply a pointwise rotation $\bm{x}\mapsto \bm{x}R^{\top}$ before
calling a base splitter; prediction averages $M$ trees and costs
$O(Mn)$.

\textbf{Random-Rotation Ensemble (RRE).} In each tree: 
\begin{enumerate}
\item sample $R\in SO(D)$ via QR: $O(D^{3})$, 
\item rotate data: $O(nD^{2})$, 
\item build a tree: $T_{\text{tree}}(n,D)$. 
\end{enumerate}
Hence 
\begin{align}
T_{\mathrm{RRE}}(n,D,M)=M\!\left(O(D^{3})+O(nD^{2})+T_{\text{tree}}(n,D)\right).
\end{align}
For $D=3$ the algebraic overhead is constant.

\textbf{Rotation Forest (Euclidean PCA).}
Let $K$ be the block count and $b\in(0,1]$ the bootstrap fraction.
Per tree one performs $K$ SVDs of size $bn\times D/K$, each $O\!\bigl((D/K)^{3}\bigr)$.
Thus 
\begin{align}
T_{\mathrm{RotF}}=M\!\left(O\!\bigl(bD^{3}/K^{2}\bigr)+T_{\text{tree}}(n,D)\right),
\end{align}
again constant for $D=3$.

\textbf{Rotation Forest with spherical PCA (sPCA,\citet{hrluo_2021c}).}
Principal geodesic analysis in $\R^{3}$ costs $O(bnD)+O(D^{3})=O(n)$
per tree (dominated by mean and projection). Therefore 
\begin{align}
T_{\mathrm{sPCA\text{-}RotF}}(n,3,M)=M\,T_{\text{tree}}(n,3)+O(M).
\end{align}

\subsection{Intrinsic Gaussian process on the sphere}

Exact GP regression with a Matérn kernel on the sphere $\mathbb{S}^{2}$ satisfies 
\begin{align}
T_{\mathrm{GP}}(n)=\Theta(n^{3}),\qquad S_{\mathrm{GP}}(n)=\Theta(n^{2}),
\end{align}
and prediction at $m$ new points costs $\Theta(mn)$. 
\begin{proof}
Optimization of the exact marginal likelihood requires one inversion
or Cholesky factorisation of the $n\times n$ covariance matrix; both
are cubic. The dense matrix itself stores $n^{2}$ entries. Posterior
mean evaluation multiplies the same matrix inverse by $m$ right-hand
sides. 
\end{proof}

\section{Other tree-based methods for manifold domains}
\label{sec:other-tree-methods-for-manifold}

Other tree-based methods for manifold domains, such as BAST
\citep{luo2021bast} and BAMDT \citep{luo2022bamdt}, build adaptive
partitions by cutting along a spanning tree in a neighborhood graph
rather than by refining faces of a triangulation. Medoid trees
\citep{bulte2024medoid} partition observations directly to reduce
within node dispersion, but they do not introduce nested simplicial
refinements. In contrast, our UHWT  refines within faces when needed. This
allows it to emulate the flexible random partitions of spanning tree
methods while maintaining a sample exact reconstruction and a
multiscale representation that follows edges on curved surfaces. 
More global approaches based on graph spectral wavelets
\citep{mostowsky2024} or diffusion frames \citep{hui2022neural}
rely on eigen-decompositions of a Laplace operator and do not yield
a recursive data driven partition. Mesh-based wavelets such as those
of \citet{schroder1995spherical} assume a fixed nested tessellation
and so cannot adapt to sparse or irregular sampling, which our data-driven partitions can handle relatively well even with shallow trees.

\section{Further numerical experiments}

\subsection{Comparison to other image denoising methods}

\label{subsec:comparison-image-denoising}

Here we provide numerical comparisons of our wavelet-based methods
and general-use tree ensemble methods. These comparisons will use
real-world images. In both applications, we will see that our boosting
wavelet approach, which utilizes both wavelet technology and ensembling,
often has superior MSE performance against the tested methods. We
will see that many existing wavelet approaches are not able to overcome
the benefit gained by ensembling single learners, and that general-use
tree ensembles cannot overcome the benefits of wavelet-specific approaches
(e.g., energy conservation, sharp-edge detection) that have been massively
developed in image processing.

\subsubsection{OCT}

\label{subsec:OCT}

\begin{figure}
\centering \includegraphics[width=1\textwidth]{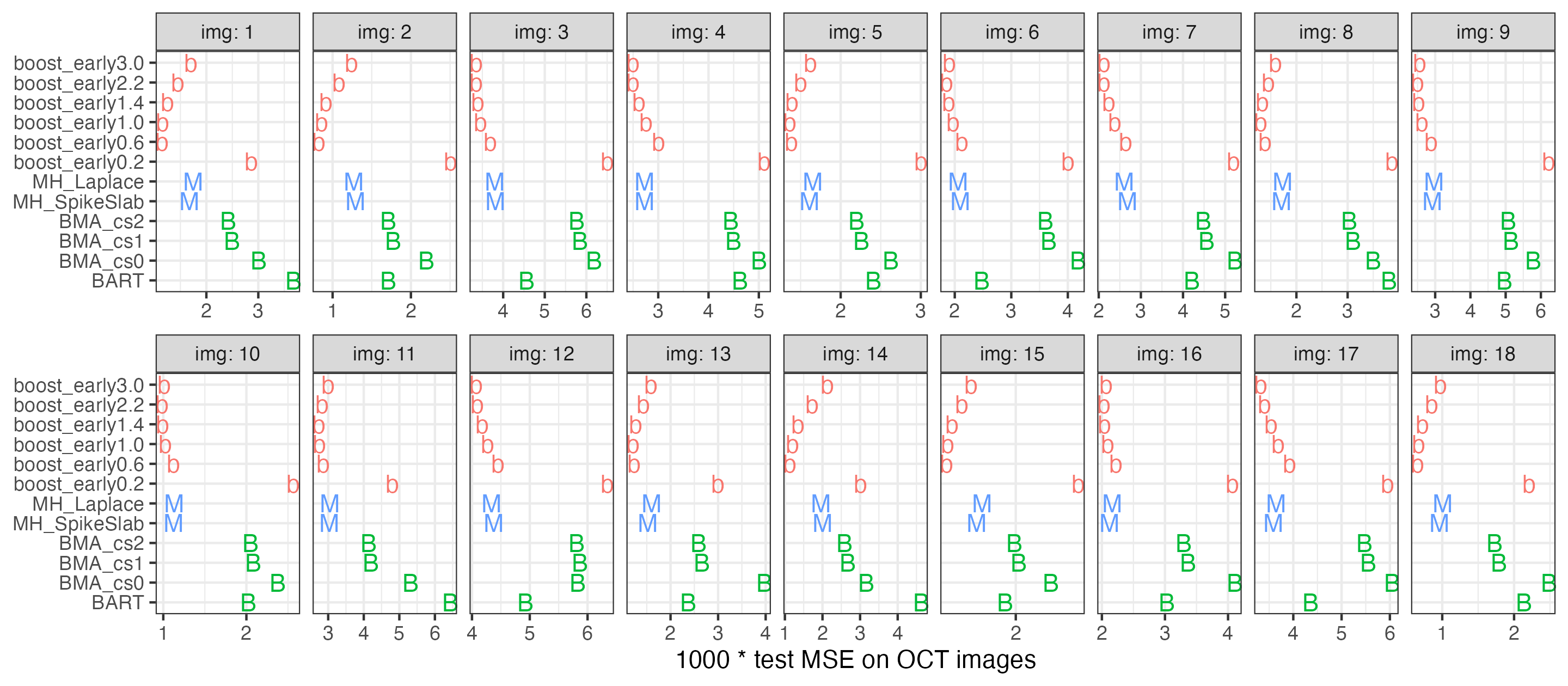}
\caption{Test MSE on 18 OCT images. boost denotes forward stagewise boosting with the
stated prune parameter $b$ used in the data-independent thresholding
quantity $\tau=b\sqrt{2\log n}$ as described in Section~MT-2.2.
MH\_Laplace and MH\_SpikeSlab denote the posterior mean of the backfitting-ensemble
draws with each tree grown using Metropolis-Hastings and either the
Laplace or spike-and-slab likelihood, respectively. BMA denotes WARP
with the stated step size for cycle spinning (cs). }
\label{fig:OCT} 
\end{figure}

For the first set of comparisons, we use optical coherence tomography
(OCT) image slices available at \url{https://people.duke.edu/~sf59/fang_tmi_2013.htm}.
For each of the 18 slices, there is an image with noise already included
and a registered image obtained by averaging 40 repeatedly sampled
scans \citep{fang2016segmentation}; we use the latter as the ``noiseless''
reference image to compare our denoised images against. %
Due to WARP requiring image size dimensions to be powers of two, we
resize these images from $900\times450$ to $1024\times512$.

On these images, we fit models using our boosting approach
and our backfitting approach with each tree grown using Metropolis-Hastings
and either the Laplace or spike-and-slab likelihood. For boosting we use
500 estimators and learning rate $0.05$; we also explore various
pruning values. For our backfitting approach we use 200 estimators
and 70 backfitting iterations and prior node-split probability $0.95*0.5^{\text{depth}}$.
We compare our approaches against WARP with cycle spinning (BMA) \citep{li2021learning,li2015fast},
SHAH \citep{fryzlewicz2016shah}, Translation-invariant (TI) de-noising
\citep{coifman1995translation}, Bayesian Additive Regression Trees
\citep{chipman2010bart} using the \texttt{BART} R package in CRAN
\citep{BARTCRAN}, and nonparametric Bayesian dictionary learning
\citep{zhou2011nonparametric}. Cycle spinning is a technique that
can be used to remove visual artifacts when reconstructing an image.

Figure~\ref{fig:OCT} shows the test MSE of the denoised images using
boosting with various pruning thresholds, our backfitting approach, BMA
with various cycle-spinning steps, and BART. For each image, we found
SHAH and TI to produce markedly larger MSEs than either boosting or BMA.
Furthermore, we found that nonparametric Bayesian dictionary learning
takes longer than boosting and requires 5GB cache for 18 OCT images. Hence,
we show the MSE results for only boosting, MH, BMA, and BART. We find for
prune parameters between 0.6 and 3.0 that boosting has smaller test MSE than
does BMA or BART. It seems that increasing the amount of cycle spinning
in BMA will further lower the test MSE, but not to an amount smaller
than the shown boosting MSEs. Furthermore, we see that the boosting MSEs exhibit
a `U' shape with respect to the prune parameter, which indicates a
bias-variance trade-off. Namely, the prune parameter $b=0.2$ provides
too little regularization to prevent overfitting. Finally, we see
that our backfitting approach produces test MSEs comparable to those
of boosting. Generally, we find that the denoised images for boosting and BMA
to be crisper than those for MH and BART, which tend to be more textured.

\subsubsection{BSDS}

\label{subsec:BSDS}

For this set of comparisons, we use 300 natural images from the Berkeley
segmentation dataset (BSDS300). For each noiseless image, we created
six noisy versions. For the first noise mechanism, we add Gaussian
noise with a standard deviation equal to, respectively, 0.2, 0.4,
and 0.6 times the noiseless image's pixel-value standard deviation.
For the second noise mechanism, we add Gaussian noise with a standard
deviation equal to $5/255$, $10/255$, and $15/255$. Thus each model
is ultimately trained on $300\times(3+3)=1800$ noisy images. Due
to the many training images and WARP requiring image size dimensions
to be powers of two, we resize these natural images to $256\times256$.

To these images, we fit models using our boosting approach
and our backfitting approach with either Laplace or spike-and-slab
likelihoods. For boosting we use 500 estimators and learning rate $0.05$;
we also explore various pruning values. For our backfitting approach
we use 1000 estimators and 500 backfitting iterations and prior node-split
probability $0.95*0.5^{\text{depth}}$. We compare our approaches
against BMA, BART, SHAH, and TI using the same code implementation
as in Section~\ref{subsec:OCT}.

The top row of Figure~\ref{fig:BSDS-mse} shows the test MSE averaged
over all 300 images for each noise level when the noise level depends
on the noiseless image's standard deviation. At each noise level,
the smallest test MSE is achieved by either boosting or BMA. boosting does best
at the smallest noise level, but BMA overtakes boosting for the larger noise
levels. However, we find that BMA tends to oversmooth textures in
the noiseless images. For example, the top row of Figure~\ref{fig:BSDS293029}
shows a noisy image (noise level $0.2$) with visible hexagonal tiling
on the ground. BMA's denoised image omits almost any hint of the hexagons,
whereas boosting is able to preserve most of the hexagons. Even when the
noise level is increased to $0.4$ (bottom row of Figure~\ref{fig:BSDS293029}),
boosting is still able to preserve a little bit of the hexagonal tiling.
The general pattern we found is that boosting tends to preserve textures
(e.g., ripples in a river, rock surfaces) and curved boundaries (e.g.,
of an owl) better than BMA does.

\begin{figure}[ht!]
\centering \includegraphics[width=1\textwidth]{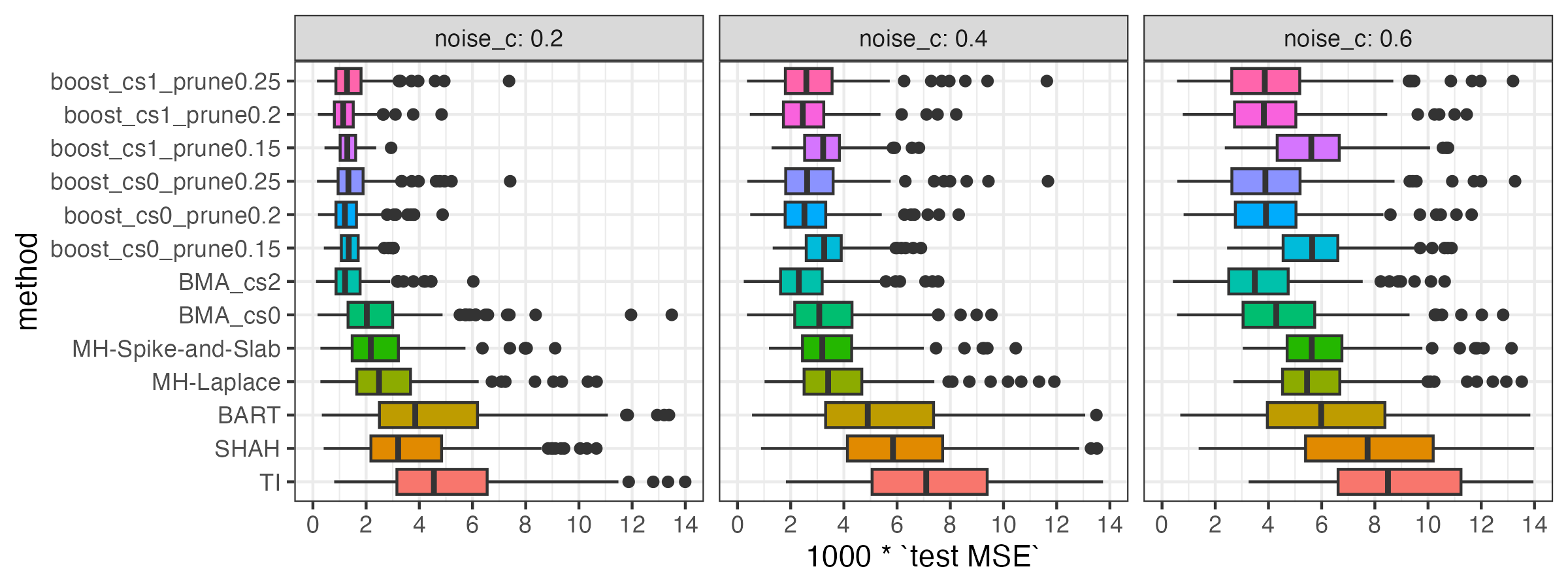}
~ \includegraphics[width=1\textwidth]{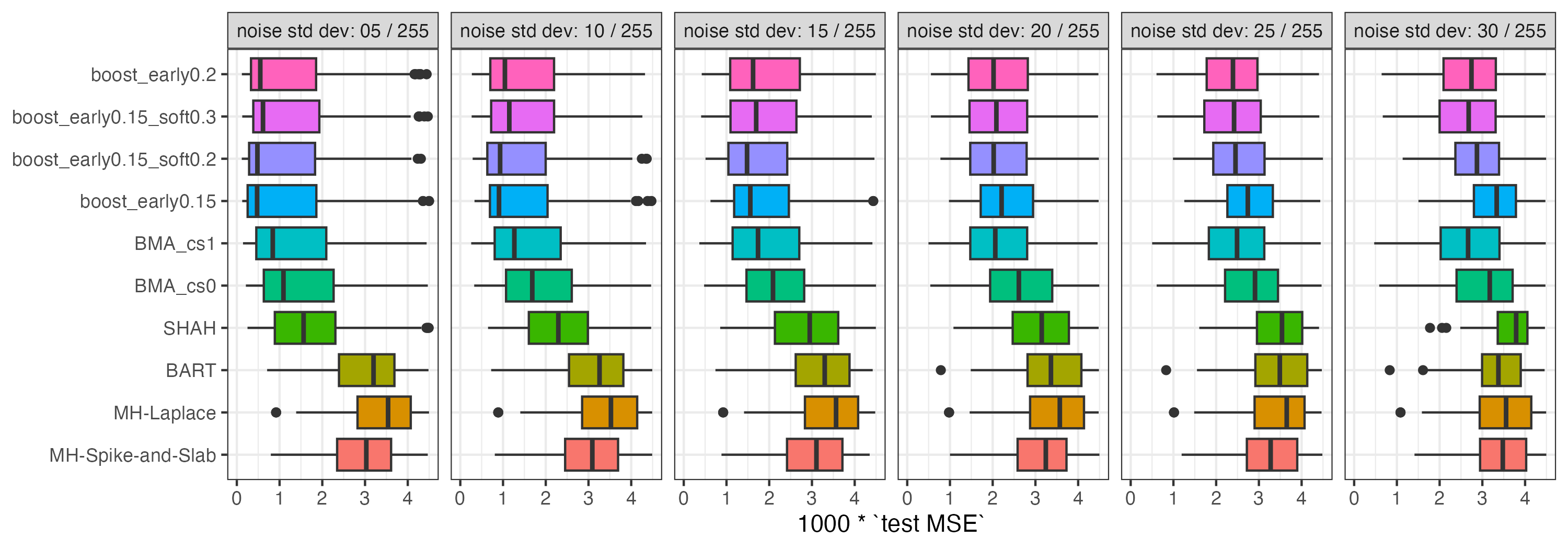}
\caption{Test MSE times $10^{3}$ averaged over 300 BSDS images. Each panel
corresponds to the stated noise level (\textbf{top:} noise standard
deviation is noise\_c $\times$ sd(image); \textbf{bottom:} noise
standard deviation is as stated and is independent of sd(image)).
Each row represents a method with some hyperparameter values; for
the BMA columns, \textquotedblleft cs\textquotedblright{} indicates
cycle spinning with some number of steps; for the boost columns, \textquotedblleft$b$\textquotedblright{}
indicates hyperparameter $b$ defined in \S~MT-6.1.
MH-Laplace and MH-Spike-and-Slab represent the posterior mean of the
backfitting-ensemble draws with each tree grown using Metropolis-Hastings
and either the Laplace or spike-and-slab likelihood, respectively.}
\label{fig:BSDS-mse} 
\end{figure}

\begin{figure}[t]
\centering \includegraphics[width=0.24\textwidth]{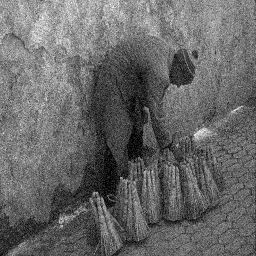}
~ \includegraphics[width=0.24\textwidth]{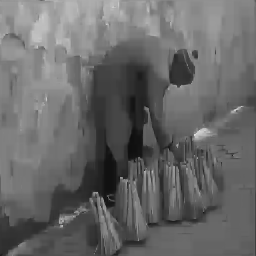}
~ \includegraphics[width=0.24\textwidth]{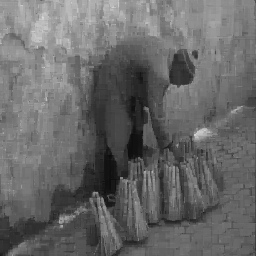}
~ \includegraphics[width=0.24\textwidth]{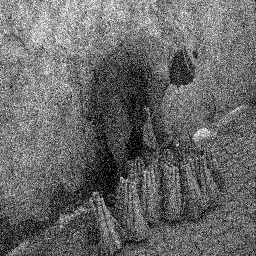}
~ \includegraphics[width=0.24\textwidth]{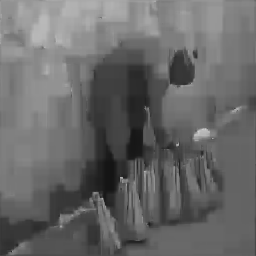}
~ \includegraphics[width=0.24\textwidth]{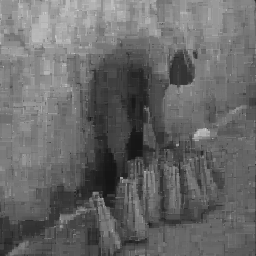}
\caption{Image 293029 from the BSDS data set. \textbf{Top row:} Left: image
with noise standard deviation $0.2\times\text{sd(image)}$. Center:
image denoised by BMA with cycle-spinning step=2 (MSE = $1.55\times10^{3}$).
Right: image denoised by   boosting with prune 0.2 and learning
rate 0.05 (MSE = $1.32\times10^{3}$). \textbf{Third row:} Left:
image with noise standard deviation $0.4\times\text{sd(image)}$.
Center: image denoised by BMA with cycle-spinning step=2 (MSE = $2.34\times10^{3}$).
Right: image denoised by   boosting with prune 0.2 and learning
rate 0.05 (MSE = $2.20\times10^{3}$). }
\label{fig:BSDS293029} 
\end{figure}

For the first noise mechanism, the additive noise is generated according
to 
\begin{align}
\sigma\in\{0.2,0.4,0.6\}\times\operatorname{sd}(I_{i}),
\end{align}
where $\operatorname{sd}(I_{i})$ is the standard deviation of the
$i$-th clean BSDS image $I_{i}$. This rule couples the magnitude
of the perturbation to image content: high-contrast or highly textured
images receive larger \emph{absolute} noise than smooth images. Consequently
the experiment becomes heteroskedastic across images even when the
nominal ``noise level'' is the same. %
The BMA denoiser, particularly with cycle spinning, adapts its shrinkage
to the noise standard deviation on a per-image basis and thereby behaves
like a variance-normalized estimator. When the noise standard deviation
$\sigma_{i}$ increases with $\operatorname{sd}(I_{i})$ (where $I_{i}$
is the $i$th clean BSDS image), the risk-optimal shrinkage on rough
images becomes more aggressive, which reduces the pixelwise MSE but
also suppresses fine textures. Models whose hyperparameters are tuned
globally across images (e.g., boosted ensembles with capacity fixed
per “noise level’’) cannot match this per-image adaptation and are
penalized precisely on the images where the absolute noise has been
inflated by design. The pattern in the top row of Figure~\ref{fig:BSDS-mse}---BMA
overtaking at the higher reported noise levels while qualitative panels
reveal oversmoothing---is consistent with this mechanism.

To eliminate the confounding, the injected noise must be decoupled
from image content. Hence for our second noise mechanism, we add i.i.d.\ Gaussian
noise with a fixed \emph{absolute} standard deviation $\sigma_{i}\in\{5,10,15,20,25,30\}/255$
after scaling all images to $[0,1]$. Hyperparameters are selected
once per $\sigma$ on a validation set of training images rather than
per test image. Model capacities are matched across methods by equating
effective degrees of freedom, for example by comparing a single tree
with $L$ leaves to a boosting run whose total leaves across trees
are approximately $L$.

We re-run the study under both the original noise rule and the corrected
protocol and examine the per-image difference 
\begin{align}
\Delta_{i}\;=\;\operatorname{MSE}_{\text{BMA}}(i)-\operatorname{MSE}_{\text{boost}}(i).
\end{align}
Under the original rule a simple regression $\Delta_{i}=\beta_{0}+\beta\,\operatorname{sd}(I_{i})+\epsilon_{i}$
should yield $\hat{\beta}<0$, indicating that BMA appears better
precisely when $\operatorname{sd}(I_{i})$ is large because those
images were assigned larger absolute noise. Under the corrected protocol
the slope should be statistically indistinguishable from zero, and
the apparent advantage of BMA on high-contrast images should diminish;
as capacity increases under matched budgets the ranking between single-tree
and ensemble estimators should stabilize or reverse for reasons intrinsic
to the estimators rather than to the noise-generation artifact.

\begin{figure}
\centering 
\includegraphics[width=0.85\textwidth]{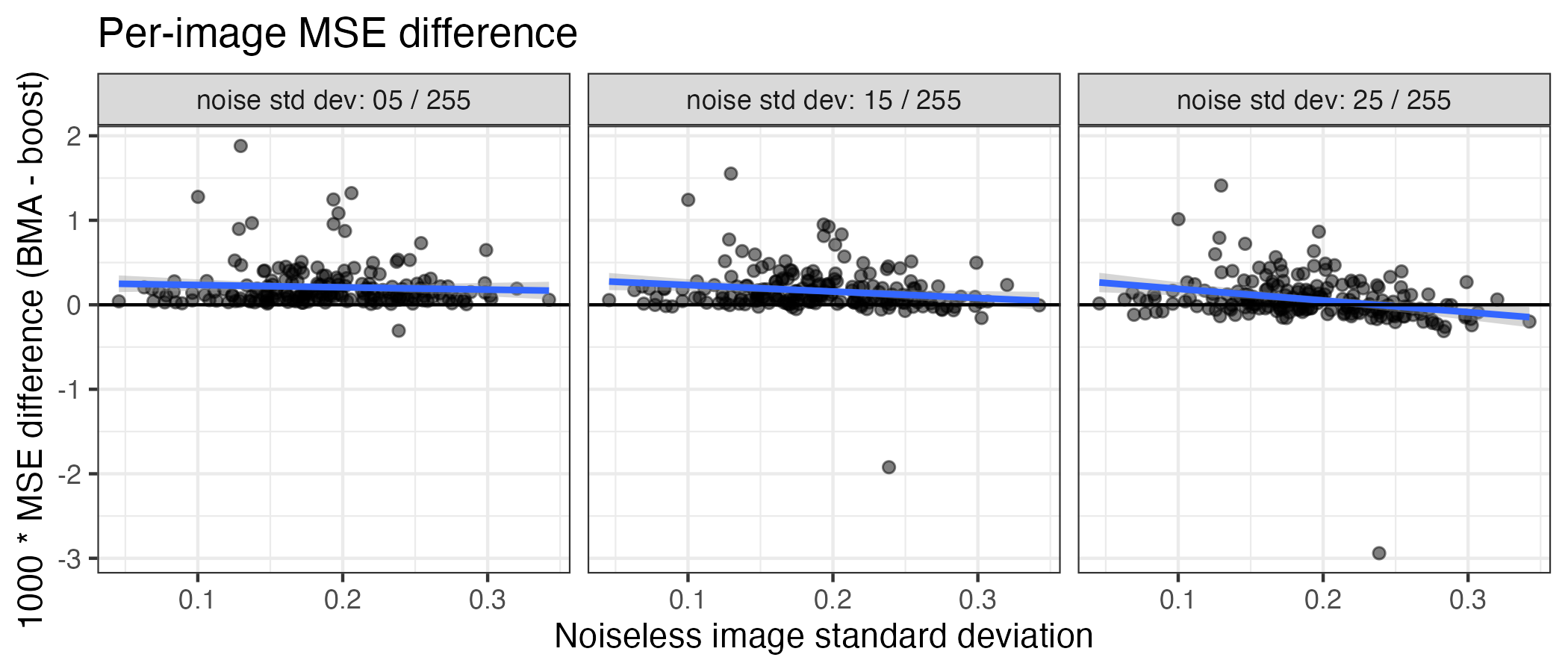} \caption{Per-image test MSE difference between BMA with one cycle-spinning
step and boosting with $b=0.2$. Each point corresponds to a BSDS image.}
\label{fig:BSDS-msediff_perimage} 
\end{figure}

\subsection{Numerical results for tensor}

\label{subsec:numerical-tensor}

We conduct an empirical study to evaluate the performance of unbalanced
Haar (UH) partitioning methods extended to three-dimensional tensor
domains. Here we consider a discrete, regular grid $\Omega\subset\mathbb{R}^{3}$
with dimensions $64\times64\times64$, representing a video sequence.
The clean video tensor signal $f\colon\Omega\to\mathbb{R}$ is corrupted
by additive i.i.d.\ Gaussian noise $\varepsilon\sim\mathcal{N}(0,\sigma^{2})$
where $\sigma=0.1\cdot\text{std}(f)$, yielding the observed tensor
$y=f+\varepsilon$. For tensor decision tree methods, we extract overlapping
patches of size $2\times2\times1$, generating $N=(62)\times(62)\times(64)=248832$
training samples.

The evaluation encompasses three reconstruction approaches. The first
follows the standard Gibbs-style backfitting, implementing single-step
proposals per component with $K=200$ additive terms 500 backfitting
iterations. The second utilizes a patch-based regression framework
using TensorDecisionTreeRegressor \citep{luo2024efficient} as weak
learners with maximum depth of 3 within a boosting ensemble
of $M=50$ estimators, where reconstruction proceeds by averaging
overlapping patch predictions. The third is a single TensorDecisionTreeRegressor
learner.

The quality of a reconstructed tensor $\hat{f}$ relative to the clean
signal $f$ is assessed via mean squared error: $\text{MSE}=\mathbb{E}[\|f-\hat{f}\|_{2}^{2}]$.
Figure~\ref{fig:Top:-noiseless-video.} illustrates the reconstruction
performance across temporal slices. The best performance is achieved
by the boosted ensemble of TensorDecisionTreeRegressors. This is perhaps
expected, seeing as these weak learners use tensor-specific technology.
In contrast, UHWTs are not tensor specific, and still the
backfitting ensemble of UHWTs achieves good performance
as well. Interestingly, the reconstruction by the boosted ensemble
looks to be a bit blurry, whereas the reconstruction by the backfitting
ensemble is crisper but does not denoise as well. Ensembling seems
necessary to achieve a satisfactory MSE, seeing as the single TensorDecisionTreeRegressor
learner produces an MSE that is at least five times either of those
by the ensemble models. 

Figure~\ref{fig:tensorUQ} demonstrates the uncertainty quantification
capabilities inherent in the Bayesian framework. The posterior mean
provides a denoised reconstruction, while the posterior standard deviation
$\sigma$ reveals regions of higher uncertainty, particularly along
boundaries and in areas with complex geometric features. The 95\%
credible interval width offers additional insight into reconstruction
confidence, with narrower intervals in homogeneous regions and wider
intervals in areas where the model exhibits greater uncertainty. This
uncertainty quantification proves valuable for assessing reconstruction
reliability and identifying regions requiring additional modeling
attention.

This evaluation confirms the viability of extending UH partitioning
to higher-dimensional tensor domains while highlighting the trade-offs
between single adaptive partitions versus additive ensemble strategies
for tensor reconstruction tasks.
\begin{figure}[H]
\centering

Clean video

\includegraphics[width=1\textwidth]{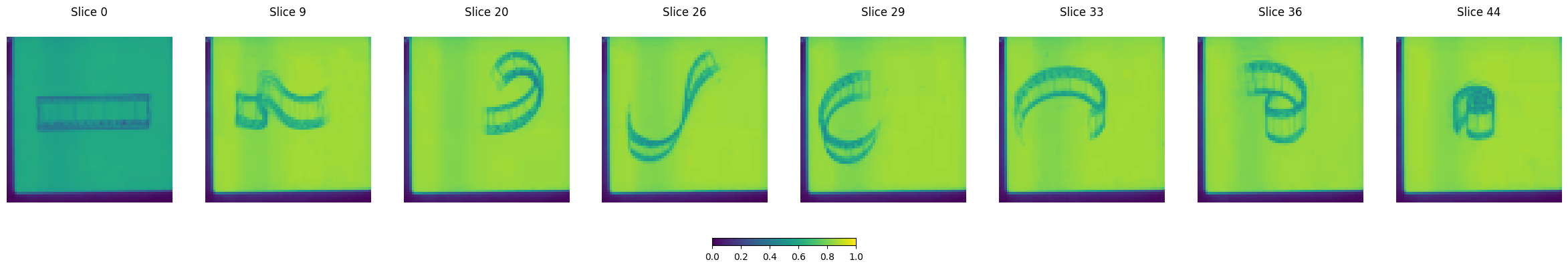}

Noisy video. Variance=$4.85\times10^{-2}$. 

\includegraphics[width=1\textwidth]{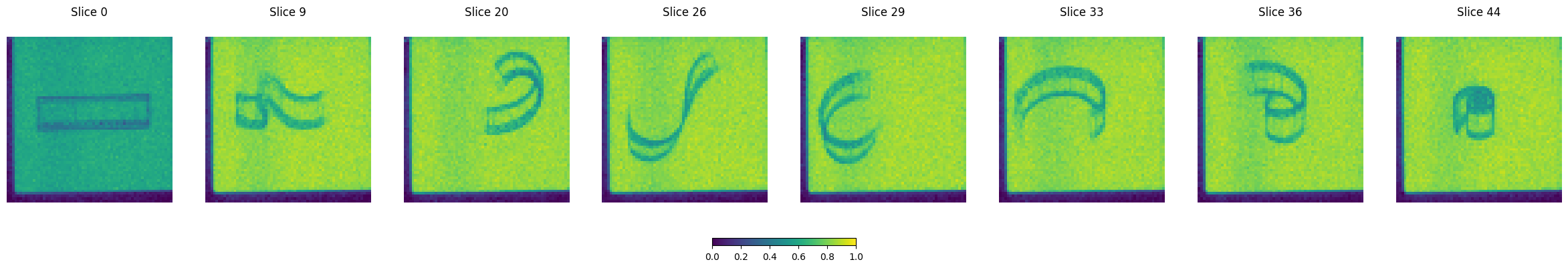}

backfitting ensemble of UHWTs. MSE=$2.44\times10^{-4}$.
Fit time = 8.5 min. 

\includegraphics[width=1\textwidth]{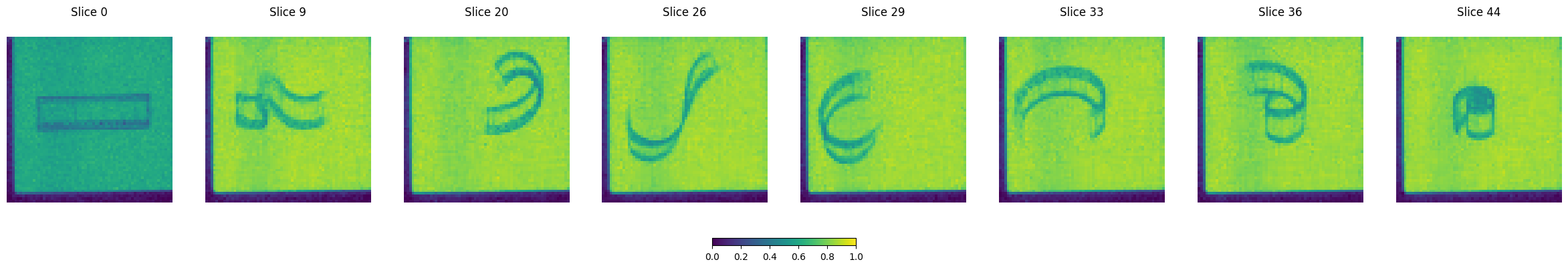}

boosted ensemble of TensorDecisionTreeRegressors \citep{luo2024efficient}. MSE=$1.33\times10^{-4}$.
Fit time = 26.5 min. 

\includegraphics[width=1\textwidth]{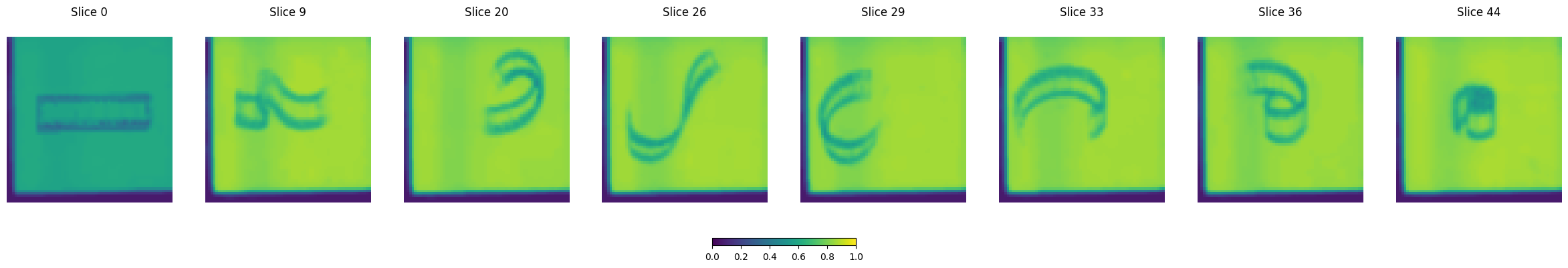}

single TensorDecisionTreeRegressor learner. MSE=$1.23\times10^{-3}$.
Fit time = 23.6 sec. 

\includegraphics[width=1\textwidth]{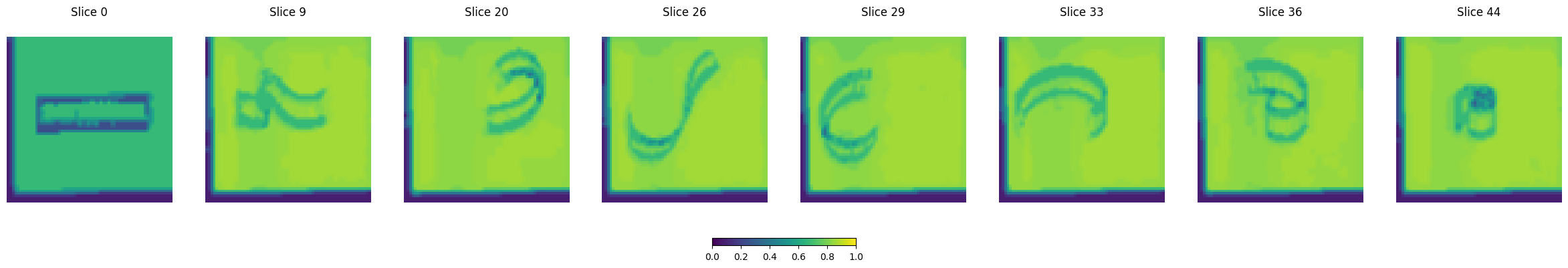}

\caption{\label{fig:Top:-noiseless-video.}Temporal slice comparison demonstrating
video denoising performance. Row 1: ground truth clean video frames
showing temporal evolution of the underlying signal structure. Row
2: noisy observations with additive Gaussian noise ($\sigma=0.1\cdot\text{std}(f)$).
Row 3: reconstructed frames using a backfitting ensemble of 200 weak
UH learners. Row 4: reconstructed frames using a boosted ensemble
of 50 weak TensorDecisionTreeRegressor learners. Row 5: a single TensorDecisionTreeRegressor
learner. }
\end{figure}

\begin{figure}[H]
\centering \includegraphics[width=1\textwidth]{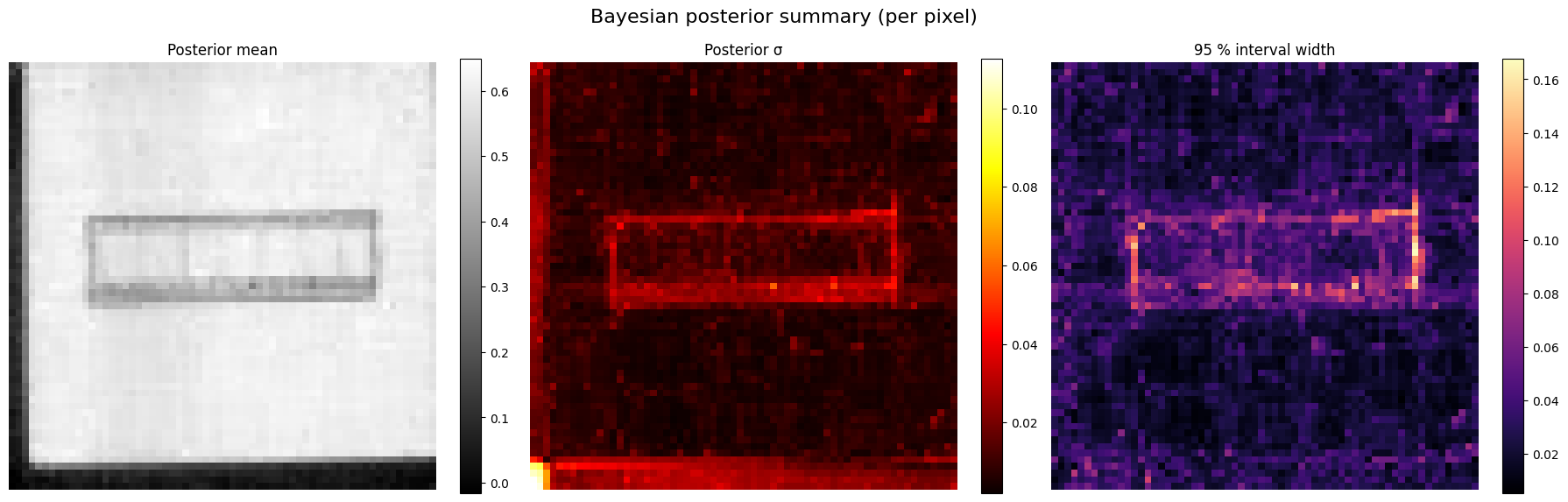}
\caption{\label{fig:tensorUQ}Bayesian uncertainty quantification for tensor
reconstruction (slice 0). Left: Posterior mean estimate providing
the denoised reconstruction. Center: Posterior standard deviation
revealing spatially-varying uncertainty, with higher values (red/yellow
regions) indicating areas where the model exhibits greater reconstruction
uncertainty, particularly along object boundaries and complex geometric
features. Right: Width of 95\% credible intervals, offering pixel-wise
confidence assessment with narrower intervals (dark blue) in homogeneous
regions and wider intervals (yellow) in structurally complex areas. }
\end{figure}

\subsection{Numerical results for sphere}

\label{subsec:Numerical-results-forspheres}

\subsubsection{Simulation experiment 1: bias-variance analysis}

\label{sec:sphere-bias-variance}

In a Random Forest framework, our fitted spherical wavelet trees are
essentially identical, and so the ensemble's variance will not be
reduced by model averaging (as seen by the roughly constant test MSE);
in this case, the random rotations provide a much-needed decorrelation
between learners as seen by the large difference in test MSE between
the rotated and non-rotated RF ensembles using our fitted spherical
wavelet trees, but the test MSE lower bound of $\approx0.4$ leaves
room for improvement. In contrast, a non-rotated Fréchet tree can
split away from the midpoint and hence already produces fairly uncorrelated
learners; in this case, the ensemble's variance is sharply reduced
by model averaging, but the test MSE fails to fall below $0.2$, which
indicates the ensemble's bias that the averaging process cannot reduce.
The random rotations, which are independent of each other and of the
data, slightly worsen the test MSE of the non-rotated Fréchet Forest.
If the learners are truly uncorrelated, then the ensemble's variance
shrinks to zero as the number of learners increases, regardless of
the base learner's variance. Hence, the increased test MSE must come
from an increase in the ensemble's bias, which itself equals each
base learner's bias. Thus, we conclude that each random rotation increases
both the bias and variance of the base learner.

In a boosting framework, a non-rotated Fréchet tree can split away
from the midpoint and hence already allows diversity between learners;
in this case, the ensemble's test MSE sharply decreases to $\approx0.18$
but plateaus after roughly 100 learners. However, the random rotations
slightly worsen the test MSE of the non-rotated Fréchet boosted ensemble.
Seeing as boosting is meant to reduce bias in the ensemble, the slightly
larger test MSE seems to result from the increase in each base learner's
variance from the random rotations. In contrast, our fitted spherical
wavelet trees intuitively do not produce much diversity between learners
due to splitting only at midpoints and hence would strongly benefit
from the inter-learner diversity induced by random rotations, even
at the cost of increased bias and variance in each base learner. This
is reflected empirically in the substantial reduction in test MSE
when random rotations are incorporated in a boosted ensemble using
our fitted spherical wavelet tree as the base learner. Perhaps surprisingly,
this ensemble achieves the smallest test MSE ($\approx0.15$) among
all tested methods. (This includes the test MSEs for three residual
deep GP models shown in Table~MT-2, where we see that
the smallest mean MSE is $0.182$, and for the spherical GP model,
whose test MSE was $0.477$.) Furthermore, this ensemble seemingly
can still improve if the number of learners were to increase beyond
500. In this scenario, for either base learner, boosting produced
smaller test MSEs than did model averaging, which suggests that each
learner's bias dominates its variance. In addition, a single spherical
wavelet tree using \texttt{adapt} has a smaller variance and larger
bias than a single Fréchet tree has, which makes it a more promising
base learner in a boosting framework. As a bonus, it can be applied
to generic topological triangulations and hence can be used for a
larger class of input domains than can a Fréchet tree, which requires
inputs to lie on a metric space.

\subsubsection{Simulation experiment 2: comparing splitting methods}

\label{sec:sphere-sim-4}

\begin{figure}[t]
\centering \includegraphics[width=0.85\textwidth]{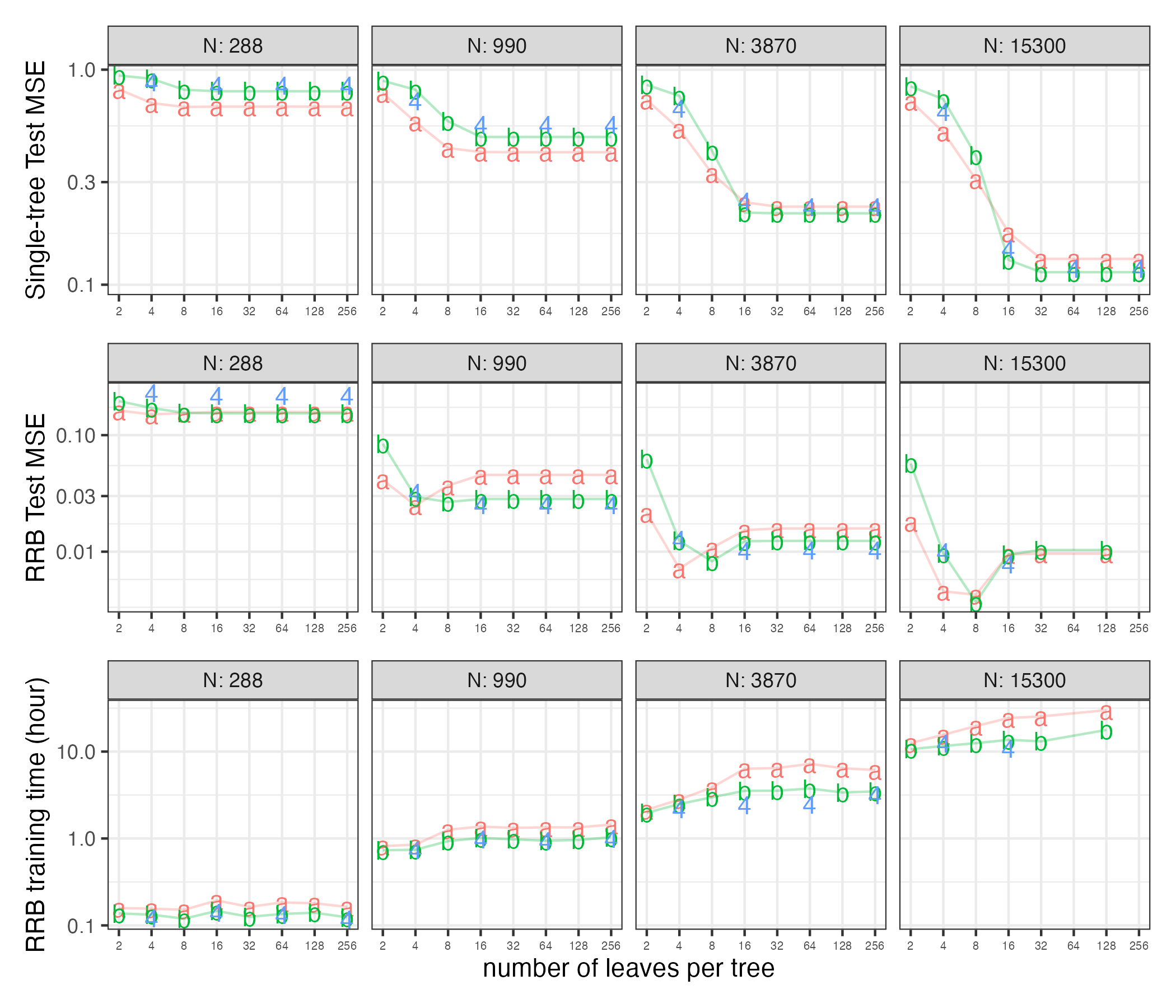}
\caption{Empirical comparison of the impact of three splitting methods --
\texttt{adapt} (`a'), \texttt{balance} (`b'), and \texttt{balance4}
(`4') -- on single-tree and random-rotation boosted models trained
on $N$ points $(\bm{x},y)$, where each location $\bm{x}=(x_{1},x_{2},x_{3})$
was drawn uniformly on the unit sphere $\mathbb{S}_{2}$, and the
observed response $y$ at $\bm{x}$ is the signal described in Figure~MT-5
plus i.i.d.\ Gaussian noise with zero mean and standard deviation
equal to $0.1$ of the standard deviation of $f$. Boosted methods
use learning rate $0.05$.}
\label{fig:sphere-mse-4} 
\end{figure}

Here we compare the impact of three splitting methods -- \texttt{adapt}, \texttt{balance},
and \texttt{balance4} -- on the empirical performance of a single-tree model
per level-0 triangle and on a random-rotation boosted (RRB) ensemble
of 500 trees per level-0 triangle. (The fourth -- \texttt{adaptvertex} --
takes too much training time.) Each comparison is for a given number
of leaves per tree and a given test function from which training and
test data, each of size $N\in\{288,990,3870,15300\}$, are generated.

First we will compare the performance of the adapt and balance splitting
methods. Figure~\ref{fig:sphere-mse-4} shows that the smallest test
MSE is usually achieved by the adapt method, and it is achieved with
4 leaves per tree. In contrast, the balance method with 4 leaves per
tree typically produces a larger test MSE, but once the number of
leaves per trees is with least 8, the resulting test MSEs are comparable
to the smallest test MSE achieved using the adapt method. Although
training a tree with 8 leaves using the balance method is typically
still a bit faster than training a tree with 4 leaves using the adapt
method, we must also consider the storage cost: a tree with 8 leaves
contains twice as much information as a tree with 4 leaves and hence
requires twice as much memory to store. This is also an important
consideration for speeding up any post-hoc operations on the tree,
such as signal reconstruction.

Now we will compare against the balance4 splitting method. For each
of the latter two methods (for the $N\ge990$ cases), we see that
the respective test MSE curve exhibits a U-shape due to the bias-variance
trade-off: as tree depth increases, the ensemble increasingly adapts
to the data but also increasingly risks overfitting. Due to the steepness
of the dips, especially for larger $N$, it is important to choose
the number of leaves per tree that minimizes the test MSE. However,
this task is more difficult for the balance4 method because there
is less granularity in the choice of the number of leaves per tree;
the number of leaves must be a power of 4 for the balance4 method,
whereas it must only be a power of 2 for the adapt and balance methods.
In particular, the difference between 4 and 16 leaves per tree can
have a significant impact on the test MSE, training time, and storage
requirements. Hence, we argue that methods that split nodes into 2
children are favorable over methods that split nodes into 4 children.

\subsubsection{Simulation experiment 3: misaligned spherical features}

\label{sec:sphere-sim-2}

Here, we consider a more complex benchmark function $g$ for spherical
functions. We want to validate the wavelet system on a spherical triangulation,
$g$ is a much better “structured yet controlled” test signal: intrinsic,
oriented and slightly asymmetric (tilt) so we can spot real directional
behavior rather than coordinate artifacts.

\begin{align}
g(\bm{x})=\sum_{i=1}^{3}\!\left[\exp\!\Big(-\frac{(\mathbf{n}_{i}^{\top}\bm{x})^{2}}{\sigma_{n}^{2}}\Big)+\beta\,\exp\!\Big(-\frac{(\mathbf{n}_{i}^{\top}\bm{x})^{2}}{\sigma_{w}^{2}}\Big)\right],
\end{align}
where the unit normals $\{\mathbf{n}_{i}\}$ define three planes (one
tilted) whose intersections with $\mathbb{S}^{2}$ are great circles.

Both functions are useful. The baseline $f$ defined in Figure~MT-5
is simple, smooth, and inexpensive to evaluate, so it is well suited
for pipeline checks and calibration. However, for fitting wavelets
on a triangulated sphere, $g$ is preferable for six reasons. First,
$g$ is intrinsic: since $\mathbf{n}_{i}^{\top}\bm{x}=\cos\theta$,
it organizes energy by geodesic distance, whereas $f$ depends on
extrinsic coordinates. Second, $g$ supplies multiple oriented features,
enabling stringent tests of directional selectivity absent from $f$'s
single dominant oscillation. Third, $g$ provides localized multiscale
content via narrow bands and diffused halos, unlike $f$'s broadly
distributed trends. Fourth, $g$ exposes mesh anisotropy because its
features are uniform along geodesics, while $f$ may align with grid
directions. Fifth, $g$ yields a balanced dynamic range with salient
circle intersections that stabilize error metrics; $f$'s linear and
saturating parts can overwhelm $\cos(10x_{2})$. Sixth, $g$ contains
curvilinear singular structures that are harder than the near-band-limited
behavior of $f$, thereby offering a more discriminative benchmark
for spherical wavelets.

As mentioned in the main text, the prediction and MSE behaviors (shown
in Figure~\ref{fig:sphere-mse-2}) are similar to those for the signal
described in Figure~MT-5. Interestingly, here the
residual deep GP models (MSEs in Table~\ref{tab:deepGP-2}) fare
much better. In fact, the test MSE ($0.0374\pm0.0030)$ of the best
performing GP model (1 layer with inducing points, so it is not a
deep GP) is smaller than all but the random-rotation boosted ensemble's
($0.0330$). %
We also note that it appears that increasing the number of ensembles
will further decrease the test MSE for the random-rotation boosted
ensemble, whereas there is no clear path for improvement for the GP
model with inducing points.

\afterpage{
\begin{figure}[ht!]
\centering \includegraphics[width=0.95\textwidth]{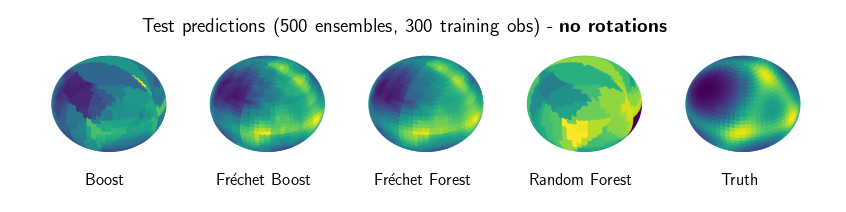}
\includegraphics[width=0.95\textwidth]{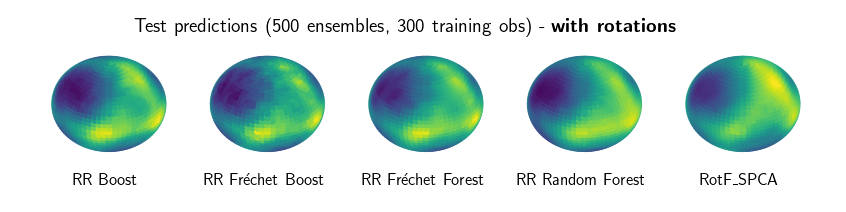}
\includegraphics[width=1\textwidth]{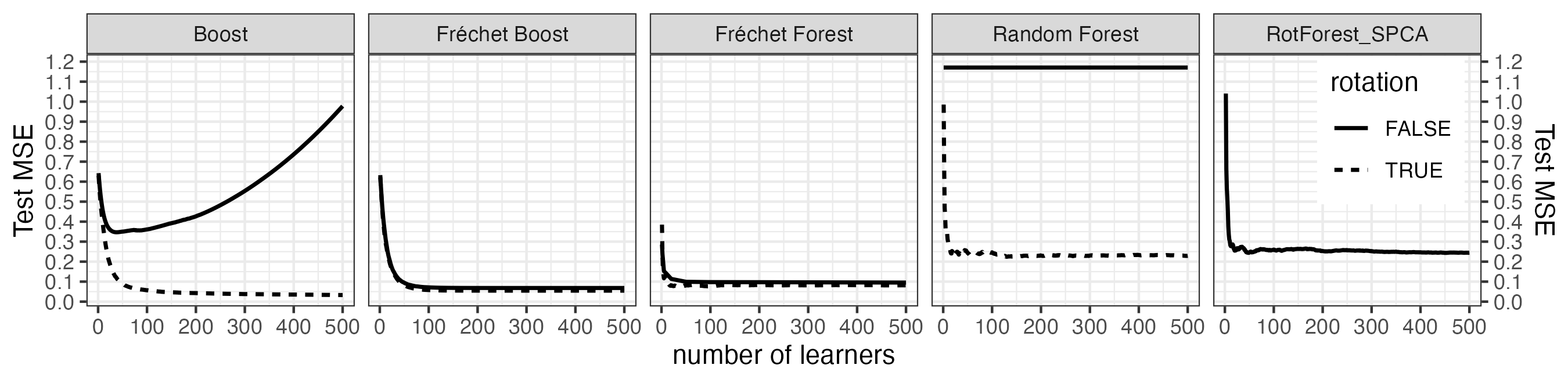}
\begin{tabular}{rlll}
\toprule 
 & 500 ensembles & 1000 ensembles & 1500 ensembles \tabularnewline
\midrule 
Test MSE for RR Boost & $29.0\times10^{-3}$ & $27.6\times10^{-3}$ & $27.0\times10^{-3}$ \tabularnewline 
\bottomrule
\end{tabular}
\caption{Empirical comparison of various tree ensemble models (\textquotedblleft RR\textquotedblright{}
refers to the method version that randomly rotates covariates before
fitting any base learner) trained on $n=300$ points $(\bm{x},y)$, where
each location $\bm{x}=(x_{1},x_{2},x_{3})$ was drawn uniformly on the
unit sphere $\mathbb{S}_{2}$, and the observed response $y$ at $\bm{x}$
is the signal $f(\bm{x})=\exp\{-16x_{3}^{2}\}+1.5\exp\{-16x_{3}^{2}/9\}+\exp\{-16x_{2}^{2}\}+1.5\exp\{-16x_{2}^{2}/9\}+\exp\{-16n(\bm{x})\}+1.5\exp\{-16n(\bm{x})/9\}$,
where $n(\bm{x})=[x_{1}\cos(35^{\circ})-x_{2}\sin(35^{\circ})+x_{3}\sin(35^{\circ})]^{2}$
plus i.i.d.\ Gaussian noise with zero mean and standard deviation
equal to $0.1$ of the standard deviation of $f$. The true signal
is shown in the right-most sphere of the top row. All tree-based methods
allow trees to grow to full depth. Boosted methods use learning rate
$0.05$. Top two rows: predictions of the ensemble methods. Third
row: MSE on 2448 test points.}
\label{fig:sphere-mse-2} 
\end{figure}

\begin{table}[h!]
\centering
\resizebox{\linewidth}{!}{
\begin{tabular}{rllll}
\toprule 
Test MSE for residual deep GP & 1 layer & 2 layers & 3 layers & 4 layers \tabularnewline
\midrule 
+hodge+spherical\_harmonic\_features  & \textbf{52.7 (0)}  & 53.1 (0.1)  & 53.1 (0.1)  & 53.2 (0.1) \tabularnewline
+inducing\_points  & \textbf{37.4 (3.0)}  & 53.1 (16.5)  & 52.9 (16.2)  & 52.6 (15.2) \tabularnewline
+spherical\_harmonic\_features  & \textbf{52.7 (0)}  & 53.6 (0.1)  & 53.8 (0.1)  & 53.8 (0.1) \tabularnewline
\bottomrule
\end{tabular}
}
\caption{Mean and standard deviation of test MSE ($\times10^{-3}$) for three
residual deep GP methods \citep{wyrwal2025residual} over 10 starting
seeds. Data were generated as described in Figure~\ref{fig:sphere-mse-2}.
Smallest test MSE for each method is bolded.}
\label{tab:deepGP-2} 
\end{table}

}

\subsubsection{Simulation experiment 4: irregular benchmark function}

\label{sec:sphere-sim-3}

Here, we consider an ``irregular'' benchmark function for spherical
functions as described in Figure~4a of \citet{wyrwal2025residual}.
The prediction and MSE behaviors, shown in Figure~\ref{fig:sphere-mse-3},
are similar to those for the signal described in Figure~\ref{fig:sphere-mse-2}.
Again, here the residual deep GP models produce competitive test MSEs;
the test MSE ($1.657\pm0.027)\times10^{-3}$ of the best performing
GP model is smaller than all but the random-rotation boosted ensemble's
($1.514\times10^{-3}$ with 1500 learners, not shown). We also note
that even after 1500 learners, the random-rotation boosted ensemble's
test MSE looks like it will further decrease as the number of learners
increases. At 1500 learners, this test MSE is smaller than any of
the deep GP MSEs shown in Table~\ref{tab:deepGP-3}.

\afterpage{
\begin{figure}[ht!]
\centering \includegraphics[width=0.95\textwidth]{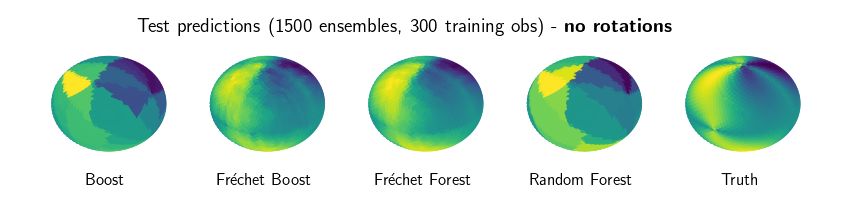}
\includegraphics[width=0.95\textwidth]{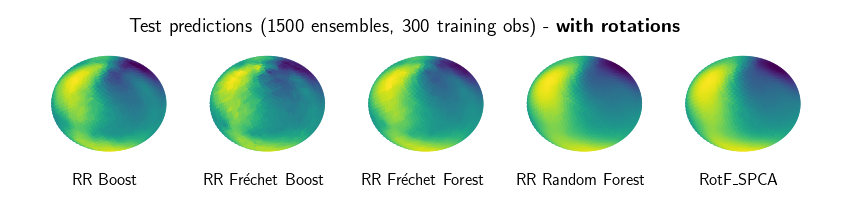}
\includegraphics[width=1\textwidth]{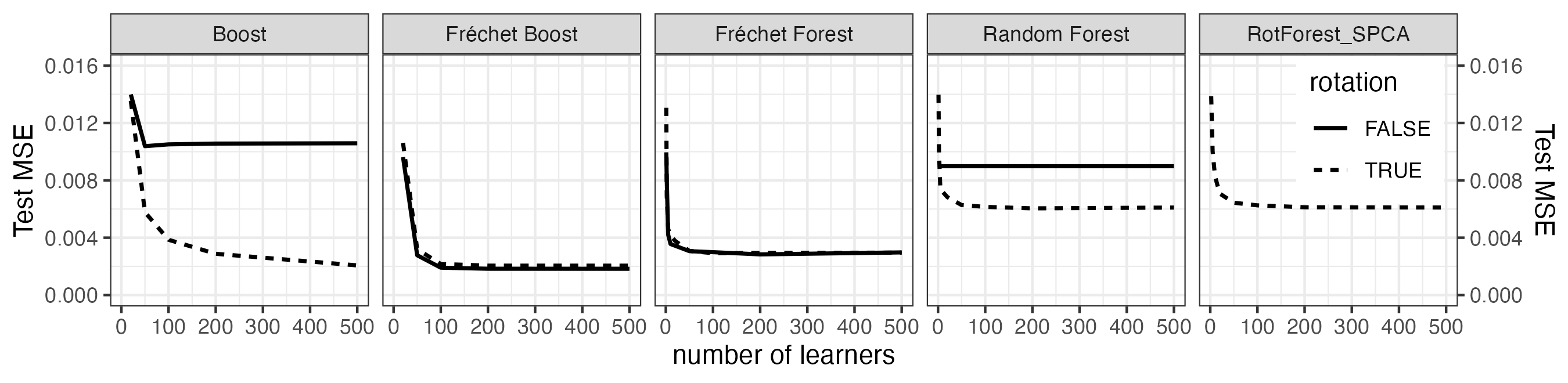}
\begin{tabular}{rlll}
\toprule 
 & 500 ensembles & 1000 ensembles & 1500 ensembles \tabularnewline
\midrule 
Test MSE for RR Boost  & $2.0676\times10^{-3}$  & $1.669\times10^{-3}$  & $\mathbf{1.514\times10^{-3}}$ \tabularnewline 
\bottomrule
\end{tabular}
\caption{Empirical comparison of various tree ensemble models (\textquotedblleft RR\textquotedblright{}
refers to the method version that randomly rotates covariates before
fitting any base learner) trained on $n=300$ points $(\bm{x},y)$, where
each location $\bm{x}=(x_{1},x_{2},x_{3})$ was drawn uniformly on the
unit sphere $\mathbb{S}_{2}$, and the observed response $y$ at $\bm{x}$
is the signal described in Figure~4a of \citet{wyrwal2025residual}
plus i.i.d.\ Gaussian noise with zero mean and standard deviation
equal to $0.1$ of the standard deviation of $f$. The true signal
is shown in the right-most sphere of the top row. All tree-based methods
allow trees to grow to full depth. Boosted methods use learning rate
$0.05$. Top two rows: predictions of the ensemble methods. Third
row: MSE on 2448 test points.}
\label{fig:sphere-mse-3} 
\end{figure}
\begin{table}[h!]
\centering %
\begin{tabular}{rllll}
\toprule 
{\tiny Test MSE for residual deep GP} & 1 layer & 2 layers & 3 layers & 4 layers \tabularnewline
\midrule 
{\tiny +hodge+spherical\_harmonic\_features}  & 3.078 (0)  & 2.555 (0.582)  & 2.511 (0.473)  & \textbf{2.424 (0.579)} \tabularnewline
{\tiny +inducing\_points}  & 3.078 (0)  & 1.716 (0.075)  & \textbf{1.662 (0.025)}  & \textbf{1.663 (0.027)} \tabularnewline
{\tiny +spherical\_harmonic\_features}  & 3.078 (0)  & 1.769 (0.205)  & \textbf{1.657 (0.027)}  & 1.677 (0.063) \tabularnewline
\bottomrule
\end{tabular}
\vskip 0.1cm
\caption{Mean and standard deviation of test MSE ($\times10^{-3}$) for three
residual deep GP methods \citep{wyrwal2025residual} over 10 starting
seeds. Data were generated as described in Figure~\ref{fig:sphere-mse-3}. Smallest test MSE for each method is bolded.}
\label{tab:deepGP-3} 
\end{table}
}

\subsubsection{Uncertainty quantification on a spherical domain}

Here we apply a frequentist approach to uncertainty quantification (UQ) on the GISS data from Section~MT-6.3.
In Figure~\ref{fig:triangulation_noaa_uq}, we also provide predictive UQ by adapting
the approach in quantile regression forests \citep{meinshausen2006quantile}
to our random-forest ensemble of $B=500$ spherical learners using \texttt{adapt}.
This UQ method can be incorporated for Random Forest or Boosting (with or without random rotations), 
but because the purpose is not necessarily to compare the UQ performance of various ensembles, 
we illustrate this UQ method for only one ensemble.
For a query
$\bm{x}$, a quantile regression forest provides conditional quantiles
by assigning to each training index ($i=1,\dots,n$) the weight 
\[
w_{i}(\bm{x})=\frac{1}{B}\sum_{b=1}^{B}\sum_{A\in\mathcal{L}(T_{b})}\frac{\mathbf{1}\{\bm{x}_{i}\in A\}\mathbf{1}\{\bm{x}\in A\}}{|\{j\colon \bm{x}_{j}\in A\}|},
\]
where $\mathcal{L}(T)$ denotes the set of leaf nodes in tree $T$,
and then inverting the weighted
empirical cumulative distribution function $\widehat{F}_{Y\mid X=x}(t)=\sum_{i=1}^{n}w_{i}(x)\mathbf{1}\{y_{i}\le t\}$.

\begin{figure}[h]
\centering \includegraphics[width=1\textwidth]{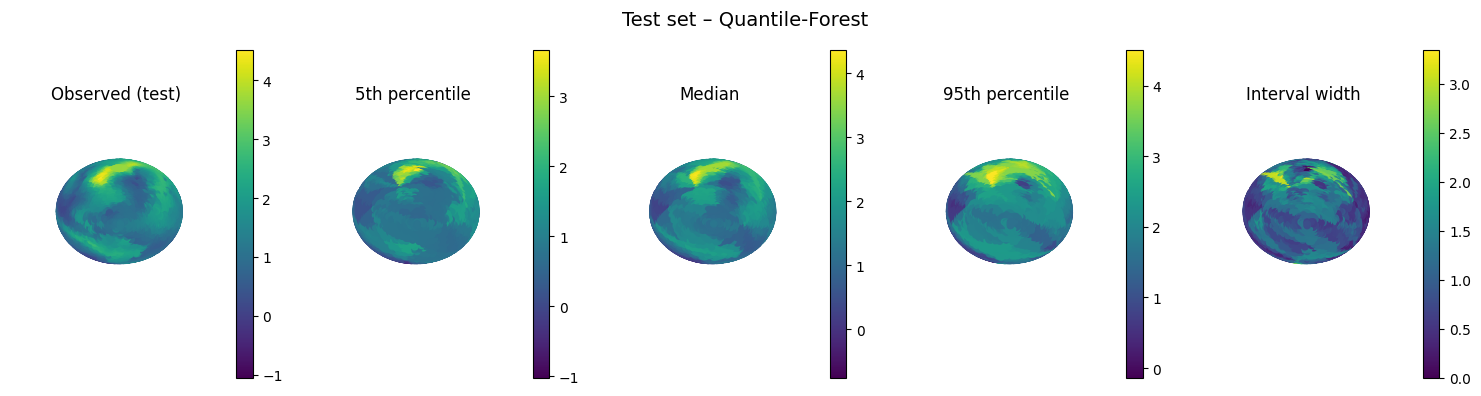}
\caption{Quantile predictions of a Random Forest of 500 spherical learners per initial triangle
using \texttt{adapt} trained data described in Figure~MT-6.
This ensemble achieved an empirical $90\%$ coverage of $87.7\%$.}
\label{fig:triangulation_noaa_uq}
\end{figure}

\end{document}